\documentclass{article}
\usepackage{graphicx} 
\usepackage[applemac]{inputenc}
\usepackage{amsfonts}
\usepackage[english]{babel}
\usepackage{authblk}
\usepackage{color, colortbl}
\usepackage{amsmath}
\usepackage{amssymb}
\usepackage{bbm}
\usepackage{tikz}
\usepackage{caption}
\usepackage{subcaption}
\usepackage{geometry}
\usetikzlibrary{arrows}
\usepackage{mathtools}
\usepackage{comment}
\usepackage{amsthm}
\usepackage{enumitem}
\usepackage{comment}
\usepackage{booktabs}

\newtheorem{theorem}{Theorem}[section]

\newtheorem{proposition}[theorem]{Proposition}
\newtheorem{corollary}[theorem]{Corollary}
\newtheorem{definition}[theorem]{Definition}

\newtheorem{observation}[theorem]{Observation}
\newtheorem{remark}[theorem]{Remark}

\title{A Mean Field Game approach for pollution regulation of competitive firms}

\author[1,2]{Gianmarco Del Sarto\footnote{gianmarco.delsarto$@$sns.it,gianmarco.delsarto@iusspavia.it}}
\affil[1]{Scuola Normale Superiore, Pisa, Italy }
\affil[2]{University School for Advanced Studies IUSS, Pavia, Italy}

\author[3]{Marta~Leocata\footnote{mleocata@luiss.it (primary author, alphabetic order)}}
\affil[3]{LUISS Guido Carli, Rome, Italy}
\author[4]{Giulia Livieri \thanks{g.livieri$@$lse.ac.uk (primary author, alphabetic order)\\ \\We are grateful to Prof. Ren\'{e} A\"id, Prof. Fausto Gozzi, Prof. Jameson Graber, Prof. Fujii, Masaaki, and seminar participants at the University of Padova, LUISS Guido Carli, University of Verona, Durham University, the \textit{Mathematics of subjective probability, Topics in Economics and Probability} (Milan, September 11-13, 2023), the \textit{DCP24 Dynamics} $\&$ \textit{Complexity} (Pisa, June 5-8, 2024), the $4^{th}$ \textit{Italian Meeting on Probability and Mathematical Statistics (Rome, June 10-14, 2024)} for discussions. Giulia Livieri acknowledges partial support by the grant \textit{Qnt4Green- Quantitative Approaches for Green Bond Market: Risk Assessment, Agency Problems, and Policy Incentives} of the Italian Ministry of Education and Research. Gianmarco Del Sarto acknowledges the support of the Italian national inter-university Ph.D. course in sustainable development and climate change. All errors are our own.}} 
\affil[4]{The London School of Economics and Political Science, London, United Kingdom}

\date{\today}

\begin{document}

\maketitle

\begin{abstract}
We develop a model based on mean-field games of competitive firms producing similar goods according to a standard AK model with a depreciation rate of capital generating pollution as a byproduct. Our analysis focuses on the widely-used cap-and-trade pollution regulation. Under this regulation, firms have the flexibility to respond by implementing pollution abatement, reducing output, and participating in emission trading, while a regulator dynamically allocates emission allowances to each firm. The resulting mean-field game is of linear quadratic type and equivalent to a mean-field type control problem, i.e., it is a potential game. We find explicit solutions to this problem through the solutions to differential equations of Riccati type. Further, we investigate the carbon emission equilibrium price that satisfies the market clearing condition and find a specific form of FBSDE of McKean-Vlasov type with common noise. The solution to this equation provides an approximate equilibrium price. Additionally, we demonstrate that the degree of competition is vital in determining the economic consequences of pollution regulation. 
\end{abstract}

\section{Introduction}\label{sec::introduction}
\noindent   The problem of excessive firm pollution has long been a part of economic theory, mainly because it imposes a negative externality on society. In particular, it is considered as the consequence of the absence of price on emission, which implies higher volumes than socially optimal levels. Therefore, from an economic point of view, one possibility is to put a price on pollution; in this way, polluters will be more conscious about the social value of their private decisions. One of the most popular measures that help tackle this problem is the emission trading system, also known as the cap-and-trade system, which gives the environmental authority direct control on the overall quantity of emissions and, at the same time, increases the acceptability of environmental policy for covered companies because they can make profit from it. The EU-ETS (European Union Emission Trading Scheme) is, together with the US Sulfur Dioxide Trading System, the most prominent example of an existing cap-and-trade system deployed in practice (e.g., \cite{carmona2010marke}). Having made this premise, understanding how the market price of carbon in an emission trading system is formed through the interaction among a large number of (indistinguishable rational competitive) firms is significant. This paper proposes an integrated production-pollution-abatement model in continuous time and studies cap-and-trade under competition via the Mean Field Game (MFG, henceforth) approach. The theorethical model is described in detail in Section \ref{sec::theoreticalmodel}. In particular, we are interested in equilibrium carbon price formation in a cap-and-trade system, i.e., the pricing of carbon endogenously using a model of (indistinguishable rational competitive) firms under the market clearing condition.\\
\indent It is important to make the following point. In the present work, we consider two types of competitions.    On the one hand, the competition in polluting firms is because we do not focus on perfect competition or monopoly. However, instead, we account for firms' strategic interactions in the output markets by assuming that firms compete \textit{\'{a} la Cournot}.   In other words, competing firms are trapped in an equilibrium where each firm's decisions impose not just a pollution externality on society but also a competitive externality on the other firms. On the other hand, there is the type of competition in the continuum limit of an infinity of small players allowed by the MFG framework. Precisely, each player only sees and reacts to the statistical distribution of the other players' states; in turn, their actions determine the evolution of the state distribution. To avoid confusion, we always refer to the former when speaking about competition.\\
\indent MFG models appeared simultaneously and independently in the original works of \cite{lasry2007mean} and \cite{caines2006large}, and are, loosely speaking, limits of symmetric stochastic differential games with a large number of players where each of them interacts with the average behavior of his/her competitors. In particular, an MFG is an equilibrium, called $\epsilon$-Nash equilibrium, that occurs when the strategy employed by a representative agent of a given population is optimal, given the costs imposed by that population. An increasing stream of research has been flourishing since 2007, producing theoretical results and a wide range of applications in many fields, such as economics, finance, crowd dynamics, social sciences in general, and, only recently, in equilibrium price formation (e.g., \cite{fujii2022mean} and references therein). We refer to the lecture notes of \cite{cardaliaguet2012notes} and the two-volume monograph by Carmona and Delarue (\cite{delarue2018probabilisticone} and \cite{delarue2018probabilistictwo}) for an excellent presentation of the MFG theory from analytic and probabilistic perspective, respectively; in the present paper, we embrace a probabilistic perspective. A related but distinct concept is that of Mean Field Control (MFC), where the goal is to assign a strategy to all the agents at once so that the resulting population behavior is optimal concerning the cost imposed by a central planner. We refer to the excellent book by \cite{bensoussan2013mean} for a comparison between MFGs and MFC.\\
\indent In general, an optimal control for an MFC is not an equilibrium strategy for an MFG. Nevertheless, in many cases, the converse holds, and quantifying the differences between the two approaches is reminiscent of what is known as the price of anarchy, i.e., the added aggregate cost of allowing all players to choose their optimal strategy independently. The MFGs for which this happens are called potential MFGs (e.g., \cite{cecchin2022weak}). The model that we propose, while conceptually constructed as a MFG equilibrium, can be solved via a reformulation of MFC by using the results in \cite{graber2016linear}.   In particular, our model belongs to the class of linear-quadratic MFGs (e.g., \cite{bensoussan2016linear}) with common noise (e.g., \cite{pham2017dynamic} and \cite{pham2016linear}).   The common noise represents an inherent uncertainty in nature affecting simultaneously all the firms participating in the game (or being
controlled by a central planner).   We characterize the solution both in terms of a stochastic maximum principle (forward-backward system of stochastic differential equations (FBSDEs)) and Riccati equations.   In particular, similarly to \cite{fujii2022mean}, though their work is inspired by financial applications, when imposing the market clearing condition, we obtain an interesting form of FBSDEs of McKean-Vlasov type with common noise as a limit problem, involving the dependence on a conditional expectation. Therefore, the existence of a unique strong solution is proved by using the well-known Peng-Wu's continuation method \cite{peng1999fully}. In addition, we quantify the relation between the finite player game and its large population limit, as well as we show that the solution of the mean-field limit problem actually provides asymptotic market clearing in the large limit. Instead, if the carbon price process is given exogenously, the MFG solutions serve as $\epsilon$-Nash equilibria for the large player game because the game is solved by an optimal control problem \cite{graber2016linear}.\\
\indent The last part of the paper presents a numerical study of the proposed model, which is divided into two parts. In the first part, we analyse the role played by the environmental authority on the average level of production of a (representative firm).   In the second part, instead, we analyse the economics of competition. In particular, the representative firm faces a strategic trade-off between output reduction and pollution abatement under competition. The latter facilitates synchronization between the representative firm and the rest of the population in the sense that they agree to reduce output by using the pollution constraint; naturally, this synchronization mechanism is expected to work under a suitable range of constraints imposed by the pollution regulator, the one for which the impact of output reduction on the representative firm's profits dominates the cost of pollution abatement, of trading, and production.  Under monopoly, instead, the representative firm can no longer leverage on the competition with the population of firms to implement the previously described synchronization mechanism. Whence, the degree of competition plays a critical role in determining the economic consequences of pollution regulation.  In particular, our model captures a rich range of competitive markets -- with monopoly and Cournot oligopoly as special cases -- and several fundamental elements of pollution generation, abatement levels and costs, and regulation, which can serve as a basis for future research. \\
\indent We proceed as follows. Notation and basic objects are introduced in Section \ref{sec::notations}. In Section \ref{sec::theoreticalmodel}, we provide a precise description of the $N$-player game, where $N$ denotes the number of firms, together with the definition of $\epsilon$-Nash equilibria. In Section \ref{sec::mfgapproximation}, the limit dynamics for the $N$-player game is introduced. The corresponding notion of solution of the MFG is defined and discussed.  In Section \ref{sec::mfcapproximation}, the MFC problem associated with the MFG in Section \ref{sec::mfgapproximation} is introduced and discussed; we prove the solvability of the FBSDE of McKean-Vlasov type and the asymptotic market clearing condition. Section \ref{sec::numerical_illustrations} provides numerical results. Finally, in Section \ref{sec::futureresearch}, we give concluding remarks, discuss further extensions of the model and future directions of research. Additional results on linear-quadratic MFG and MFC are confined in Appendix \ref{app::LQMFG} and Appendix \ref{app::LQMFC}, respectively.
\section{Notations}
\label{sec::notations}
Because we are going to derive some broad-gauged results in Appendix \ref{app::LQMFG} and Appendix \ref{app::LQMFC}, the notation in this section will be quite general. Let $d, d_0, d_1, d_2 \in \mathbb{N}$, where $\mathbb{N}$ is the set of positive integers, which will be the dimensions of the space of private states, common noise values, idiosyncratic noise values and control actions, respectively. The $n$-dimensional Euclidean space $\mathbb{R}^n$, with $n \in \mathbb{N}$ a generic index, is equipped with the standard Euclidean norm, always indicated by $|\cdot|$. Moreover, we denote by $\mathcal{S}^{n}$ the set of all $n \times n$ symmetric matrices with real entries. In general, we identify the space of all $n \times m$ dimensional matrices with real entries with $\mathbb{R}^{n \times m}$.\\
\indent Let $N \in \mathbb{N}$. Let $(\overline{\Omega}^{0}, \overline{\mathcal{F}}^{0}, \overline{\mathbb{P}}^{0})$ and $(\overline{\Omega}^{i}, \overline{\mathcal{F}}^{i}, \overline{\mathbb{P}}^{i})_{i=1}^N$ be $(N+1)$ complete probability spaces equipped with filtrations $(\overline{\mathcal{F}}^{i}_t)$, $i \in \{0,\ldots,N\}$. In particular, $(\overline{\mathcal{F}}^{0}_t)$ is the completion of the filtration generated by the $d_{0}$-dimensional Brownian motion $(W^0(t))$, and, for each $i \in \{0,\ldots,N\}$, $(\overline{\mathcal{F}}^{i}_t)$ is the complete and right-continuous augmentation of the filtration generated by $d_1$-dimensional Brownian motions $(W^i(t))$, as well as a $(W^i(t))$-independent $d$-dimensional square-integrable random variables $(\xi^{i})_{i=1}^N$, which have by assumption the same law. Finally, we introduce the product probability spaces $\Omega^{i} = \overline{\Omega}^{0} \times \overline{\Omega}^{i}$, $\mathcal{F}^{i}$, $(\mathcal{F}_t^{i})$, $\mathbb{P}^{i}$, $i \in \{1,\ldots,N\}$, where $(\mathcal{F}^{i}, \mathbb{P}^{i})$ is the completion of $(\overline{\mathcal{F}}^{0} \otimes \overline{\mathcal{F}}^{i}, \overline{\mathbb{P}}^{0} \otimes \overline{\mathbb{P}}^{i})$ and $(\mathcal{F}_t^{i})$ is the complete and right-continuous augmentation of $(\overline{\mathcal{F}}_t^{0} \otimes \overline{\mathcal{F}}_t^{i})$. In the same way, we define the complete probability space $(\Omega, \mathcal{F}, \mathbb{P})$ equipped with $(\mathcal{F}_t)$ satisfying the usual conditions as a product of $(\overline{\Omega}^{i}, \overline{\mathcal{F}}^{i}, \overline{\mathbb{P}}^{i}, (\overline{\mathcal{F}}_t^i))_{i=0}^{N}$.\\
\indent Let $\Gamma$ be a closed and convex subset of $\mathbb{R}^{d_2}$, the set of control actions, or action space. Moreover, given a probability space $(\Omega, \mathcal{G}, \mathbb{P})$ and a filtration $(\mathcal{G}_t)$ in $\mathcal{G}$, let:
\begin{enumerate}[label=(S\arabic*)]
\item\label{itm:S1} $\mathbb{L}^2(\mathcal{G};\mathbb{R}^{n})$ be the set of $\mathbb{R}^{n}$-valued $\mathcal{G}$-measurable square-integrable random variables $U$. 
\item\label{itm:S2} $\mathbb{S}^2((\mathcal{G}_t);\mathbb{R}^{n})$ be the set of $\mathbb{R}^{n}$-valued $(\mathcal{G}_t)$-adapted continuous processes $(U(t))$ such that $\|U\|_{\mathbb{S}^2}:=\mathbb{E}\left[\sup_{t \in [0,T]}|U(t)|^2\right]^{\frac{1}{2}}<\infty$.
\item\label{itm:S3} $\mathbb{H}^2((\mathcal{G}_t);\mathbb{R}^{n})$ be the set of $\mathbb{R}^{n}$-valued $(\mathcal{G}_t)$-progressively measurable processes $(U(t))$ such that $\|U\|_{\mathbb{H}^2}:=\mathbb{E}\left[\int_{0}^{T} |U(t)|^2\,dt\right]<\infty$. 
\end{enumerate}
\noindent We denote by $\mathcal{L}(U)$ the law of a random variable $U$, and by $\overline{U}(s)=\mathbb{E}[U(s)|\overline{\mathcal{F}}_s^{0}]$ the conditional expectation of $U(s)$ given $W^0(s)$.\\
\indent For $\mathcal{S}$ a Polish space, let $\mathcal{P}(\mathcal{S})$ denote the space of probability measures on $\mathcal{B}(\mathcal{S})$, the Borel sets of $\mathcal{S}$. For $s \in \mathcal{S}$, let $\delta_{s}$ indicate the Dirac measure concentrated in $s$. Equip $\mathcal{P}(\mathcal{S})$ with the topology of weak convergence of probability measures. Then $\mathcal{P}(\mathcal{S})$ is again a Polish space. Let $d_{\mathcal{S}}$ be a metric compatible with the topology of $\mathcal{S}$ such that $(\mathcal{S}, d_{\mathcal{S}})$ is a complete and separable metric space. Given a complete compatible metric $d_{\mathcal{S}}$ on $\mathcal{S}$, we also consider the space of probability measures on $\mathcal{B}(\mathcal{S})$ with finite $p$-moments, with $p \geq 1$:    
\begin{equation*}
    \mathcal{P}_p(\mathcal{S}) \doteq \left(\nu \in \mathcal{P}(\mathcal{S})\,:\,\exists s_0 \in \mathcal{S}\,:\,\int_{\mathcal{S}}d_{\mathcal{S}}(s, s_0)^p\nu(\,ds)<\infty\right).
\end{equation*}
In particular, $\mathcal{P}_p(\mathcal{S})$ is a Polish space. A compatible complete metric is given by: 
\begin{equation*}
    d_{\mathcal{P}_p(\mathcal{S})}(\nu,\tilde{\nu})\doteq\left(\inf_{\substack{\alpha \in \mathcal{P}(\mathcal{S}\times\mathcal{S})\,:\,[\alpha]_1 = \nu\,\text{and}\,[\alpha]_2 =\tilde{\nu}}}\int_{\mathcal{S}\times\mathcal{S}}d_{\mathcal{S}}(s,\tilde{s})^p\alpha(\,ds, \,d\tilde{s})\right)^{1/p},
\end{equation*}
where $[\alpha]_1$ ($[\alpha]_2$) denotes the first (second) marginal of $\alpha$; $d_{\mathcal{P}_p(\mathcal{S})}$ is often referred to as the $p$-Wasserstein (or Vasershtein) metric.
\section{Theoretical model}\label{sec::theoreticalmodel}
This section proposes a stochastic equilibrium model for environmental markets accounting for the design of today's emission system. Our model is an integrated production-pollution-abatement model (e.g., \cite{anand2020pollution}) in continuous time, which combines a model of competing producers with a pollution model that includes pollution generation and abatement. Precisely, we consider $N \geq 1$, $N \in \mathbb{N}$, indistinguishable competing, profit-maximizing firms, whose carbon emissions are regulated in a cap-and-trade fashion. Although the regulation of carbon allowances occurs over several periods and allowances can be banked from one period to the other, we follow \cite{barrieu2014market}, \cite{carmona2010marke}, \cite{fell2010alternative}, \cite{aid2023optimal}, \cite{kollenberg2016emissions}, \cite{kollenberg2016dynamics}, and we focus on a single period of $T$ years at the end of which compliance is assessed.   We assume that capital is created according to a standard AK model with a positive depreciation rate of capital and with a positive technological level $A_k^{i}$; see the term ``revenues" in the cost functional in Equation \eqref{eq::costfunctional}. Let $K^{i}(t)$ be the level of capital at time $t$ of firm $i$, for $i=1,\ldots,N$. We assume that the dynamics of $(K^i(t))$ is described by the following stochastic differential equation (SDE)    
\begin{equation}\label{eq::capitaldynamics}
    dK^i(t) = (\kappa_f^{i} K^{f,i}(t) + \kappa_g^{i} K^{g,i}(t)-\delta^{i} K^{i}(t))\,dt + \sigma K^{i}(t)\,dW^{1,i}(t),\quad K^{i}(0)=\kappa_0,
\end{equation}
where $\kappa_f^{i}, \kappa_g^{i}, \sigma, \delta^{i}$ are positive constants. $K^{f,i}(t)$ and $K^{g,i}(t)$ represent the amount of fossil-fuel and green energy based level of capital used by firm $i$ for capital creation. $\delta^{i}$ is the depreciation rate of capital. The quantity $\sigma K^{i}(t)$ represents the standard deviation of the level of capital of firm $i$ and depends on $K^{i}(t)$ itself; $(W^{1,i}(t))$ is a standard Brownian motion. Firm $i$,  for $i=1, \ldots, N$, controls the level of capital trend via $K^{f,i}(t)$ and $K^{g,i}(t)$, which increases at a rate $\kappa_f^{i} K^{f,i}(t) + \kappa_g^{i} K^{g,i}(t)-\delta^{i} K^{i}(t)$, while the volatility is uncontrolled. We assume that $K^{f,i}(t), K^{g,i}(t) \in \mathbb{H}^2((\mathcal{F}_t^{i}),\mathbb{R})$. Moreover, we assume that firm $i$, for $i=1, \ldots, N$, faces quadratic costs, say $C^i(K^{f})$ and $C^i(K^{g})$, in the capital levels\footnote{Admittedly, we have not found a precise reference for this assumption, but it seems reasonable and necessary to obtain a linear quadratic form for our problem.}. More precisely, we have:
\begin{equation}\label{eq::productioncosts}
    C^{i,f}(K^{f}) = c_{1,1}^{i}K^{f,i} + c_{1,2}^{i}(K^{f,i})^2,\,\,\text{and}\,\,C^{i,g}(K^{g}) = c_{2,1}^{i}K^{g,i} + c_{2,2}^{i}(K^{g,i})^2.
\end{equation}
Firms generate pollution as a byproduct of the production process. Let $E_i$ denote the pollution, in terms of carbon emissions, generated by firm \textit{i} prior to any investment in abatement; we name $E_i$ as business-as-usual (BAU) emissions\footnote{Henceforth, to keep the reading, we use the term emission instead of carbon emission.}. Clearly, $E^{i}$ must be increasing in the capital level, and it seems reasonable to assume that it is a function of $K^{f,i}(t)$. We follow \cite{anand2020pollution}, and we impose a linear relation between the BAU emissions and the latter by assuming:
\begin{equation}\label{eq::bau}
    dE^i(t) = \textcolor{black}{\kappa_e^{i}} K^{f,i}(t)\,dt + \sigma_1^{i} d\widetilde{W}^{2,i}(t),\quad E^{i}(0)=E_0,
\end{equation}
where \textcolor{black}{$k_e^i>0$ characterizes the linear relationship,} and $\widetilde{W}^{2,i}(t)=\sqrt{1-\rho_i^2} W^{2,i}(t) + \rho_i W^{0,1}(t)$, $\rho_i \in [0,1]$, so that the correlation between $\widetilde{W}^{2,i}(t)$ and $\widetilde{W}^{2,j}(t)$ is $r_{i j}:=\rho_i\rho_j$. The increments of $W^{2,i}$ and $W^{0,1}$ are independent, and independent from the increments of $W^{1,i}$. The noise decomposition captures the fact that the emission of firm $i$ is affected by its own idiosyncratic noise $d W^{2,i}$ and by the common economic business cycle $d W^{0,1}$. A similar model for the BAU emissions is employed in \cite{aid2023optimal}. Here, we also assume the presence of a short-term emission shock $d e^{i}(t) = \sigma_2^{i} dW^{3,i}(t)$, $e^{i}(0)=e_0$ which may represent, e.g., the outage of a carbon-friendly production unit that is instantaneously replaced by a more polluting one (e.g., \cite{hitzemann2018equilibrium}). Increments of $W^{3,i}$ are independents from the increments of $W^{1,i}$, $W^{2,i}$, $W^{0,1}$. Therefore, the total emission dynamics is given by
\begin{equation}\label{eq::totalemission}
    d\widetilde{E}^i(t) = dE^{i}(t) + de^{i}(t) = (\textcolor{black}{\kappa_e^{i}} K^{f,i}(t)\,dt + \sigma_1^{i} d\widetilde{W}^{2,i}(t)) + (\sigma_2^{i} dW^{3,i}(t)),
\end{equation}
where $\widetilde{E}^i(0)=\widetilde{E}_0$. We describe now our model of pollution abatement. Firm $i$, for $i=1, \ldots, N$, breaks down emissions via two complementary notions: (i) abatement level, and (ii) abatement cost. Under the abatement effort rate $\alpha^{i}(t) \in \mathbb{H}^2((\mathcal{F}_t^{i}), \mathbb{R})$, emissions of firm $i$ becomes:       
\begin{equation}\label{eq::total}
    d\widetilde{E}^{i,\alpha}(t) = (\textcolor{black}{\kappa_e^{i}} K^{f,i}(t)-\alpha^{i}(t))\,dt + \sigma_1^{i} d\widetilde{W}^{2,i}(t)) + (\sigma_2^{i} dW^{3,i}(t)).
\end{equation}
In this way, the firm controls its emission trend, which increases at a rate $(\textcolor{black}{\kappa_e^{i}} K^{f,i}(t)-\alpha^{i}(t))$; on the other hand, the volatility remains uncontrolled.   Notice that, contrary to \cite{anand2020pollution}, our model does not assume that pollution is observable and a deterministic function of the level of capital. 
As regards the abatement costs, the extant literature assumes that pollution abatement costs are increasing and quadratic (see, e.g., \cite{subramanian2007compliance}), or at least convex increasing in the quantity of emission abated, which is $\alpha^{i}(t)$ in our model (see, e.g., \cite{levi2004converting}). This is because usually the initial units of emissions are easy to abate, but once the low-hanging
fruit have been exploited, pollution abatement becomes increasingly difficult (see, e.g., \cite{anand2020pollution}). Thus, we assume the following quadratic form for the abatement cost function:
\begin{equation}\label{eq::costabatement}
    C_i(\alpha) = h_i \alpha^i + \frac{1}{2 \eta_i} (\alpha^i)^2,\quad h_i, \eta_i > 0,
\end{equation}
where the constant $\eta_i$ is positively correlated with the flexibility of the abatement process and, therefore, with the reversibility of the decision. Before describing the dynamics for the bank account $\widetilde{X}$, we detail the competition mechanism in our model. We assume that firm $i$, for $i=1, \ldots, N$, faces linear inverse demand curve $p^{i}(t):=p(K^{i}(t),K^{-i}(t))$, which can be derived by a suitable quadratic utility function\footnote{It is not difficult to see that such a linear demand function can be derived from a quadratic utility function of the following form -- for the sake of simplicity we denote by $q_i$ the quantity produced by firm $i$ --: $U(q) = a\sum_{i=1}^{N}q_i-C_1\left(\sum_{i=1}^N q_i^2 + C_2 \sum_{\substack{j=1\\j\neq i}}^{N-1} q_i q_j\right)-q \cdot p$, where $C_1 = \frac{b}{2}\left(1-\gamma\left(1-\frac{1}{N}\right)\right)$ and $C_2 = \frac{\gamma}{N(1-\gamma)+\gamma}$ and solving the utility maximization problem for a representative consumer; the first-order condition gives Equation \eqref{eq::demandcurve}.}, given by 
\begin{equation}\label{eq::demandcurve}
    p^{i}(t) = a - b \textcolor{black}{(1-\gamma)A_k^{i}K^{i}(t)} - b \gamma \frac{1}{N}\sum_{\substack{j=1}}^{N} \textcolor{black}{A_k^{j} K^j(t)},
\end{equation}
where $a$ and $b$ are positive constants, and $\gamma \in [0,1]$ captures the degree of production substitution, and hence competition, and \textcolor{black}{$A_k^{i}$ represents the technological level and hence $A_k^{i} K^i(t)$ is the production function of the firm $i$}. We continue with the following observation: the demand function in Equation \eqref{eq::demandcurve} comprises a range of competitive markets, in which monopoly and Cournot oligopoly are polar cases.   Indeed, it can be written in the following way:
\begin{equation*}
\begin{split}
    p^{i}(t) = a - b \left(1 - \gamma\left( 1 - \frac{1}{N} \right)\right) A_k^{i} K^{i}(t) - b \gamma \frac{1}{N} \sum_{\substack{j=1\\j\neq i}}^{N-1} A_k^{j} K^j(t)
\end{split}
\end{equation*}
When $\gamma=0$ (i.e., monopoly), then $p^{i}(t)= a - b A_k^{i} K^{i}(t)$.  When $\gamma = 1$ (i.e., Cornout oligopoly), then $p^{i}(t) = a - b \gamma \frac{1}{N} \sum_{\substack{j=1}}^{N} A_k^{j} K^j(t)$ and so there is perfect competition. On the other hand, $\gamma \in (0,1)$ captures the degree of substitution.    In particular, the price $p^{i}(t)$ is influenced more by the level of production of the corresponding good $i$ with respect to the total quantity produced by all the other firms (firm $i$ excluded); indeed $b\left(1-\gamma\left(1-\frac{1}{N}\right)\right) > b \gamma \frac{1}{N}$ for every $\gamma \in (0,1)$, which is a natural (although myopic in the sense of emotions) postulation. Notice that this type of asymmetry is not in contrast with the symmetry required by the MFG framework. We now turn to the description of the dynamics for the bank account.\\
\indent The bank account's dynamics is specified as in \cite{aid2023optimal}. We assume that the regulator opens for each firm $i$, $i=1,\ldots,N$, at $t=0$ a bank account $\widetilde{X}^{i}$ and allocates permits, which are represented by the cumulative process $\widetilde{A}^{i}$. The dynamics of the bank account is given by: 
\begin{equation}\label{eq::bankaccount}
    d\widetilde{X}^i_t = \beta^i(t)\,dt + d\widetilde{A}^{i}(t) - d\widetilde{E}^{i,\alpha}(t),
\end{equation}
where $\beta^i(t)\in\mathbb{H}^2((\mathcal{F}^{i}_t),\mathbb{R})$ is the trading rate in the liquid allowance market; emissions, trade, and bank account are measured in tons or in multiples of tons). We assume that $\widetilde{A}^{i}$ has the following dynamics:
\begin{equation}\label{eq::cumulativeallowance}
    d\widetilde{A}^{i}(t) = \tilde{a}^i(t)\,dt + \tilde{\sigma}_2\,dW^{0,2}(t),\quad\widetilde{A}^{i}(0)=\widetilde{A}_0.
\end{equation}
where $(W^{0,2}(t))$ is a standard Brownian motion common to all firms independent from all the other noises involved in the model.    The fact that $(W^{0,2}(t))$ is independent from $(W^{0,1}(t))$ is admittedly a very heavy assumption; the case of correlated common noises is an important subject for future research; see the discussion in Section \ref{sec::futureresearch}.  The choice of a dynamic allocation mechanism, instead of a static one, is because of the presence of the common shock in the BAU emissions dynamic, otherwise there would be no benefit from the implementation of a dynamic allocation scheme; see \cite{aid2023optimal}. The quantity $\tilde{a}$ represents the rate. Notice that the sign of $\widetilde{A}^{i}$ can be either negative, meaning that the regulator is placing a penalty on the firm bank account, or positive, meaning that the regulator is giving true permits to the firm. Admittedly, the dynamics for $\widetilde{A}^{i}(t)$ could be more general, by adding, for instance, a pure jump part (see \cite{aid2023optimal}, Section 2). Firm $i$, $i=1, \ldots, N$, controls the trading rate $\beta^i(t)$.\\
\indent Now, let $(\overline{\omega}_t)$ be the price of allowances. It can be either exogenous, for example described by the Black-Scholes model, or endogenous, in the sense that it is determined by the fundamental condition of the market. More precisely, the total number of tons of emissions being purchased by a firm via the emission exchange at a given time must be equal to the number of those being sold by others via the emission exchange at the same time. In particular, the balance between the sales and the purchases, called market clearing conditions, must hold at any point in time. The market clearing condition reads as: 
\begin{equation}\label{eq::marketclearingcondition}
    \sum_{i=1}^{N}\hat{\beta}^{i}(t) = 0,\quad dt \otimes d\mathbb{P}-\text{a.e.},
\end{equation}
where $\hat{\beta}_t^{i}$ is the trading rate of the $i$th firm. We will be interested in finding an appropriate price process $(\overline{\omega}_t)$ so that it achieves the market clearing condition among the rational firms.  More precisely, as stated in the introduction, we are interested in finding a market equilibrium, which is defined as trading strategies and market price such that each firm has minimized its criteria and the market clears for the market price. We assume that firms are price-takers, an assumption that is in line with the large number of companies regulated under today's emission trading systems\footnote{Indeed, more than 5,500 firms are regulated under the EU ETS (e.g., \cite{calel2016environmental})}, and minimize their cost functional by finding an optimal trade-off between implementing abatement measures, trading permits in the market, and taking the risk of penalty payments. Another possibility is to model the price as an underlying martingale plus a drift representing a form of permanent price impact (e.g., \cite{carmona2015probabilistic}):
\begin{equation*}
    d\overline{\omega}_t = \frac{\tilde{\nu}}{N} \sum_{i=1}^{N} c^{'}(\beta^{i}(t))\,dt + \sigma_0 dW_t,
\end{equation*}
which in the particular case $c(\beta)=\beta^2$ corresponds to the influential Almgren-Chriss model (\cite{almgren2001optimal}); $\tilde{\nu}>0$. Notice that in this case, we would obtain a mean-field game of control with common noise. We follow \cite{fujii2022mean}, and we state that the process $(\overline{\omega}_t)$ is likely to be given by a $\overline{\mathcal{F}}^{0}$-progressively measurable process since the effects from the idiosyncratic parts from many firms are expected to be canceled out. Moreover, we assume that if firm $i$, $i=1,\ldots,N$, places a (market) order of $\beta^i(t)$ when the market price is $(\overline{\omega}_t)$, then the cost incurred by the firms is given by
\begin{equation}\label{eq::costoftrading}
    \beta^{i}(t)\overline{\omega}_t + \frac{1}{2\nu}(\beta^{i}(t))^2,
\end{equation}
where $\nu>0$ is the (constant) market depth parameter which takes into account a price impact effect as in the original work of \cite{kyle1985continuous}. We make now the following remark. 
\begin{remark}
In our model, we do not enforce constraints on the controls because we give priority to finding explicit solutions to our problem. Indeed, our goal is to analyze the qualitative behaviour of the system, a goal achieved in a satisfactory way in the numerical section. A similar ``relaxation" of the problem is done also in, e.g., \cite{aid2016optimal2} and \cite{alasseur2020extended}.    
\end{remark}
Let $\mathcal{H}_1^{N} := \mathbb{H}^2((\mathcal{F}_t^{i}),\mathbb{R}^4)$, and let $\mathcal{H}_N^{N}$ the set of all $N$-dimensional vectors $\boldsymbol{v}^{N}:=(v^{N,1},\ldots,v^{N,N})$ such that $v^{N,i} \in \mathcal{H}_1^{N}$, with the vector $v^{N, i}(t)$ defined as $v^{N, i}(t):=(K^{f,i}(t), K^{g,i}(t), \alpha^{i}(t), \beta^{i}(t))$, for $i = 1, \ldots, N$. Each element of $\mathcal{H}_N^{N}$ is called strategy vector. Under the assumption of risk neutrality, firm $i$, for $i = 1, \ldots, N$, evaluates a strategy vector $\boldsymbol{v}^N \in \mathcal{H}_N^{N}$ according to its expected cost (notice that we highlight the dependence on the vector $\boldsymbol{v}^{N}$ in the term in Equation \eqref{eq::demandcurve})
\begin{equation}\label{eq::costfunctional}
\begin{split}
    \mathcal{J}^{i}(\boldsymbol{v}^N)  &= \mathbb{E}\Bigg[\int_{0}^{T}\underbrace{-p^{i}(t, \boldsymbol{v}^N)\textcolor{black}{A_k^{i}K^{i}(t)}}_{- \text{revenues}} + \underbrace{\beta^i(t)\overline{\omega}_t+\frac{1}{2\nu}(\beta^{i}(t))^2}_{\text{cost of trading \eqref{eq::costoftrading}}} + \underbrace{C_i(\alpha(t))}_{\text{abatement cost \eqref{eq::costabatement}}}\\
                            &+ \underbrace{C^{i,f}(K^{f}(t))+C^{i,g}(K^{g}(t))}_{\text{costs of production \eqref{eq::productioncosts}}}\,dt + \underbrace{\lambda(\widetilde{X}^i(T))^2}_{\text{final penalization}}\Bigg]
\end{split}
\end{equation}
where the term $\lambda (\widetilde{X}_T^i)^2$, with $\lambda>0$, is the terminal monetary penalty on the bank accounts set by the regulator, which is a regularized version of the terminal cap penalty function applied in practice, which is zero if the firm is compliant and linear otherwise. However, as cited in \cite{carmona2009optimal}, optimal strategies cannot be found in closed form in this case. Through the previous penalty, the firm is going to pay both if its bank account is above or below the compliance zero level. In particular, notice that the Brownian motions in the dynamics for $\widetilde{A}^{i}$ and $\widetilde{E}^{i,\alpha}$ are chosen for tractability reason, because of the additive quadratic penalty $\lambda (\widetilde{X}_T^i)^2$ in the cost functional \eqref{eq::costfunctional}.\\
\indent Finally, the dynamics for the capital level and the bank account are given by:
\begin{equation}\label{eq::statedynamics}
    \begin{cases}
        \begin{split}
        dK^i(t)&= (\kappa_f^{i} K^{f,i}(t) + \kappa_g^{i} K^{g,i}(t)-\delta^{i} K^i(t))\,dt + \sigma K^{i}(t)\,dW^{1,i}(t),\quad K^{i}(0)=\kappa_0.\\
        d\widetilde{X}^{i}(t)&= (\beta^i(t)+\tilde{a}^i(t)+\alpha^i(t)-\textcolor{black}{\kappa_e^{i}} K^{f,i}(t))\,dt + \tilde{\sigma}_2\,dW^{0,2}(t)-\sigma_1^{i} \rho_i\,dW^{0,1}(t)\\
                             &-\sigma_1^{i}\sqrt{1-\rho_i^2}\,dW^{2,i}(t)-\sigma_2^{i} dW^{3,i}(t),\quad\widetilde{X}^{i}(0)=\widetilde{X}_0.
    \end{split}
    \end{cases}
\end{equation}
In addition, see, again, Equation \eqref{eq::demandcurve}, we have:
\begin{equation*}
    p^{i}(t,\boldsymbol{v}^N) = a - b \textcolor{black}{(1-\gamma)A_k^{i}K^{i}(t)} - b \gamma \frac{1}{N}\sum_{\substack{j=1}}^{N} \textcolor{black}{A_k^{j} K^j(t)} K^j(t).
\end{equation*}
In the present paper, we take a non-cooperative game point of view. The aim of each firm $i$, for $i=1,\ldots,N$, is to minimize the cost in Equation \eqref{eq::costfunctional} by controlling the level of capital linked to fossil-fuel and green technologies, the quantity of emissions abated, and the trading rate in the allowance market. In a non-cooperative game setting, we are led to the analysis of a non-zero sum stochastic game with $N$ players and to the search of $\epsilon$-Nash equilibrium. In the next definition, we use the standard notation $[v^{N,-i}, v]$ to indicate a strategy vector equal to $\boldsymbol{v}^N$ for all firms but the $i$-th, which deviates by playing $v \in \mathcal{H}_1^{N}$ instead.  
\begin{definition}[$\epsilon$-Nash equilibrium for the $N$-players game]\label{def::nashequilibria}
Let $\epsilon \geq 0$. A strategy vector $\boldsymbol{v}^N \in \mathcal{H}_N^{N}$ is called $\epsilon$-Nash equilibrium for the $N$-player game if for every $i \in \{1,\ldots,N\}$ and for any deviation $v \in \mathcal{H}_1^{N}$ we have:
\begin{equation*}
    \mathcal{J}^{i}(\boldsymbol{v}^N) \leq \mathcal{J}^{i}([v^{N,-i}, v]) + \epsilon.
\end{equation*}
\end{definition}

\noindent Before proceeding, we summarize our notation in Table \ref{tab:par_model}.

\begin{table}[htbp]
    \centering
    \caption{The table summarizes and explains the notations used in the theoretical model. The index $i$ generally refers to company $i$ in the economy, and by omitting the $i$ we refer to an economy-wide version of the respective variable. Time dependence is indicated in parentheses. BAU stands for Business As Usual. GBM stands for Geometric Brownian Motion.}
    \begin{tabular}{ll} \toprule
        \textbf{Notation.} & \textbf{Synthetic description} \\ \midrule
        $K^{i}(t)$         & Level of capital of company $i$ at time $t$.\\
        $K^{f,i}(t)$       & Fossil-fuel level of capital of company $i$ implemented at time $t$.\\
        $K^{g,i}(t)$       & Green level of capital of company $i$ implemented at time $t$.\\
        $\kappa_f^{i}$     & Amount of fossil-fuel level of capital of company $i$ implemented $\forall\,t$.\\
        $\kappa_g^{i}$     & Amount of green level of capital of company $i$ implemented $\forall\,t$.\\
        $\sigma$            & Volatility of the level of capital of company $i$.\\
        $C^{i,f}(K^{f})$    & Cost associated to the fossil-fuel level of capital; quadratic; $c_{1,1}^{i}, c_{1,2}^{i}$. \\
        $C^{i,g}(K^{g})$    & Cost associated to the green level of capital; quadratic; $c_{2,1}^{i}, c_{2,2}^{i}$. \\
        $E^{i}(t)$          & BAU emissions of company $i$ at time $t$.\\
        $\sigma_1^{i}$          & Volatility of BAU emissions of company $i$.\\
        $\rho_i$            & Correlation between the idiosyncratic and the common noise of $E^{i}(t)$.\\   
        $e^{i}(t)$          & Short-term emission shocks of company $i$ at time $t$.\\
        $\sigma_2^{i}$          & Volatility of short-term emission shocks of company $i$.\\
        $\widetilde{E}^{i}(t)$ & BAU + short-term emission of company $i$ at time $t$.\\
        $\alpha^{i}(t)$     & Emission abatement effort rate of company $i$ implemented at time $t$.\\
        $C_i(\alpha)$       & Abatement cost function function; linear-quadratic; $h_i$, $1/(2 \eta_i)$.\\
        $p^{i}(t)$          & Price of the good produced by company $i$ at time $t$; inverse demand; $a$, $b$.\\ 
        $\gamma$            & Degree of production substitution; $\gamma \in [0,1]$.\\
        $\widetilde{X}^{i}(t)$ & Carbon emission bank account of company $i$ at time $t$.\\ 
        $\beta^{i}(t)$      & Trading rate in $\widetilde{X}^{i}(t)$ of company $i$ implemented at time $t$.\\  
        $\widetilde{A}^{i}(t)$ & Allocated permits of company $i$ at time $t$; GBM, $\tilde{a}^i(t),\,\tilde{\sigma}_2$.\\
        $\nu$               & Market depth.\\
        $\lambda$           & Terminal monetary penalty on the bank account set by the regulator.\\
        \textcolor{black}{$\delta^{i}$}           & Depreciation rate of the capital stock.\\
        \textcolor{black}{$A_k^{i}$}           & Productivity level of the capital.\\
        \bottomrule
    \end{tabular}
 \label{tab:par_model}
\end{table}

\section{A mean field game approximation with common noise for the $N$ player game.}\label{sec::mfgapproximation}
In this section, we consider the filtered probability space $(\Omega, \mathcal{F}, \mathbb{P}, (\mathcal{F}_t))$, $d_1=3$ Brownian motions $(W^{j}(t))$, $1\leq j \leq 3$, which are mutually independent and independent from the completion of the filtration $(\overline{\mathcal{F}}_t^0)$, defined in Section \ref{sec::notations}.    In order to find the expression for the MFG approximation, we follow \cite{lacker2019mean} and we introduce the \textit{type vectors} $\zeta_i = (\kappa_f^{i}, \kappa_g^{i}, \sigma_1^{i}, \rho_i, \sigma_2^{i}, \delta^{i}, A_k^{i})$, for $i = 1,\ldots,N$.  As said in the introduction, the finite set of firms becomes a continuum and competes with the rest of the (infinite) population.   In particular, the MFG is defined in terms of a representative firm who is assigned a \textit{random}-type vector $\zeta = (\kappa_f, \kappa_g, \sigma_1, \rho, \sigma_2, \delta, A_k)$ at time zero, which encodes the distribution of the (continuum of) firms' types. Formally, cfr. \cite{lacker2019mean}, the type vector $\zeta_i$ induces an empirical measure called the \textit{type distribution}, which is the probability measure on the type space $\mathcal{Z}^{e}:=\mathbb{R} \times \mathbb{R} \times \mathbb{R}_{+} \times [0,1] \times \mathbb{R}_{+} \times [0,1] \times \mathbb{R}_{+}$, given by $m_N(A) = \frac{1}{N}\sum_{i=1}^{N} \mathrm{1}_{A}(\zeta_i),$ for Borel sets $A \subset \mathcal{Z}^{e}$. We assume now that as the number of firms becomes large, $N\rightarrow\infty$, the just introduced empirical measure $m_N$ has a weak limit $m$, in the sense that $\int_{\mathcal{Z}^{e}} \varphi dm_N \rightarrow \int_{\mathcal{Z}^{e}} \varphi dm$ for every bounded continuous function $\varphi$ on $\mathcal{Z}^{e}$; this holds almost surely if the $\zeta_i^{'}$ are i.i.d. samples from $m$.  In particular, the probability measure $m$ represents the distribution of type parameters $\zeta$ among the continuum of firms. At this point, let $x_0=(k_0, \tilde{x}_0)$ be a random vector which is independent from $(\overline{\mathcal{F}}_t^0)$. The representative firm's level of capital $K$ and bank account $\widetilde{X}$ solve
\begin{equation}\label{eq::statedynamicsmfglimit}
    \begin{cases}
        \begin{split}
            dK(t)&=(\kappa_f K^{f}(t) + k_g K^{g}(t)-\delta K^i(t))\,dt +\sigma K(t)\,dt,\quad K(0)=\kappa_0.\\
            d\widetilde{X}(t)&=(\beta(t) + \tilde{a}(t) + \alpha(t) - \textcolor{black}{\kappa_e} K^f(t))\,dt + \tilde{\sigma}_2 dW^{0,2}(t)-\sigma_1 \rho dW^{0,1}(t)\\
                             &-\sigma_1\sqrt{1-\rho^2}dW^{2}(t)-\sigma_2 dW^{3}(t),\quad \tilde{X}(0)=\tilde{x}_0,
        \end{split}
    \end{cases}
\end{equation}
where $v(t):=(K^{f,i}(t), K^{g,i}(t), \alpha^{i}(t), \beta^{i}(t))$ belongs to the space $\mathbb{H}^2((\mathcal{F}_t),\mathbb{R}^4)$. Moreover, by setting $\overline{K}(t)=\mathbb{E}[K(t)|\overline{\mathcal{F}}_s^0]$, with $s \leq t$, we denote 
\begin{equation*}
    p^{\overline{K}}(t) = a - b\textcolor{black}{(1-\gamma)A_kK(t)} - b \gamma \textcolor{black}{A_k\overline{K}(t)},
\end{equation*}
where we assumed that for a large number of firms $N$, we approximate the dynamics in Equation \eqref{eq::demandcurve} by the expression in the previous equation, where we used directly the quantity $\overline{K}(t)$ instead of a generic $(\overline{\mathcal{F}}_t^0)$ adapted real-valued process since the dynamics of $p^{i}$ is uncontrolled (see, also, the discussion in \cite{alasseur2020extended}, Section 3, Page 653).\\
We now consider the following cost functional: 
\begin{equation}\label{eq::functionalmfglimit}
    \begin{split}
        \mathcal{J}^{NE}(v;\overline{K})&= \mathbb{E}\Bigg[\int_{0}^{T}- p^{\overline{K}}(t) K (t) + \beta(t)\overline{\omega}_t+\frac{1}{2\nu}(\beta(t))^2 + C(\alpha(t))\\
                                &+ C^{f}(K^{f}(t))+C^{g}(K^{g}(t))\,dt + \lambda(\widetilde{X}(T))^2\Bigg],
    \end{split}
\end{equation}
where the superscript $\text{NE}$ stands for Nash equilibrium. Equations \eqref{eq::statedynamicsmfglimit} and \eqref{eq::functionalmfglimit} represent a MFG of linear quadratic type which fits into the framework studied in \cite{graber2016linear}, Section 3, apart from the presence of terms of order zero in both the private state dynamics $X(t):=(K(t),\widetilde{X}(t))^T$ and the running cost functional; see the discussion in Appendix \ref{app::LQMFG} and  Appendix \ref{app::LQMFC}. For the reader's convenience and to maintain the present work as self-contained as possible, Appendix \ref{app::LQMFG} presents the class of linear quadratic MFGs considered here. In particular, by using the notation in Appendix \ref{app::LQMFG}, the non-zero matrices characterizing the dynamics of $X(t)$ are the following:
\begin{equation}\label{eq:matricesdsexogenousprice}
    A_0(s) =
    \begin{bmatrix}
            0\\
 \tilde{a}(s)\\  
    \end{bmatrix},\,\,
    \textcolor{black}{A =
    \begin{bmatrix}
        -\delta &    0\\
0& 0\\  
    \end{bmatrix}},\,\,
    B = 
    \begin{bmatrix}
        \kappa_f & \kappa_g & 0 & 0 \\
       \textcolor{black}{-\kappa_e} &                0 & 1 & 1 \\
    \end{bmatrix},\,\,
\end{equation}
\begin{equation}\label{eq:matricesdWexogenousprice}
    C_{0,2} =
    \begin{bmatrix}
            0\\
            -\sigma_1\sqrt{1-\rho^2}\\  
    \end{bmatrix},\,\, 
    C_{0,3} =
    \begin{bmatrix}
            0\\
            -\sigma_2\\  
    \end{bmatrix},\,\,
    C_1  =
    \begin{bmatrix}
   \sigma & 0\\
        0 & 0\\
    \end{bmatrix},\,\,
\end{equation}
\begin{equation}\label{eq:matricesdW0exogenous}
\hspace{-4cm}
    F_{0,1} =
    \begin{bmatrix}
            0\\
-\sigma_1\rho\\  
    \end{bmatrix},\,\,
    F_{0,2} =
    \begin{bmatrix}
              0\\
\tilde{\sigma}_2\\  
    \end{bmatrix}.
\end{equation}
\noindent Whereas, the ones characterizing the cost functional are given by:   
\begin{equation}\label{eq:matricesexogenouspricefunctional}
    Q = 
    \begin{bmatrix}
        \textcolor{black}{b(1-\gamma) A_k^2} & 0\\
        0 & 0\\
    \end{bmatrix}
    \quad
    \overline{Q} = 
    \begin{bmatrix}
        \textcolor{black}{\frac{b\gamma A_k^2}{2}} & 0\\
        0       & 0\\
    \end{bmatrix}
    \quad
    R = 
    \begin{bmatrix}
        c_{1,2} & 0      & 0& 0\\
        0       & c_{2,2}& 0& 0\\
        0       & 0& \frac{1}{2\eta}&0\\
        0       & 0& 0& \frac{1}{2\nu}
    \end{bmatrix}.
\end{equation}

\begin{equation}\label{eq:matricesexogenouspricefunctionaltwo}
    q = 
    \begin{bmatrix}
            \textcolor{black}{- \frac{aA_k}{2}}\\
              0\\
    \end{bmatrix}
    \quad 
    r(s)=
    \begin{bmatrix}
            \frac{c_{1,1}}{2}\\
            \frac{c_{2,1}}{2}\\
            \frac{h}{2}\\
            \frac{\overline{\omega}(s)}{2}\\
    \end{bmatrix}
\quad H=\begin{bmatrix} 
        0 & 0 \\
        0 & \lambda\\
   \end{bmatrix}
\end{equation}
Before proceeding, we make the following remark. 
\begin{remark}\label{rmk::marketimpact}
In order to satisfy assumptions \ref{itm:N5}, the matrix $R$ has to be a positive-definite matrix, implying that $\nu$ must be in $(0,\infty)$. In particular, we do not consider directly the case without frictions as done in \cite{aid2023optimal}, which, admittedly, could be a common assumption in the carbon market (see, e.g., \cite{kollenberg2016emissions}). 
\end{remark}

Let us now assume that $(\overline{\omega}_t) \in \mathbb{H}^2((\overline{\mathcal{F}}^0_t);\mathbb{R})$, with $\overline{\omega}_T \in \mathbb{L}^2(\overline{\mathcal{F}}^0_T;\mathbb{R})$, is given. Then, we have the following characterization (see also \cite{graber2016linear}, Proposition 3.2), which is due to the uniform convexity of the functional $\mathcal{J}^{NE}(v;\overline{K})$ that guarantees the existence of a unique minimizer (see either \cite{graber2016linear}, Lemma 2.2, or \cite{yong2013linear}, Proposition 2.7). In order to not burden the reading, we omit the proof of the subsequent proposition because we will provide in Section \ref{sec::mfcapproximation}, Proposition \ref{prop::propositionexistenceanduniquenessmfc}, the proof of the characterization of the associated MFC problem, which follows the same line of argument. Before continuing, let us make the following observation about the notation.   The subsequent results involve the processes $Z \in \mathbb{H}^2((\mathcal{F}_t);\mathbb{R}^{2 \times 3})$ and $Z_0 \in \mathbb{H}^2((\mathcal{F}_t);\mathbb{R}^{2 \times 2})$. To denote the entry $(i,j)$ of $Z$ (resp. of $Z_0$), we will use the notation $Z^{(i)}_{j}$ (resp. $Z_{0,j}^{(i)}$).
\begin{proposition}\label{prop::propositionexistenceanduniquenessmfg}
Let $\overline{K}(t)=\mathbb{E}[K(t)|\overline{\mathcal{F}}_t^0]$, and $x_0 = (k_0, \tilde{x}_0)$ be a random vector which is independent from $(\overline{\mathcal{F}}^0_t)$. Then, there exists a unique control $v=v(\overline{K}, x_0)$ minimizing the functional in Equation \eqref{eq::functionalmfglimit}. Furthermore, let $X(s)=(K(s),\widetilde{X}(s))$ be the corresponding trajectory, i.e., the solution of Equation \eqref{eq::statedynamicsmfglimit} with control $v$. Then there exists a unique solution $(Y, Z, Z_0)\in \mathbb{S}^2((\mathcal{F}_t);\mathbb{R}^2) \times \mathbb{H}^2((\mathcal{F}_t);\mathbb{R}^{2 \times 3}) \times \mathbb{H}^2((\mathcal{F}_t);\mathbb{R}^{2 \times 2})$ of the following BSDE, with $s \in [0,T]$:
\begin{equation}\label{eq::bsdemfg}
    \begin{cases}
    \begin{split}
        dY^{(1)}(s)=&-\left(\textcolor{black}{-\delta Y^{(1)}(s)\,ds}+\sigma Z_1^{(1)}(s) + b \textcolor{black}{\overline{K}}(s) + \frac{b \gamma}{2}\overline{K}(s)-\frac{a}{2}\right)\,ds\\
                   &+ \sum_{j=1}^{3}Z_j^{(1)}dW^{j}(s)+ \sum_{j=1}^{2} Z_{0,j}^{(1)}dW^{0,j}(s),\quad Y^{(1)}(T)=0;\\
        dY^{(2)}(s)=& \sum_{j=1}^{3}Z_j^{(2)}dW^{j}(s)+ \sum_{j=1}^{2} Z_{0,j}^{(2)}dW^{0,j}(s),\quad Y^{(2)}(T)=\lambda(\widetilde{X}(T))^2
    \end{split}
    \end{cases}
\end{equation}
satisfying the coupling condition, with $s \in [0,T]$, \textit{a.s.},
\begin{equation}\label{eq::couplingmfg}
\begin{cases}
 K^{f}(s) &= - \frac{\kappa_f}{c_{1,2}} Y^{(1)}(s) + \frac{\kappa_f}{c_{1,2}} Y^{(2)}(s) - \frac{c_{1,1}}{2 c_{1,2}},\\
 K^{g}(s) &= - \frac{\kappa_g}{c_{2,2}} Y^{(1)}(s) - \frac{c_{2,1}}{2 c_{2,2} },\\
\alpha(s) &=- 2\eta Y^{(2)}(s) - \eta h,\\
 \beta(s) &= - 2\nu Y^{(2)}(s) - \nu \overline{\omega}(s).
\end{cases}
\end{equation}
Conversely, suppose $(X, v, Y, Z, Z_0) \in \mathbb{S}^2((\mathcal{F}_t);\mathbb{R}^2)  \times 
\mathbb{H}^2((\mathcal{F}_t);\mathbb{R}^{4})\times
\mathbb{S}^2((\mathcal{F}_t);\mathbb{R}^2)  \times \mathbb{H}^2((\mathcal{F}_t);\mathbb{R}^{2 \times 3}) \times \mathbb{H}^2((\mathcal{F}_t);\mathbb{R}^{2 \times 2})$ is a solution to the forward-backward system \eqref{eq::statedynamicsmfglimit}, \eqref{eq::bsdemfg} and coupling condition \eqref{eq::couplingmfg}. Then $v$ is the optimal control minimizing $\mathcal{J}^{NE}(v;\overline{K})$, and $X(s)$ is the optimal trajectory. In particular, $v$ is a mean field Nash equilibrium. 
\end{proposition}
Before proceeding,  let us comment on the coupling condition in Equation \eqref{eq::couplingmfg}, which applies also to the coupling condition in Equation \eqref{eq::couplingmfc} as well as when $\overline{\omega}_t$ is replaced by the corresponding expression for the endogenous price (see Section \ref{sec::marketclearingconditionandequilibriumprice}). First,  from the third equation in Equation \eqref{eq::couplingmfc} we obtain that
\begin{eqnarray*}
    -2 Y^{(2)}(s) = \frac{\alpha(s)}{\eta} + h,
\end{eqnarray*}
i.e., $-2 Y^{(2)}(s)$ is equal to the marginal abatement cost $C^{'}(\alpha(t))$; see Equation \eqref{eq::costabatement}. Plugging it into the last equation in Equation \eqref{eq::couplingmfc}, leads us to
\begin{equation*}
    \beta(s) = \nu \left(\frac{\alpha(s)}{\eta}+h-\overline{\omega}_t\right).
\end{equation*}
Whence, the firm buys (resp. sells) if its marginal abatement cost is higher (resp. lower) than the market price, in agreement with the economic intuition. Similarly, from the second equation in Equation \eqref{eq::costabatement} we obtain that
\begin{equation*}
    - Y^{(1)}(s) = \frac{1}{2 \kappa_g}\left(c_{2,1} + 2 c_{2,2} K^{g}(s)\right),
\end{equation*}
i.e., $-Y^{(1)}(s)$ is proportional to the marginal level of green capital cost $(C^{g})^{'}(K^{g})$; see Equation \eqref{eq::productioncosts}.   Now plugging the previous term into the expression for $K^{f}(s)$, leads us to:
\begin{eqnarray*}
    K^{f}(s) = \frac{1}{2 c_{1,2}}\left(\frac{\kappa_f}{k_g}(C^{g})^{'}(K^{g}) - \kappa_f C^{'}(\alpha(t))-c_{1,1}\right)
\end{eqnarray*}
Whence, the firms decide the fossil fuel level of capital $K^{f}(s)$ by roughly (see the subsequent discussion) comparing the marginal level of green capital cost $(C^{g})^{'}(K^{g})$ with the sum of the marginal abatement cost $C^{'}(\alpha(t))$ and the baseline cost $c_{1,1}$, again in line with economic intuition.   
\section{A mean field control approximation with common noise for the $N$ player game.}\label{sec::mfcapproximation}
Interestingly, the MFG in Equations \eqref{eq::statedynamicsmfglimit}-\eqref{eq::functionalmfglimit} is equivalent to \textit{a}\footnote{See Remark \ref{rmk::remarkmfgmfc}.} mean field type control problem (see \cite{graber2016linear}, Proposition 3.3 and Corollary 3.4), whose general formulation and resolution is given in Appendix \ref{app::LQMFC}. In our case, the state dynamics is as in Equation \eqref{eq::statedynamicsmfglimit} with the associated matrices as in Equations \eqref{eq:matricesdsexogenousprice}-\eqref{eq:matricesdW0exogenous}, and the objective functional $\mathcal{J}^{LQ}_{x,t}$ is given by:
\begin{equation}\label{eq::functionalmfclimit}
 \begin{split}
        \mathcal{J}^{LQ}(v;\overline{K})&= \mathbb{E}\Bigg[\int_{0}^{T}- p^{\overline{K}}(t) \overline{K} (t) + \beta(t)\overline{\omega}_t+\frac{1}{2\nu}(\beta(t))^2 + C(\alpha(t))\\
                                &+ C^{f}(K^{f}(t))+C^{g}(K^{g}(t))\,dt + \lambda(\widetilde{X}(T))^2\Bigg].
    \end{split}
\end{equation}
Whence, the matrices characterizing the cost functional are in Equations \eqref{eq:matricesexogenouspricefunctional}--\eqref{eq:matricesexogenouspricefunctionaltwo}. 
In particular, the following proposition, analogous to Proposition \ref{prop::propositionexistenceanduniquenessmfg}, holds true (see Proposition 2.4 in \cite{graber2016linear}).

\begin{proposition}\label{prop::propositionexistenceanduniquenessmfc}
Suppose $v$ is an optimal control minimizing the objective functional $\mathcal{J}^{LQ}(v;\overline{K})$ in Equation \eqref{eq::functionalmfclimit} with corresponding trajectory $X(s) = (K(s), \widetilde{X}(s))$ solution of Equation \eqref{eq::statedynamicsmfglimit} with control $v$.    Then there exists a unique solution $(Y, Z, Z_0)\in \mathbb{S}^2((\mathcal{F}_t);\mathbb{R}^2) \times \mathbb{H}^2((\mathcal{F}_t);\mathbb{R}^{2 \times 3}) \times \mathbb{H}^2((\mathcal{F}_t);\mathbb{R}^{2 \times 2})$ of the following BSDE, with $s \in [0,T]$
\begin{equation}\label{eq::bsdemfc}
    \begin{cases}
    \begin{split}
        dY^{(1)}(s)=&-\left(\textcolor{black}{-\delta Y^{(1)}(s)\,ds}+\sigma Z_1^{(1)}(s) + b K(s) + \frac{b \gamma}{2}\overline{K}(s)-\frac{a}{2}\right)\,ds\\
                   &+ \sum_{j=1}^{3}Z_j^{(1)}dW^{j}(s)+ \sum_{j=1}^{2} Z_{0,j}^{(1)}dW^{0,j}(s),\quad Y^{(1)}(T)=0;\\
        dY^{(2)}(s)=& \sum_{j=1}^{3}Z_j^{(2)}dW^{j}(s)+ \sum_{j=1}^{2} Z_{0,j}^{(2)}dW^{0,j}(s),\quad Y^{(2)}(T)=\lambda(\widetilde{X}(T))^2
    \end{split}
    \end{cases}
\end{equation}
satisfying the coupling condition, with $s \in [0,T]$, \textit{a.s.},
\begin{equation}\label{eq::couplingmfc}
\begin{cases}
 K^{f}(s) &= - \frac{\kappa_f}{c_{1,2}} Y^{(1)}(s) + \frac{\kappa_f}{c_{1,2}} Y^{(2)}(s) - \frac{c_{11}}{2 c_{1,2}},\\
 K^{g}(s) &= - \frac{\kappa_g}{c_{2,2}} Y^{(1)}(s) - \frac{c_{2,1}}{2 c_{2,2} },\\
\alpha(s) &=- 2\eta Y^{(2)}(s) - \eta h,\\
 \beta(s) &= - 2\nu Y^{(2)}(s) - \nu \overline{\omega}(s).
\end{cases}
\end{equation}
Conversely, suppose $(X, v, Y, Z, Z_0) \in \mathbb{S}^2((\mathcal{F}_t);\mathbb{R}^2)  \times 
\mathbb{H}^2((\mathcal{F}_t);\mathbb{R}^{4})\times
\mathbb{S}^2((\mathcal{F}_t);\mathbb{R}^2)  \times \mathbb{H}^2((\mathcal{F}_t);\mathbb{R}^{2 \times 3}) \times \mathbb{H}^2((\mathcal{F}_t);\mathbb{R}^{2 \times 2})$ is a solution to the forward-backward system \eqref{eq::statedynamicsmfglimit}, \eqref{eq::bsdemfc} and coupling condition \eqref{eq::couplingmfc}. Then $v$ is the optimal control minimizing $\mathcal{J}^{LQ}(v;\overline{K})$, and $X(s)$ is the optimal trajectory.
\end{proposition}

\begin{proof}

The proof follows the steps of Theorem 1.59 in \cite{delarue2018probabilistictwo}. First, let us write the MFC problem in Equations \eqref{eq::statedynamicsmfglimit} and \eqref{eq::functionalmfclimit} in matrix notation as in Appendix \ref{app::LQMFC}; matrices are given in \eqref{eq::statedynamicsmfglimit}-\eqref{eq::functionalmfglimit} and \eqref{eq:matricesdsexogenousprice}-\eqref{eq:matricesdW0exogenous}. 
\begin{equation}\label{eq::functionalmfclimitmatrixnotation}
    \begin{split}
        \mathcal{J}^{LQ}_{t,x}(v; \overline{X}) &= \mathbb{E}\Bigg[\int_{t}^{T}\langle Q X(s), X(s)\rangle+\langle \overline{Q}\,\overline{X}(s), \overline{X}(s)\rangle +\langle R\,v(s), v(s)\rangle\\
                                &+ 2\langle q, X(s)\rangle + \langle r(s), v(s)\rangle\,ds + \langle H, X(T), X(T)\rangle\Bigg].
    \end{split}
\end{equation}
\begin{equation}\label{eq::dynamicsmfclimitmatrixnotation}
    \begin{split}
       dX(s) &= (A_0(s) + B v(s))\,ds\textcolor{black}{+AX(s)\,ds}\\
             &+ C_{0,2}\,dW^2(s) +C_{0,3}\,dW^3(s) + C_1 X(s)\,dW^1(s)\\
             &+F_{0,1}\,dW^{0,1}(s) + F_{0,2}\,dW^{0,2}(s).
    \end{split}
\end{equation}
In the previous equations, as usual, $\bar{X}(s)$ denotes the conditional expectation of $X(s)$ given $\overline{\mathcal{F}}_s^{0}$. Second, let $(v, X)$ be the optimal pair for the MFC problem, $v^{h}$ the control $v^{h} = v + h \tilde{v}$, and $X^{h}$ the trajectory associated to $v^{h}$. Then, let us define the so-called variation process $(V(s))$ as the solution of the following SDE:
\begin{equation*}
    dV(s) =\textcolor{black}{AV(s)\,ds}+ B v(s)\,ds + C_1 V(s) dW^1(s),\quad s \in [0,T]
\end{equation*}
By repeating the computations in \cite{delarue2018probabilisticone}, Lemma 6.10, we have that the following limit holds true 

\begin{equation*}
    \lim_{\epsilon \rightarrow 0} \mathbb{E}\left[\sup_{s \in [0,T]}\Big\vert \frac{X^h(s)-X(s)}{\epsilon}-V(s)\Big\vert^2\right]=0.
\end{equation*}
We now observe that
\begin{equation*}
    \begin{split}
        \mathcal{J}^{LQ}_{t,x}(v^h); \overline{X})&=\mathbb{E}\Bigg[\int_{t}^{T}\Big[ \langle Q\,X^h(s), X^h(s)\rangle+\langle\overline{Q}\,\overline{X}^h(s),\overline X^h(s)\rangle+\langle R v^h(s),v^h(s)\rangle \\
        &\hspace{2cm}+2\langle q ,X^h(s)\rangle+2\langle r(s), v^h(s)\rangle
        \Big]\,ds + \langle H X^h(T), X^h(T)\rangle \Bigg],
    \end{split}
\end{equation*}
from which the Gateaux derivative of $\mathcal{J}^{LQ}_{t,x}(v^h; \overline{X})$ in the direction $h$ reads as
\begin{equation*}
    \begin{split}
    &\frac{d}{dh}\mathcal{J}^{LQ}_{t,x}(v^h; \overline{X})\\
    &=2\mathbb{E}\Bigg[\int_t^T\langle QX^h,V\rangle+\langle \bar Q\bar X^h,\bar V\rangle+\langle Rv^h,\tilde v\rangle+\langle q,V\rangle+\langle r,\tilde v\rangle \,ds +\langle HX^h_T,V_T\rangle\Bigg].
    \end{split}
\end{equation*}

Whence, the optimal condition for $(v, X)$
\begin{equation*}
    \frac{d}{dh}\mathcal{J}^{LQ}_{t,x}(v^h; \overline{X})\Big\vert_{h=0}=0
\end{equation*}
is
\begin{equation}\label{eq::optimalitycondition}
    \mathbb{E}\Bigg[\int_t^T\langle QX^h,V\rangle+\langle \bar Q\bar X^h,\bar V\rangle+\langle Rv^h,\tilde v\rangle+\langle q,V\rangle+\langle r,\tilde v\rangle \,ds +\langle HX^h_T,V_T\rangle\Bigg]=0.
\end{equation}
By \cite{buckdahn2009mean,buckdahn2009mean2}, we know that the BSDE in Equation \eqref{eq::bsdemfc} admits a unique solution. Then, by applying Ito's formula to the process $(\langle Y(t), V(t)\rangle)$,
\begin{equation*}
    \begin{split}
       \langle Y(t), V(t)\rangle&=\int_t^T-\left(\textcolor{black}{\langle A^T Y(s), V(s) \rangle}+\langle C_1^T Z_1(s), V(s)\rangle + \langle Q X(s), V(s)\rangle  \right)\,ds\\
       &-\int_t^T\langle \left(\overline{Q}\,\overline{X}(s), V(s)\rangle+\langle q, V(s)\rangle\right)\,ds
+\int_t^T \langle B^T Y(s), \tilde{v}(s)\rangle + \langle C_1^T Z_1(s), V(s)\rangle\,ds\\
&+\textcolor{black}{\int_t^T \langle AV(s),Y(s)\rangle\,ds}+\text{Martingale}.\\
    \end{split}
\end{equation*}
By taking the expectation on both sides, by recalling that $\mathbb{E}\left[\langle Q \overline{X}(s), V(s)\rangle\right] = \mathbb{E}\left[\langle Q \overline{X}(s), \overline{V}(s)\rangle\right]$ and $Y(T) =H X(T)$, we have
\begin{equation*}
\begin{split}
    \mathbb{E}\Bigg[&\int_{t}^{T}\left(\langle Q X(s), V(s)\rangle + \langle \overline{Q}\overline{X}(s), V(s)\rangle+
        \langle q, V(s)\rangle\right)-\langle B^T Y(s), \tilde{v}(s)\rangle\,ds\\
                    &+ \langle H X(T), V(T)\rangle\Bigg]=0.
\end{split}
\end{equation*}
By using the optimally condition in Equation \eqref{eq::optimalitycondition} and the arbitrariness of $\tilde v$ we obtain the coupling condition \eqref{eq::couplingmfc}.\\
\indent To prove the converse, it is sufficient to observe that the functional $\mathcal{J}^{LQ}_{t,x}(v; \overline{X})$ is strictly convex (see, again, \cite{graber2016linear}, Lemma 2.2, or \cite{yong2013linear}, Proposition 2.7). Then, given a solution $(X,v,Y,Z,Z_0)$ to the forward-backward system \eqref{eq::statedynamicsmfglimit}, \eqref{eq::bsdemfc}, we know that the Gateaux derivative of $\mathcal{J}^{LQ}_{t,x}(v; \overline{X})$ at $v$ is zero, which means that $v$ is a minimizer.

\end{proof}

\indent   Before proceeding,  we make the following remark.  
\begin{remark}\label{rmk::remarkmfgmfc}
The MFC problem in Equations \eqref{eq::statedynamicsmfglimit} and \eqref{eq::functionalmfclimit} is not the MFC problem version of \eqref{eq::statedynamicsmfglimit} and \eqref{eq::functionalmfglimit}. Indeed, the latter would have the following objective functional
\begin{equation}\label{eq::functionalmfclimitmatrixnotationtrue}
    \begin{split}
        \mathcal{J}^{LQ}_{t,x}(v; \overline{X}) &= \mathbb{E}\Bigg[\int_{t}^{T}\langle Q X(s), X(s)\rangle+\boldsymbol{2}\langle \overline{Q}\,\overline{X}(s), \overline{X}(s)\rangle +\langle R\,v(s), v(s)\rangle\\
                                &+ 2\langle q, X(s)\rangle + \langle r(s), v(s)\rangle\,ds + \langle H, X(T), X(T)\rangle\Bigg].
    \end{split}
\end{equation}
\end{remark}
\noindent   Nicely, by using the derivations in Appendix \ref{app::LQMFC}, we can explicitly write down the solution of problem \eqref{eq::statedynamicsmfglimit}--\eqref{eq::functionalmfclimit} in terms of the following system of Riccati equations, where matrices $C_1, Q, B, R, H$ and vectors $r(t), q, A_0(t)$ are defined in \eqref{eq:matricesdsexogenousprice}, \eqref{eq:matricesdWexogenousprice}, \eqref{eq:matricesdW0exogenous}, \eqref{eq:matricesexogenouspricefunctional}, \eqref{eq:matricesexogenouspricefunctionaltwo}.
\begin{equation}\label{eq::riccati_P}
\begin{cases}
        \overset{\cdot}{P}(t) + C_1^T P(t) C_1 + Q - P(t) B R^{-1} B^T P(t) \textcolor{black}{+P(t)A+A^TP(t)}= 0;\\
        P(T) = H.\\
\end{cases}
\end{equation}
\begin{equation}\label{eq::riccati_PI}
    \begin{cases}
        \overset{\cdot}{\Pi}(t) + C_1^T P(t) C_1 + (Q + \bar{Q}) - \Pi(t) B R^{-1} B^T \Pi(t)\textcolor{black}{+\Pi(t) A+A^T\Pi(t)} = 0;\\
        \Pi(T) = H.\\
    \end{cases}
\end{equation}
\begin{equation}\label{eq::riccati_phi}
\begin{cases}
    \overset{\cdot}{\phi}(t) - \Pi(t) B R^{-1} r(t) + q +\Pi(t) A_0(t) - \Pi(t) B R^{-1} B^T \phi(t) \textcolor{black}{+A^T\phi(t)}=0; \\
    \phi(T)=0.
    \end{cases}
\end{equation}
\noindent   In particular, Theorem 2.6 in \cite{graber2016linear} ensures that the previous Equations \eqref{eq::riccati_PI}--\eqref{eq::riccati_phi} admit a unique solution, where $P$, $\Pi$ are two deterministic processes in $\mathcal{S}^2$, whereas $\phi$ is a deterministic process in $\mathbb{R}^2$. In addition, both the optimal control and the associated optimal trajectory for the problem \eqref{eq::statedynamicsmfglimit} and \eqref{eq::functionalmfclimitmatrixnotation} are expressed in terms of the solutions $P$, $\Pi$ and $\phi$, $t \in [0,T]$:
\begin{equation}\label{eq::optimalcontrolriccati}
    v(t) = - R^{-1} B^T P(t) (X(t)-\overline{X}(t)) - R^{-1} (B^T \Pi(t) \overline{X}(t) + r(t) + B^T \phi(t)).
\end{equation}
\begin{equation}\label{eq::optimaltraject}
\begin{split}
    dX(t)   &=(A_0(s)-B R^{-1} B^T P(t) (X(t) - \overline{X}(t)))\,ds\\
            &- B R^{-1} (B^T \Pi(t) \overline{X}(t) + B^T \phi(t) + r(t)))\,ds\\
            &+ C_1 X(s)\,dW^1(s) + C_{0,2} dW^2(s) + C_{0,3} dW^3(s)\\
            &+F_{0,1}dW^{0,1}(s)+F_{0,2}dW^{0,2}(s)
\end{split}
\end{equation}
\noindent   Proposition \ref{prop::propositionexistenceanduniquenessmfc} ensures that there exists a unique adapted solution of the forward-backward system \eqref{eq::statedynamicsmfglimit}, \eqref{eq::bsdemfc}, and therefore such a solution coincides with the solution constructed with the Riccati Equations \eqref{eq::riccati_PI}--\eqref{eq::riccati_phi}.\\
\indent We make now the following observation regarding the so-called Price of Anarchy (PoA, henceforth).
\begin{observation}[PoA]
The objective functional $\mathcal{J}^{NE}(v;\overline{K})$ (Equation \eqref{eq::functionalmfglimit}) and $\mathcal{J}^{LQ}(v;\overline{K})$ (Equation \eqref{eq::functionalmfclimit}) are not precisely the same, even though $v$ solves both a fixed point Nash equilibrium and a mean field type control problem. The difference between $\mathcal{J}^{NE}(v;\overline{K})$ and $\mathcal{J}^{LQ}(v;\overline{K})$ is called PoA since it represents the added aggregate cost of allowing all players to choose their optimal strategy independently.    
\begin{equation*}
\begin{split}
    \mathcal{J}^{NE}(v;\overline{K})-\mathcal{J}^{LQ}(v;\overline{K}) 
            &= b \frac{\gamma}{2}\mathbb{E}\left[\int_{t}^{T}(\mathbb{E}[K(s)|\mathcal{\overline{F}}^{0}_{s}])^2\,ds\right]
\end{split}
\end{equation*}
It is strictly positive as soon as the conditional expectation of the equilibrium capital level is different from zero, and strictly increasing in $\gamma$, reflecting the fact that greater competition yields a greater cost of non-cooperation. 
\end{observation}

\section{Market clearing condition and equilibrium price.}\label{sec::marketclearingconditionandequilibriumprice}
We start this section with the definition of market equilibrium for the finite player game. 
\begin{definition}\label{def::market_equilibrium_finite_player}
    For the $N$ player game a market equilibrium is, for every $t \in [0,T]$, a $N$-dimensional vector $\mathbf{\beta}^{\star, N}(t)=(\beta^{\star, 1}(t), \ldots, \beta^{\star,N}(t))$ such that: (1) each $\beta^{\star,i}(t) \in \mathcal{H}_1^{N}$; (2) $\beta^{\star,i}(t)$ is the $i^{th}$ component of the $\epsilon$-Nash equilibrum for the $N$-player game (see Definition \ref{def::nashequilibria}), and (3) the asymptotic market clearing condition $\lim_{N \rightarrow \infty} \sum_{i=1}^{N}\beta^{\star,i}(t) = 0$ holds true.
\end{definition}

At this point,  we observe that because our MFG is equivalent to an optimal control problem, the mean field Nash equilibrium is an $\epsilon$-Nash equilibrium for the $N$-player game, in the sense of Definition \ref{def::nashequilibria}; see Theorem 3.6 in \cite{graber2016linear}.   By using Proposition \ref{prop::propositionexistenceanduniquenessmfc}, the optimal trading rate $\beta^{\star,i}(t)$ of each firm is given by:  
\begin{equation*}
    \beta^{\star,i}(t) = -2 \nu Y^{(2),i}(t)-\nu\overline{\omega}(t),\quad t \in [0,T].
\end{equation*}
Because we model the trading mechanism as part of the firms' decision problem, the equilibrium (market-clearing) price of emission allowances emerges endogenously. In the present situation, Equation \eqref{eq::marketclearingcondition} is equivalent to the following condition:     
\begin{equation}\label{eq::marketclearingcondition}
    \frac{1}{N}\sum_{i=1}^{N}\beta^{\star,i}(t) = \frac{1}{N}\sum_{i=1}^{N}\left(-2 \nu Y^{(2),i}(t)-\nu\overline{\omega}(t)\right)=0,
\end{equation}
from which we have that
\begin{equation}\label{eq::equilibriumprice}
    \overline{\omega}(t) = - \frac{2}{N}\sum_{i=1}^{N}Y^{(2),i}(t).
\end{equation}
The previous solution is of course inconsistent with our standing assumption that the price process $(\overline{\omega}_t)$ is a $(\overline{\mathcal{F}}^0)$-adapted process. However, we can argue as in \cite{fujii2022mean}, Page 267, and expect that in the large-$N$ limit, the market price of allowances may be given by 
$\overline{\omega}_t=-2\mathbb{E}[Y^{(2)}(t)|\overline{\mathcal{F}}^0_t]$.  Therefore, we need to consider a different BSDE system with respect to the one in Proposition \ref{prop::propositionexistenceanduniquenessmfc}. More precisely, the coupling condition for $\beta(s)$ in Equation \eqref{eq::couplingmfc} is now given by:
\begin{equation*}
    \beta(s) = - 2 \nu Y^{(2)}(s) + 2 \nu \mathbb{E}[Y^{(2)}(s)|\overline{\mathcal{F}}_t^{0}],
\end{equation*}
which leads to the following matrix-based representation of the coupling condition 
\begin{equation}\label{eq::couplingmfcendogenous}
    v(t) = - R^{-1} (B^T\,Y(t) + \tilde{r} + D\,\overline{Y}(s)),\quad t \in [0,T],\,\,a.s.
\end{equation}
where 
\begin{equation}\label{eq::matricescouplingnew}
    \tilde{r}
    =
    \begin{bmatrix}
        \frac{c_{1,1}}{2}\\
        \frac{c_{2,1}}{2}\\
        \frac{h}{2}\\
        0
    \end{bmatrix},\quad
    D 
    =
    \begin{bmatrix}
        0 & 0 \\
        0 & 0 \\
        0 & 0 \\
        0 &-\frac{1}{2}
    \end{bmatrix}.
\end{equation}
\noindent   Equations \eqref{eq::statedynamicsmfglimit} and \eqref{eq::bsdemfc}, instead, remain unchanged; notice that now $v(t)$ is as in Equation \eqref{eq::couplingmfcendogenous}. Their matrix-based representation is given by the following equations: 
\begin{equation}\label{eq::statedynamicsmfglimitmatrixform}
    \begin{split}
     dX(s) &= (A_0(s) \textcolor{black}{+AX(s)}+B v(s))\,ds\\
             &+ C_{0,2}\,dW^2(s) +C_{0,3}\,dW^3(s) + C_1 X(s)\,dW^1(s)\\
             &+F_{0,1}\,dW^{0,1}(s) + F_{0,2}\,dW^{0,2}(s),\quad X(0) = x_0.
    \end{split}
\end{equation}
\begin{equation}\label{eq::bsdematrixform}
    \begin{split}
     dY(s)  &=- (\textcolor{black}{A^TY(s)}+C_1^T Z_1 (s) + Q X(s) + \overline{Q}\,\overline{X}(s) + q)\,ds\\
            &+\sum_{j=1}^{3}Z_j(s) dW^{j}(s)+ \sum_{j=1}^{2} Z_{0,j}(s) dW^{0,j}(s),\quad Y(T) = H X(T).
    \end{split}
\end{equation}
We now state and prove the following short-term existence result.
\begin{theorem}\label{thm::localexistence}
There exists some constant $\tau>0$ which depends only on the matrices $A_0(s), B, C_{0,2}, C_{0,3}, C_1, F_{0,1}, F_{0,2}, Q, \overline{Q}, q$ such that for any $T \leq \tau$, there exists a unique strong solution $(X, Y, Z, Z_0) \in \mathbb{S}^2((\mathcal{F}_t);\mathbb{R}^2)  \times 
\mathbb{S}^2((\mathcal{F}_t);\mathbb{R}^2)  \times \mathbb{H}^2((\mathcal{F}_t);\mathbb{R}^{2 \times 3}) \times \mathbb{H}^2((\mathcal{F}_t);\mathbb{R}^{2 \times 2})$ to the FBSDE \eqref{eq::statedynamicsmfglimitmatrixform}--\eqref{eq::bsdematrixform}.
\end{theorem}
\begin{proof}
The proof is an adaption of the arguments used in \cite{delarue2018probabilisticone},   Theorem 4.24.   The main difference is that there exists a term involving $\mathbb{E}[Y^{(2)}(s)|\overline{\mathcal{F}}_t^{0}]$ via the coupling condition in \eqref{eq::couplingmfcendogenous}.    Let $\Phi$ be the map constructed in the following way.  For any element $(X,Y) \in \mathbb{S}^{2}((\overline{\mathcal{F}}_t);\mathbb{R}^2)$, let $(Y, Z)$ be the solution of the following BSDE, where $s \in (0,T]$: 
\begin{equation}\label{eq::BSDEprooflocalexistence}
    \begin{cases}
        dY(s) = -\left(\textcolor{black}{AY(s)}+C_1^T Z_1(s) + Q X(s) + \overline{Q}\,\overline{X}(s) + q\right)\,ds\\
        \quad\quad\quad\quad\quad\quad\quad\quad\quad\quad\quad\quad\quad+\sum_{j=1}^{3}Z_j(s)dW^{j}(s)+\sum_{j=1}^{2}Z_{0,j}(s)dW^{0,j}(s)\\
        Y(T) = H X(T)
    \end{cases}
\end{equation}
Notice that $X \in \mathbb{S}^2((\mathcal{F}_t);\mathbb{R}^2)$ and the pair $(Y,Z) \in \mathbb{S}^2((\mathcal{F}_t);\mathbb{R}^2) \times \mathbb{H}^2((\overline{\mathcal{F}}_t);\mathbb{R}^{2\times 3})$ are progressively measurable with respect the completion of the filtration generated by $(W(s)-W(t))_{s \in [t,T]}$ and $(W^{0}(s)-W^{0}(t))_{s \in [t,T]}$.    Then, we associate to the couple $(Y,Z)$ the solution $(\widetilde{X}(s))$ of the following SDE, where $s \in (0,T]$
\begin{equation}\label{eq::sdextildelocalexistence}
    \begin{split}
        d\widetilde{X}(s)&=\left(A_0(s)\textcolor{black}{+A\widetilde X(s)}-B R^{-1}(B^T Y(s) + \tilde{r} + D \overline{Y}(s)\right)\,ds\\
                         &+C_{0,2}dW^2(s)+C_{0,3}dW^3(s)+C_1 \widetilde{X}(s) dW^1(s)\\
                         &+F_{0,1}dW^{0,1}(s)+F_{0,2}dW^{0,2}(s),\quad\quad \widetilde{X}(0) = x_0.
    \end{split}
\end{equation}
The map $\Phi$ is given by $\Phi : X \rightarrow (Y,Z) \rightarrow \widetilde{X}$.  The aim is to show that $\Phi$ is a contraction for small $T$.  To this end, let $X^{1}$ and $X^2 \in \mathbb{S}^2((\mathcal{F}_t);\mathbb{R}^2)$ and denote by $(Y^1, Z^1)$ and $(Y^2,Z^2)$ be the associated solution of the BSDE in \eqref{eq::BSDEprooflocalexistence}.  In addition, set $\widetilde{X}^1 = \Phi(X^1)$ and $\widetilde{X}^2 = \Phi(X^2)$.   The fact that $\Phi$ is a contraction follows from the following standard estimates for SDEs and BSDEs:
\begin{equation*}\label{eq::boundSDEs}
    \begin{split}
        &\mathbb{E}\left[\sup_{s \in [0,T]}\Big|\widetilde{X}^1(s)-\widetilde{X}^2(s)\Big|\right]+\mathbb{E}\left[\sup_{s \in [0,T]}\Big|\overline{\widetilde{X}^1}(s)-\overline{\widetilde{X}^2}(s)\Big|\right]\\
        &\leq   C T \mathbb{E}\Bigg[\sup_{s \in [0,T]}\Big|Y^1(s)-Y^2(s)\Big|^2+\sup_{s \in [0,T]}\Big|\overline{Y^1}(s)-\overline{Y^2}(s)\Big|^2\\
        &+\sum_{j=1}^3\int_{0}^{T}\Big|Z_j^1(s)-Z_j^2(s)\Big|\,ds+\sum_{j=1}^2\int_{0}^{T}\Big|Z_{0,j}^1(s)-Z_{0,j}^2(s)\Big|\,ds\Bigg]\\
        &\leq C T \left( \mathbb{E}\Bigg[\sup_{s \in [0,T]}\Big|X^1(s)-X^2(s)\Big|^2\Bigg] + \mathbb{E}\Bigg[\sup_{s \in [0,T]}\Big|\overline{X^1}(s)-\overline{X^2}(s)\Big|^2\Bigg]\right)
    \end{split}
\end{equation*}
\end{proof}

\noindent Before proceeding, we make the following observation. Should the solution of the FBSDE \eqref{eq::statedynamicsmfglimitmatrixform}--\eqref{eq::bsdematrixform} be linked to that of some MFC problem, a term of the form $D\,\bar{v}(s)$ would be present in the state dynamics because of the presence of a term like $D \mathbb{E}[Y(t)|\overline{\mathcal{F}}_t^0]$ in the coupling condition. However, this is not the case for our FBSDE.\\
\indent The next theorem gives us the unique existence of solutions to the FBSDE \eqref{eq::statedynamicsmfglimitmatrixform}--\eqref{eq::bsdematrixform} for general $T$.

\begin{theorem} \textcolor{black}{Under the assumption that \[\left(\frac{\kappa_f^2}{c_{1,2}} + \frac{\kappa_g^2}{c_{2,2}}-\frac{\kappa_f\kappa_e}{c_{1,2}}\right)>0,\quad\left(2\eta+\nu+\frac{\kappa_e^2}{c_{1,2}}-\frac{\kappa_f\kappa_e}{c_{1,2}}\right)>0,\]}
there exists a unique strong solution $(X, Y, Z, Z_0) \in \mathbb{S}^2((\mathcal{F}_t);\mathbb{R}^2)  \times 
\mathbb{S}^2((\mathcal{F}_t);\mathbb{R}^2)  \times \mathbb{H}^2((\mathcal{F}_t);\mathbb{R}^{2 \times 3}) \times \mathbb{H}^2((\mathcal{F}_t);\mathbb{R}^{2 \times 2})$ to the FBSDE \eqref{eq::statedynamicsmfglimitmatrixform}--\eqref{eq::bsdematrixform}.
\end{theorem}
\begin{proof}
The proof hinges on the continuation method of \cite{peng1999fully} and reduces to verify that assumption (H2.1) in the previous paper holds true for our system, in expectation. Notice that in our case their $2\times 2$ full-rank matrix $G$ is the identity matrix. In order to facilitate the comparison, we rewrite the FBSDE \eqref{eq::statedynamicsmfglimitmatrixform}--\eqref{eq::bsdematrixform} in terms of the following functional $b, f, \sigma, \sigma_0, \Phi$ and of a vector $\theta$ where all the static parameters are collected
\begin{equation*}
    \begin{split}
        &\textcolor{black}{b(s,Y(s),X(s),\tilde{a}(s),\theta):=\left(A_0(s)+AX(s)-B R^{-1}(B^{T}Y(s)+\tilde{r}+D\overline{Y}(s)\right)}\\
        &\textcolor{black}{f(s,Z_1(s),Y(s),X(s),\theta):=\left(A^TY(s)+C_1^T Z_1(s)+Q X(s) + \overline{Q} \overline{X}(s) + q\right)}\\
        &\Phi(X(T)):=H X(T),\\
        &\sigma(X^{(1)}(s), \theta):=
    \begin{bmatrix}
        \sigma X^{(1)}(s)&0&0\\
        0&-\sigma_1\sqrt{1-\rho^2}&-\sigma_2
    \end{bmatrix}\\
    &\sigma_0(\theta):=
    \begin{bmatrix}
            0&0\\
-\sigma_1\rho&\tilde{\sigma}_2\\    
    \end{bmatrix}\\
    &\theta = (\kappa_f, \kappa_g, c_{1,2}, c_{2,2}, \eta, \nu, c_{1,1}, c_{2,1}, h, \sigma_1, \rho, \sigma_2, \sigma, \tilde{\sigma}_2, b, \gamma, a, \lambda)
    \end{split}
\end{equation*}
as
\begin{equation*}\label{eq::fbsde_continuation}
    \begin{split}
        dX(s)&=\textcolor{black}{b(s,Y(s),X(s),\tilde{a}(s),\theta)}\,ds + \sigma(X^{(1)}(s), \theta)dW(s) + \sigma_0(\theta)dW^{0}(s).\\
        dY(s)&=\textcolor{black}{-f(s,Z_1(s),Y(s),X(s),\theta)}\,ds + Z(s)dW(s) + Z_0(s)dW^{0}(s),
    \end{split}
\end{equation*}
where $X(0)=x_0$ and $Y(T) = \Phi(X(T))$.  
\noindent We use the following notation:
\begin{equation*}
    u = \begin{bmatrix}
            x \\
            y \\
            z \\
            z_0
        \end{bmatrix},\quad\quad A(s,u) = \begin{bmatrix}
            - f\\
              b\\
        \sigma\\
        \sigma_0\\
        \end{bmatrix}(s,u).
\end{equation*}
Besides, for all pairs $(x, y, z, z_0), (x^{'}, y^{'}, z^{'}, z^{'}_0) \in \mathbb{L}^{2}(\mathcal{F};\mathbb{R}^2 \times \mathbb{R}^2 \times \mathbb{R}^{2\times 3}\times\mathbb{R}^{2\times 2})$, we denote by $\hat{x}=x-x^{'}$, $\hat{y}=y-y^{'}$, $\hat{z}=z-z^{'}$, and $\hat{z}_0=z_0-z_0^{'}$. We have:
\begin{equation}\label{eq:term_f}
\begin{split}
        &\textcolor{black}{\langle -f(s,z_1,y, x, \theta) - (-f(s, z_1^{'}, y',x^{'}, \theta)), x-x^{'}\rangle}\\
        &= - \sigma \hat{z}_1^{(1)} \hat{x}^{(1)} - b(1-\gamma)A_k^2 (\hat{x}^{(1)})^2 - \frac{b \gamma A_k^2}{2} \hat{\overline{x}}^{(1)} \hat{x}^{(1)}+\delta \hat y^{(1)}\hat x^{(1)}       
\end{split}
\end{equation}
\textcolor{black}{\begin{equation}\label{eq:term_b}
    \begin{split}
        &\langle b(s, y, \tilde{a}, \theta)-b(s, y^{'},\tilde{a}, \theta), y-y^{'} \rangle\\
        &=-\left(\frac{\kappa_f^2}{c_{1,2}} + \frac{\kappa_g^2}{c_{2,2}}\right)(\hat{y}^{(1)})^2 + 2 \frac{\kappa_f\kappa_e}{c_{1,2}}\hat{y}^{(1)}\hat{y}^{(2)}-\left(2(\eta+\nu)+\frac{\kappa_e^2}{c_{1,2}}\right)(\hat{y}^{(2)})^2 + \nu \hat{\overline{y}}^{(2)}  \hat{y}^{(2)}-\delta\hat{x}^{(1)}\hat{y}^{(1)}
    \end{split}
\end{equation}}
and 
\begin{equation}\label{eq:term_sigma}
\langle \sigma (x,\theta)-\sigma(x',\theta),z\rangle= \sigma \hat{x}^{(1)} \hat{z}_1^{(1)};
\end{equation}
notice that $\sigma_0$ does not lead to any contribution since it is state-independent. By combining \eqref{eq:term_f}, \eqref{eq:term_b}, \eqref{eq:term_sigma}, we get
\textcolor{black}{\begin{equation*}
\begin{split}
&\langle A(s, u)- A(s, u'),u-u'\rangle=\\
&- b(1-\gamma)A_k^2 (\hat{x}^{(1)})^2 - \frac{b \gamma A_k^2}{2} \hat{\overline{x}}^{(1)} \hat{x}^{(1)}-\left(\frac{\kappa_f^2}{c_{1,2}} + \frac{\kappa_g^2}{c_{2,2}}\right)(\hat{y}^{(1)})^2 + 2 \frac{\kappa_f\kappa_e}{c_{1,2}}\hat{y}^{(1)}\hat{y}^{(2)}\\
&-\left(2(\eta+\nu)+\frac{\kappa_e^2}{c_{1,2}}\right)(\hat{y}^{(2)})^2 + \nu \hat{\overline{y}}^{(2)}  \hat{y}^{(2)}.  
 \end{split}
\end{equation*}}
\noindent Now, by taking the expectation, by the law of total expectation and Jensen inequality:
\textcolor{black}{
\begin{eqnarray*}\label{eq:monotonicity}
&&\mathbb{E}\left[\langle A(s, u)- A(s, u'),u-u'\rangle\right]\leq - b(1-\gamma)A_k^2 \mathbb{E}[(\hat{x}^{(1)})^2]-\frac{b\gamma A_k^2}{2} \mathbb{E}[(\hat{x}^{(1)})^2] \\
&&- \left(\frac{\kappa_f^2}{c_{1,2}} + \frac{\kappa_g^2}{c_{2,2}}-\frac{\kappa_f\kappa_e}{c_{1,2}}\right)\mathbb{E}[(\hat{y}^{(1)})^2]\\
&&\quad-\left(2\eta+\nu+\frac{\kappa_e^2}{c_{1,2}}-\frac{\kappa_f\kappa_e}{c_{1,2}}\right) \mathbb{E}[(\hat{y}^{(2)})^2]\\
&\leq& - \left(\frac{\kappa_f^2}{c_{1,2}} + \frac{\kappa_g^2}{c_{2,2}}-\frac{\kappa_f\kappa_e}{c_{1,2}}\right)\mathbb{E}[(\hat{y}^{(1)})^2]\\
&&\quad-\left(2\eta+\nu+\frac{\kappa_e^2}{c_{1,2}}-\frac{\kappa_f\kappa_e}{c_{1,2}}\right) \mathbb{E}[(\hat{y}^{(2)})^2]\\
&\leq& - \min\left(\left(\frac{\kappa_f^2}{c_{1,2}} + \frac{\kappa_g^2}{c_{2,2}}-\frac{\kappa_f\kappa_e}{c_{1,2}}\right),\left(2\eta+\nu+\frac{\kappa_e^2}{c_{1,2}}-\frac{\kappa_f\kappa_e}{c_{1,2}}\right)\right)\mathbb{E}[|\hat{y}|^2]\\
\end{eqnarray*}}
In addition, we have:
\begin{equation*}
    \langle \Phi(x_T)-\Phi(x_T^{'}), \hat{x}\rangle = \lambda (\hat{x}^{(1)})^2 \geq 0.
\end{equation*}
\noindent   In particular, the monotone conditions in (H2.3) in \cite{peng1999fully} hold with $\beta_1 = 0$ and \textcolor{black}{$\beta_2 = \min\left(\left(\frac{\kappa_f^2}{c_{1,2}} + \frac{\kappa_g^2}{c_{2,2}}-\frac{\kappa_f\kappa_e}{c_{1,2}}\right),\left(2\eta+\nu+\frac{\kappa_e^2}{c_{1,2}}-\frac{\kappa_f\kappa_e}{c_{1,2}}\right)\right)$}. Now, the continuation method hinges on the following steps. First, one has to introduce a family of FBSDE indexed by a parameter $\varrho \in [0,1]$,
\begin{equation}
    \begin{split}
        &dx_t^{\varrho} = [-(1-\varrho)\beta_2(y_t^{\varrho})+\varrho b(t, u_t^{\varrho}, \tilde{a}(t),\theta)+\phi_t]\,dt\\
        &+[ \varrho \sigma(t, u_t) + \psi_t]dW(t)+\sigma_0(\theta)\,dW^0(t)\\
        &dy_t^{\varrho} = - [ \varrho f(t, u_t^{\varrho})+\gamma_t]\,dt + z_t^{\varrho}dW(t) + z_{0,t}^{\varrho}\,dW^0(t),\\
        &x_0^{\varrho}=x_0, \quad y_T^{\varrho} = \varrho \Phi(x_T^{\varrho}) + (1-\varrho) x_T^{\varrho} + \xi,
    \end{split}
\end{equation}
where $\phi, \psi$ and $\gamma$ are given processes in $\mathbb{H}^2((\mathcal{F}_t);\mathbb{R}^2)$ and $\mathbb{H}^2((\mathcal{F}_t);\mathbb{R}^{2\times 3})$, respectively, and $\xi \in \mathbb{L}^2(\mathcal{F}_T;\mathbb{R}^2)$. Notice that the previous system is constructed in such a way that for $\rho=0$ its solution is straightforward, whereas for $\rho=1$ it implies the existence of a unique strong solution to the FBSDE \eqref{eq::statedynamicsmfglimitmatrixform}--\eqref{eq::bsdematrixform}. In particular, the monotone conditions above allow to extend the existence from $\varrho=0$ to $\varrho=1$. The proof follows directly from the one of Theorem 2.2 in \cite{peng1999fully}, and it is therefore omitted.
\end{proof}
It is important to note that the existence of a unique strong solution to the FBSDE \eqref{eq::statedynamicsmfglimitmatrixform}--\eqref{eq::bsdematrixform} holds for every level of competition $\gamma \in [0,1]$.

\begin{corollary}\label{prop::propositionestimates}
    The solution $(X, Y, Z^{0}, Z)$ to the FBSDE \eqref{eq::statedynamicsmfglimitmatrixform}--\eqref{eq::bsdematrixform} satisfies the following estimate
    \begin{equation*}
        \mathbb{E}\left[\sup_{t \in [0,T]}|X(t)|^2+\sup_{t \in [0,T]} |Y(t)|^2+\sum_{j=1}^3\int_{0}^{T}|Z_j(t)|^2\,dt + \sum_{j=1}^2\int_{0}^{T}|Z_{0,j}(t)|^2\,dt\right] \leq C,
    \end{equation*}
    where $C$ is a constant depending only on $T$ and on the matrices of the system $A_0(s)$, $B$, $C_{0,2}$,$C_{0,3}$, $C_1$, $F_{0,1}$, $F_{0,2}$, $Q$, $\overline{Q}$, $q$.   
\end{corollary}
\begin{proof}
By applying Ito's formula to $|Y(t)|^2$ we obtain:
\begin{equation}\label{eq::equationItoProof}
    \begin{split}
        &\mathbb{E}[|Y(t)|^2] + \mathbb{E}\left[\sum_{j=1}^{3}\int_{t}^{T}|Z_j(s)|^2\,ds\right] + \mathbb{E}\left[\sum_{j=1}^2\int_{t}^{T}|Z_{0,j}(s)|^2\,ds\right]\\
        &\leq \mathbb{E}\left[H |X(T)|^2\right] + \epsilon\,\mathbb{E}\left[\int_{t}^{T}|Z_1(s)|^2\,ds\right] \\
        &+ C \mathbb{E}\left[\int_{t}^{T}|Y(s)|^2\,ds+\int_{t}^{T}|X(s)|^2\,ds+\int_{t}^{T}|\overline{X}(s)|^2\,ds\right].
    \end{split}
\end{equation}
Then, by choosing $\epsilon>0$ small, say $\epsilon<1$, and applying Gronwall's inequality, we obtain
\begin{equation*}
    \begin{split}
        &\mathbb{E}[|Y(t)|^2] + \mathbb{E}\left[\sum_{j=1}^{3}\int_{t}^{T}|Z_j(s)|^2\,ds\right] + \mathbb{E}\left[\sum_{j=1}^2\int_{t}^{T}|Z_{0,j}(s)|^2\,ds\right]\\
        &\leq \mathbb{E}\left[H |X(T)|^2\right] + C \mathbb{E}\left[\int_{t}^{T}|X(s)|^2\,ds+\int_{t}^{T}|\overline{X}(s)|^2\,ds\right].
    \end{split}
\end{equation*}
Now, by using again Ito's formula, a simple application of the Burkholder-Davis-Gundy inequality, Young's inequality, and the triangular inequality gives, for new constants $C_H$, $C_{Q}$, $C_{\overline{Q}}$,  $C>0$:
\begin{equation*}
    \begin{split}
        &\mathbb{E}\left[\sup_{t \in [0,T]}|Y(t)|^2 + \sum_{j=1}^{3}\int_{t}^{T}|Z_j(s)|^2\,ds + \sum_{j=1}^{2}\int_{t}^{T}|Z_{0,j}(s)|^2\,ds\right]\\
        &\leq C_H \mathbb{E}\left[|X(T)|^2+\int_{0}^{T}|Y(s)|^2\,ds\right]\\
        &+ \mathbb{E}\left[\epsilon\int_{0}^{T}|Z(s)|^2\,ds+C_Q\int_{0}^{T}|X(s)|^2\,ds+C_{\overline{Q}}\int_{0}^{T}|\overline{X}(s)|^2\,ds\right]\\
        &+\epsilon_1\mathbb{E}\left[\sup_{t \in [0,T]}|Y(t)|^2\right] + C \mathbb{E}\left[\sum_{j=1}^{3}\int_{0}^{T}|Z_j(s)|^2\,ds\right]\\
        &+\epsilon_2\mathbb{E}\left[\sup_{t \in [0,T]}|Y(t)|^2\right] + C \mathbb{E}\left[\sum_{j=1}^{2}\int_{0}^{T}|Z_{0,j}(s)|^2\,ds\right]\\
    \end{split}
\end{equation*}
Choosing now $\epsilon_1, \epsilon_2$ such that $\epsilon_1+\epsilon_1<1$ and using Equation \eqref{eq::equationItoProof}, we get for a new constant $C>0$:
\begin{equation}\label{eq::equationItoProof1}
    \begin{split}
        &\mathbb{E}\left[\sup_{t \in [0,T]}|Y(t)|^2 + \sum_{j=1}^{3}\int_{t}^{T}|Z_j(s)|^2\,ds + \sum_{j=1}^{2}\int_{t}^{T}|Z_{0,j}(s)|^2\,ds\right]\\
        &\leq C\mathbb{E}\left[|X(T)|^2+\int_{0}^{T}|Y(s)|^2\,ds+C\int_{0}^{T}|X(s)|^2\,ds +C\int_{0}^{T}|\overline{X}(s)|^2\,ds\right]
    \end{split}
\end{equation}
Analogous computations can be performed on $(\bar{Y}(t))$, which lead, for a new constant $C>0,$ to:
\begin{equation}\label{eq::equationItoProof2}
    \begin{split}
        &\mathbb{E}\left[\sup_{t \in [0,T]}|\overline{Y}(t)|^2 + \sum_{j=1}^{3}\int_{t}^{T}|\overline{Z}_j(s)|^2\,ds + \sum_{j=1}^{2}\int_{t}^{T}|\overline{Z}_{0,j}(s)|^2\,ds\right]\\
        &\leq C\mathbb{E}\left[|X(T)|^2+\int_{0}^{T}|X(s)|^2\,ds+\int_{0}^{T}|\overline{X}(s)|^2\,ds\right]\\
        &+C\mathbb{E}\left[\int_{0}^{T}|X(s)|^2\,ds+\int_{0}^{T}|\overline{Y}(s)|^2\,ds\right]
    \end{split}
\end{equation}
By combining Equation \eqref{eq::equationItoProof1} and \eqref{eq::equationItoProof2}, for a new constant $C_{\epsilon}>0$, we get:
\begin{equation}\label{eq::equationItoProof3}
    \begin{split}
        &\mathbb{E}\left[\sup_{t \in [0,T]}|Y(t)|^2+\sup_{t \in [0,T]}|\overline{Y}(t)|^2+\overline{C}_{\epsilon}\sum_{j=1}^{3}\int_{0}^{T}|Z_j(s)|^2\,ds+\overline{C}_{\epsilon}\sum_{j=1}^{2}\int_{0}^{T}|Z_{0,j}(s)|^2\,ds\right]\\
        &\leq C\mathbb{E}\left[|X(T)|^2 + |\overline{X}(T)|^2 \right] + C \mathbb{E}\left[\int_{0}^{T}|X(s)|^2\,ds+\int_{0}^{T}|\overline{X}(s)|^2\,ds\right]\\
        &+C \mathbb{E}\left[\int_{0}^{T}|Y(s)|^2\,ds + \int_{0}^{T}|\overline{Y}(s)|^2\,ds \right]
    \end{split}
\end{equation}
where $\overline{C}_{\epsilon}:=(1-C_{\epsilon})$. On the other hand, the standard estimates for SDEs give, for a new constant $C>0$, 
\begin{equation}\label{eq::equationItoProof4}
    \begin{split}
        &\mathbb{E}\left[\sup_{t \in [0,T]}|X(t)|^2\right]+\mathbb{E}\left[\sup_{t \in [0,T]}|\overline{X}(t)|^2\right]\\
        &\leq |X(0)|^2+\mathbb{E}\left[\int_{0}^{T}|Y(s)|^2\,ds+\int_{0}^{T}|\overline{Y}(s)|^2\,ds\right]+C.
    \end{split}
\end{equation}
Combining the inequalities \eqref{eq::equationItoProof3} and \eqref{eq::equationItoProof4} and a simple application of the Burkholder-Davis-Gundy inequality establishes the claim.
\end{proof}
We are now ready to investigate if our FBSDE \eqref{eq::statedynamicsmfglimitmatrixform}--\eqref{eq::bsdematrixform} actually provides an approximation of the market price and if so, how accurate it is.  In particular, if we use $(- 2 \mathbb{E}[Y^{(2)}(t)|\overline{\mathcal{F}}_t^{0}])$ as the input $(\overline{\omega}_t)$, where $(Y^{(2)}(t))$ is the unique solution to the FSBDE \eqref{eq::statedynamicsmfglimitmatrixform}--\eqref{eq::bsdematrixform}, then by Theorem 3.6 in \cite{graber2016linear} the optimal strategy for the individual firm is given by
\begin{equation}\label{eq::optimaltrading}
    \beta^{\star,i}(t):=-2 \nu Y^{(2),i}(t) + 2 \nu \mathbb{E}[Y^{(2)}(t)|\overline{\mathcal{F}}_t^{0}],
\end{equation}
where $(Y^{(2),i})$ is the (second component of the) solution to \eqref{eq::statedynamics} and \eqref{eq::bsdemfc} with $(\overline{\omega}_t = - 2 \mathbb{E}[Y^{(2)}(t)|\overline{\mathcal{F}}_t^{0}])$ and $W^{1} \equiv W^{1,i}$ and $W^{3} \equiv W^{3,i}$.    The next theorem shows that the market clearing condition in the large-$N$ limit, i.e.,
\begin{equation*}
    \lim_{N \rightarrow \infty} \frac{1}{N}\sum_{i=1}^{N}\beta^{\star,i}(t) = 0,\quad dt \otimes d\mathbb{P}-a.s.
\end{equation*}
holds.
\begin{theorem}\label{thm::marketclearingcondition}
Let $T>0$ and $(\beta^{\star,i}(t))$ be defined as in Equation \eqref{eq::optimaltrading}.  Then
\begin{equation}\label{thm::mktclearingone}
    \lim_{N \rightarrow \infty}\mathbb{E}\left[\Bigg|\frac{1}{N}\sum_{i=1}^{N}\beta^{\star,i}(t)\Bigg|^2\right] = 0.
\end{equation}
Moreover, there exists some constant $C$ independent of $N$ such that:
\begin{equation}\label{thm::mktclearingtwo}
    \mathbb{E}\left[\Bigg|\frac{1}{N}\sum_{i=1}^{N}\beta^{\star,i}(t)\Bigg|^2\right] \leq \frac{C}{N}
\end{equation}
\end{theorem}
\begin{proof}
    The proof is similar to the one of Lemma 5.1 in \cite{fujii2022equilibrium}. Because $(Y^{(2),i}(t))_{i=1}^{N}$ are $(\overline{\mathcal{F}}^0_t)$ conditionally i.i.d., the ladder property of the conditional expectation yields
    \begin{equation*}
        \begin{split}
            \mathbb{E}\left[\Bigg|\frac{1}{N}\sum_{i=1}^{N}\beta^{\star,i}(t)\Bigg|^2\right] &= \mathbb{E}\left[\Bigg|\frac{1}{N}\sum_{i=1}^{N}\left(-2 \nu Y^{(2),i}(t) + 2 \nu \mathbb{E}[Y^{(2)}(t)|\overline{\mathcal{F}}_t^{0}]\right)\Bigg|^2\right]\\
            &\leq   \frac{4 \nu^2}{N^2}\sum_{i=1}^{N}\mathbb{E}\left[\Big|Y^{(2),i}(t)-\mathbb{E}[Y^{(2)}(t)|\overline{\mathcal{F}}_t^{0}]\Big|^2\right]
        \end{split}
    \end{equation*}
\end{proof}
\noindent Since $\sup_{t \in [0,T]}\mathbb{E}\left[\Big|Y^{(2),i}(t)-\mathbb{E}[Y^{(2)}(t)|\overline{\mathcal{F}}_t^{0}]\Big|^2\right] \leq 2 \sup_{t \in [0,T]} \mathbb{E}\left[\Big|Y^{(2),1}(t)\Big|^2\right]$, the conclusion follows from the estimates in Corollary \ref{prop::propositionestimates}.  Finally, we conclude this section by providing explicit solutions to the FSBDE \eqref{eq::statedynamicsmfglimitmatrixform}--\eqref{eq::bsdematrixform} in terms of the following system of Riccati equations; the proof of its derivation follows the same line of arguments as in Appendix \ref{app::LQMFC} and it is, therefore, omitted.   In particular, Theorem \ref{thm::localexistence} guarantees that solutions to \eqref{eq::statedynamicsmfglimitmatrixform}--\eqref{eq::bsdematrixform} are uniquely determined in terms of these Riccati equations.

\begin{equation}\label{eq::riccati_P_new}
\begin{cases}
        \overset{\cdot}{P}(t) + C_1^T P(t) C_1 + Q - P(t) B R^{-1} B^T P(t) \textcolor{black}{+P(t) A+A^TP(t)}= 0;\\
        P(T) = H.\\
\end{cases}
\end{equation}
\begin{equation}\label{eq::riccati_PI_new}
    \begin{cases}
        \overset{\cdot}{\Pi}(t) + C_1^T P(t) C_1 + (Q + \bar{Q}) - \Pi(t) B R^{-1} \textcolor{black}{(B^T+D)} \Pi(t) +\textcolor{black}{+\Pi(t) A+A^T\Pi(t)}= 0;\\
        \Pi(T) = H.\\
    \end{cases}
\end{equation}
\begin{equation}\label{eq::riccati_phi_new}
\begin{cases}
    \overset{\cdot}{\phi}(t) - \Pi(t) B R^{-1} \tilde{r} + q +\Pi(t) A_0(t) - \Pi(t) B R^{-1} (B^T+D) \phi(t) 
    \textcolor{black}{+A^T\phi(t)}=0 \\
    \phi(T)=0.
    \end{cases}
\end{equation}
Matrices $C_1, Q, \overline{Q}, B, R, H, q, A_0(t)$ are as in Section \ref{sec::mfcapproximation}, whereas $\tilde{r}$ and $D$ are defined in Equation \eqref{eq::matricescouplingnew}. In particular, $(X(t))$ and solves the following equation:
\begin{equation}\label{eq::optimaltrajectfinal}
\begin{split}
    dX(t)   &=(A_0(s)+\textcolor{black}{AX(s)}-B R^{-1} B^T P(t) (X(t) - \overline{X}(t)))\,ds\\
            &- B R^{-1} ((B^T + D) \Pi(t) \overline{X}(t) + (B^T+D) \phi(t) + \tilde{r})\,ds\\
            &+ C_1 X(s)\,dW^1(s) + C_{0,2} dW^2(s) + C_{0,3} dW^3(s)\\
            &+F_{0,1}dW^{0,1}(s)+F_{0,2}dW^{0,2}(s),\quad\quad X(0)=x_0.
\end{split}
\end{equation}
In addition, the equilibrium price is given by
\begin{equation}\label{eq::equilibriumprice}
    \overline{\omega}_t = - 2 \left((\Pi(t))_{2,1}\overline{K}(t)+(\Pi(t))_{2,2}\overline{\widetilde{X}}(t)\right)-2(\phi(t))_{2},
\end{equation}
where $(\Pi(t))_{\ell,m}$ denotes the entry $(\ell,m)$ of the matrix $\Pi(t)$ and $(\phi(t))_2$ the second component of the vector $\phi(t)$, $t \in [0,T]$.
\section{Numerical Illustration}
\label{sec::numerical_illustrations}
We illustrate here the firm's behavior in the policy scheme described in the previous sections. We consider the objective of reducing carbon emission over $T = 5$ years. The solutions of the Riccati equations \eqref{eq::riccati_P_new}--\eqref{eq::riccati_phi_new} are computed by using the MATLAB numerical integrator ode45 with a temporal resolution of $\Delta t = 10^{-3}$, which is also employed to simulate the SDE \eqref{eq::optimaltrajectfinal} via the Euler-Maruyama method. All the expectations below are computed via the classical Monte Carlo method using $5 \cdot 10^3$ trajectories. Table \ref{tab:par_model_simulation} reports the employed numerical value for the parameters along with a synthetic yet exhaustive description in the caption.
\begin{table}[htbp]
    \centering
    \caption{Numerical values for the parameters employed in the numerical illustration. The values of the parameter for which the source is not reported are taken for the sake of the exercise. The value for the volatility of the BAU emissions $\sigma_1$ is taken from \cite{aid2023optimal}.  The average correlation to the common shock of emissions is taken from \cite{aid2023optimal}. The previous two values are estimated by considering the six main sectors covered by the EU ETS (Public Power and Heat, Pulp and Paper, Cement, Lime and Glass, Metals, Oil and Gas, Other) and the corresponding yearly verified emissions from 2008 to 2012. The $\nu$ parameter is the current value of the market depth in the EEX exchange (e.g., \cite{biagini2024}). The flexibility parameter $\eta$ is taken from \cite{biagini2024}, where it is obtained by updating the approximation in \cite{gollier2024cost}, Section 4.4. The value of $h$ is taken from \cite{aid2023optimal} and it roughly corresponds to the running maximum of the EU ETS carbon price, reached in February 2023. The initial level of BAU emissions is taken from \cite{biagini2024} and it corresponds to the order of magnitude of the energy and industrial pollution in the past four years, as per the European Environmental Agency and the data provider Statista.com. The default value for $\tilde{a}$ has been chosen according to the following rationale. In \cite{aid2023optimal} the initial firm's allowances endowment is $A_0=0.1$. By assuming that the same amount of allowances are distributed over $T=5$ years and a constant value for $\tilde{a}(t)$ we have that $0.1 = \int_{0}^{T}\tilde{a}(t)\,dt=\tilde{a} T$, whence $\tilde{a}=0.1/T$.}
    \begin{tabular}{ll} \toprule
        \textbf{Parameter} & \textbf{Numerical value and short description} \\ \midrule
        $\kappa_f$     & $5$; parameter linked to the fossil-fuel level of capital;\\
        $\kappa_g$     & $3\gamma+0.2$; parameter linked to the green level of capital;\\
        $\delta$       & $1\%$; depreciation rate of the level of capital;\\
        $\sigma$       & $0.5\%$ \textrm{Gtons/year}$^{\frac{1}{2}}$; volatility of the level of capital.\\
        $\kappa_e$     & $2$ \textrm{Gtons/(level of capital)}; proportionality factor production - emissions.\\
        $\tilde{a}$    & $[0,20]$ \textrm{Gtons/year}; growth rate emissions permits. Default: $0.5/T$ \textrm{Gtons/year}.\\
        $\tilde{\sigma}_2$ & $0.2$ \textrm{Gtons/year}$^{\frac{1}{2}}$;volatility of the emissions permits.\\
        $\sigma_1$     & 0.2 \textrm{Gtons/year}$^{\frac{1}{2}}$; volatility of the BAU emissions.\\
        $\sigma_2$     & 0.5 \textrm{Gtons/year}$^{\frac{1}{2}}$; volatility of the short-term emissions.\\
        $\rho$         & 0.92; average correlation to the common shock of emissions.\\
        $a$            & $50$ \textrm{Eur}; parameter of the inverse demand curve.\\
        $b$            & $0.07$; parameter of the inverse demand curve.\\
        $A_k$          & $2$; productivity level of the capital.\\
        $\nu$          & $285.713$ \textrm{Gton}$^2$/\textrm{Eur} a year; market depth.\\
        $\eta$         & $0.211$ \textrm{Gton}$^2$/\textrm{Eur} a year; flexibility parameter.\\
        $h$            & $80$ \textrm{Eur}/\textrm{ton}; abatement cost coefficient, liner part.\\
        $c_{1,1}$      & $0.01$ \textrm{Eur/(level of capital)}; cost fossil-fuel level of capital, linear part.\\
        $c_{1,2}$      & $3$ \textrm{Eur $\cdot$ year}/\textrm{(level of capital)$^2$};cost fossil-fuel level of capital, quadratic part.\\
        $c_{2,1}$      & $0.02$ \textrm{Eur/(level of capital)}; cost green level of capital, linear part.\\
        $c_{2,2}$      &$4$ \textrm{Eur $\cdot$ year}/\textrm{(level of capital)$^2$};cost green level of capital, quadratic part.\\
        $\lambda$      & $7.5 \cdot 10^{-*}$ \textrm{Eur/ton$^2$} with $* \in \{-7,-5-3\}$; final penalty.\\
        $\kappa_0$     & $30$; initial level of capital.\\
        $\widetilde{E}_0$  & $4$ \textrm{Gtons}; initial level of BAU emissions.\\
        $A_0$           & $0.1$ \textrm{Gtons}; initial level of emissions permits. \\
        $\widetilde{X}_0$ & $A_0-\widetilde{E}_0$; initial level of the bank account.\\
        \bottomrule
    \end{tabular}
 \label{tab:par_model_simulation}
\end{table}

\subsection{Cap-and-trade system: the role of the regulator.}
\label{subsec::the_role_of_the_regulator}
This subsection emphasizes the role played by the regulator, which may be separated into two components, namely the dynamical allocation of emission allowances $\widetilde{A}(t)$ and the severity of the cap, reflected into the parameter $\lambda$. Naturally, ceteris paribus, the average level of production increases as $\tilde{a}$ increases,  although it does not play a first-order role to the representative firm's production; see Figure \ref{fig::averageproductionvstildaa}. According to our model, this dependence may be explained by the fact that the optimal production features a linear dependence on $\tilde{a}$, with a slope that depends on the solution of the Riccati equations; see Equation \eqref{eq::optimaltrajectfinal}. Also, the average pollution abatement rate $\alpha(t)$, the optimal average trading rate $\beta(t)$, and the average price of permits $\bar{\omega}_t$ naturally decreases as $\tilde{a}$ increases; see the sub-figures in Figure \ref{fig::averagesvstildaa}, from left to right and from top to bottom. We also report the simulation of one trajectory of the just mentioned quantities.  In particular, apart from the level, the dynamics of $\alpha(t)$ and $\beta(t)$ looks very similar.    We interpret this result as being symptomatic of the type of dependence on $Y^{(2)}(t)$.    More precisely, from the (optimal) coupling condition in Equation \eqref{eq::couplingmfcendogenous} we have that
\begin{equation*}
    \beta(t) = 2 \nu (\overline{Y}^{(2)}(t)-Y^{(2)}(t))\quad\text{and}\quad\alpha(t) = 2 \eta Y^{(2)}(t) - \eta h.
\end{equation*}
Whence, both $\beta(t)$ and $\alpha(t)$ depends linearly on $Y^{(2)}(t)$ and are both affected by the idiosyncratic noise. Nonetheless, $\alpha(t)$ depends on the common shocks too. Finally, by construction, $\overline{\omega}=-\mathbb{E}[\overline{Y}^{(2)}(t)|\overline{\mathcal{F}}_t^0]$, and therefore the dependence on $\overline{Y}^{(2)}(t)$ is, again, linear. Moreover, as in \cite{aid2023optimal}, we observe large oscillations of the price near the maturity; see the last sub-figure in Figure \ref{fig::averagesvstildaa}.\\
\indent Furthermore, since the market power of the representative firm is (almost) equal to the one of the population ($\gamma=0.5$), the representative firm cannot charge higher prices to compensating possible lower sales; see also the discussion in the next Subsection \ref{subsec::the_economics_of_competition}.  Therefore, if the regulator tightens the cap, i.e., if $\lambda$ increases, all the other things being equal, then the production of the representative firm unambiguously decreases (Figure \ref{fig::averageproductionvstildaa}). Notice that this fact generalizes the results for monopoly of \cite{requate1993equivalence} and the one for Cournot oligopoly under taxes of \cite{requate1993pollution}. Consistently, the representative firm increases the use of the green level of capital at the expenses of the fossil-fuel level of capital; see Figure \ref{fig::averagecapital}. Finally,    Figure \ref{fig::averagebankaccount} plots three trajectories, one for every considered $\lambda$,  of the bank account $\widetilde{X}(t)$.   It seems that the representative firm pays more for a lower level of $\lambda$; in other words, a relaxation in the final penalty implicitly induces the firm to emit more. However, in the present work, the dynamic allocation of the regulator is exogenous and we do not have any compliance constraint, neither on the expected emissions nor on a point-wise value on the terminal net emissions of the representative firm, as in the very recent research paper \cite{biagini2024}.   Extending our model to such a setting is an interesting direction for future research; see Section \ref{sec::futureresearch}. 

\begin{figure}
    \centering
    \includegraphics[width=0.85\linewidth]{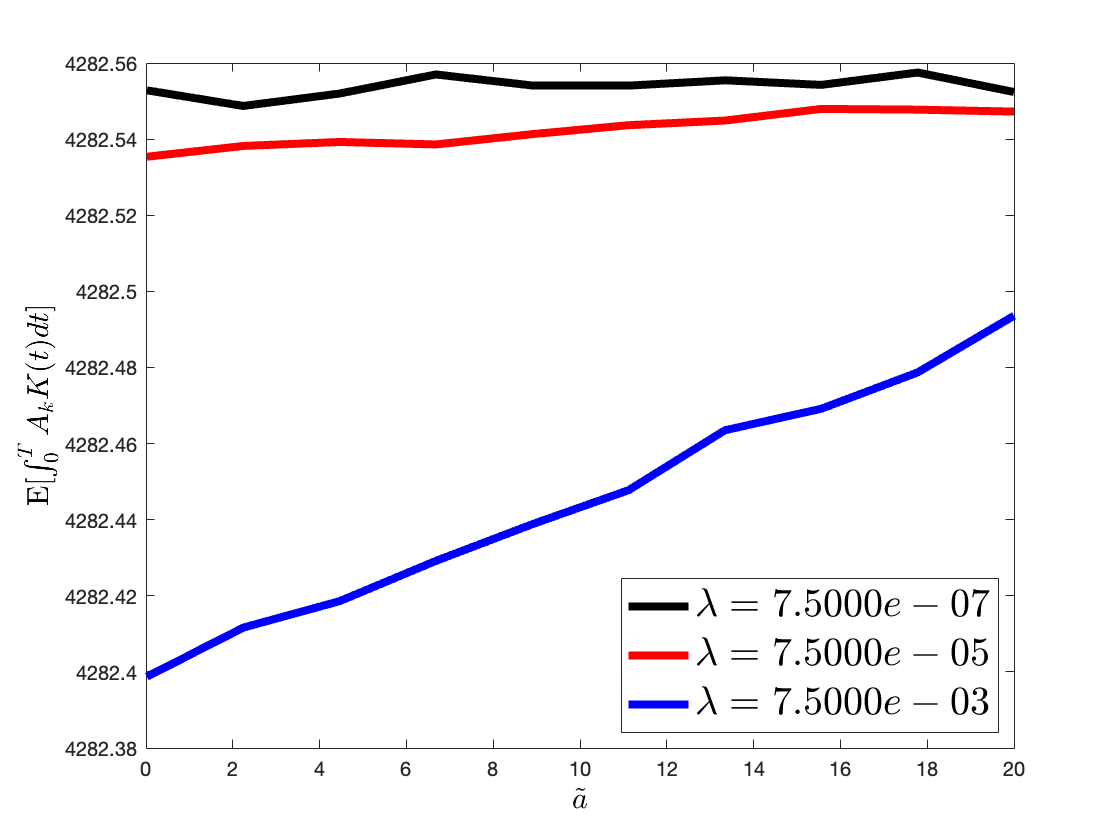}
    \caption{Effect of $\tilde{a}$ on the average level of production $\mathbb{E}\left[\int_{0}^{T}A_k K(t)\,dt\right]$. The figure plots the average level of production as a function of the growth rate emission permits $\tilde{a}$ for different levels of $\lambda$. Parameters value as in Table \ref{tab:par_model_simulation}.}
    \label{fig::averageproductionvstildaa}
\end{figure}

\begin{figure}
    \centering
    \includegraphics[width=0.49\linewidth]{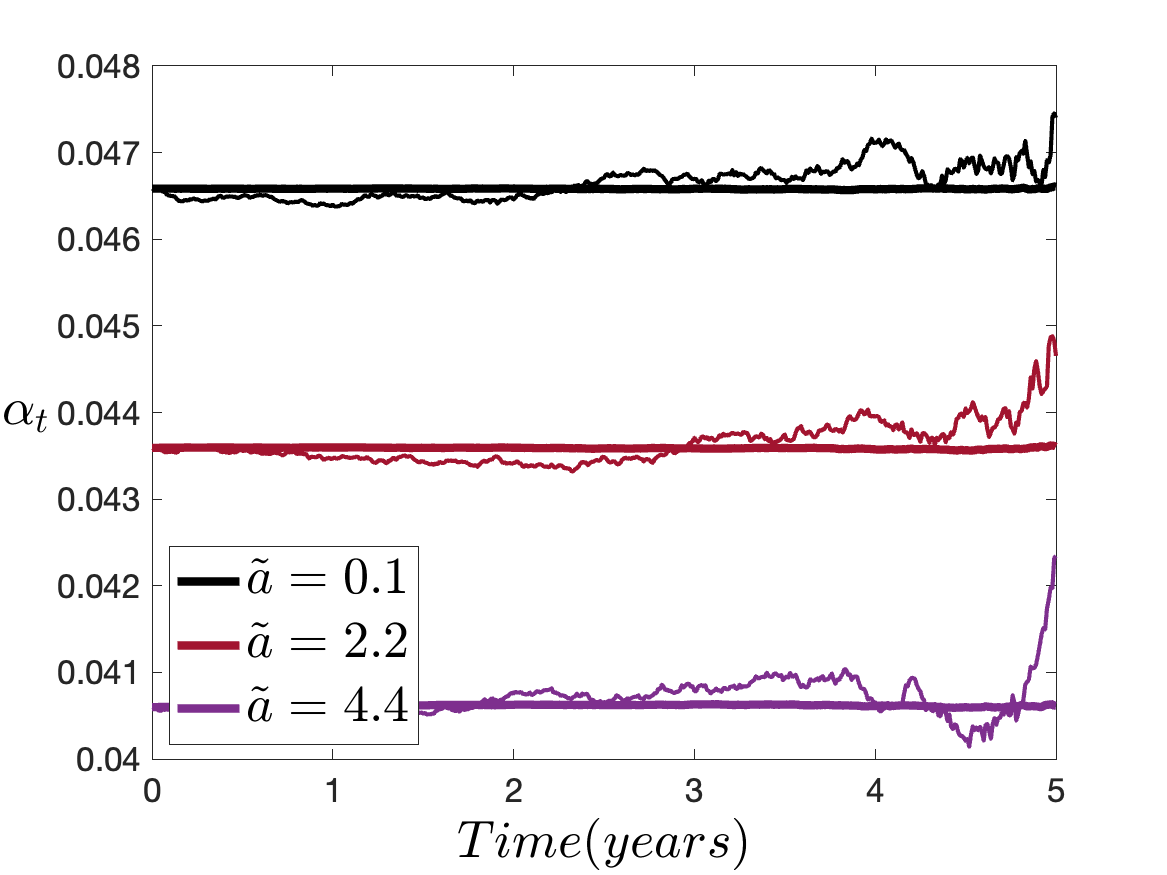}
    \includegraphics[width=0.49\linewidth]{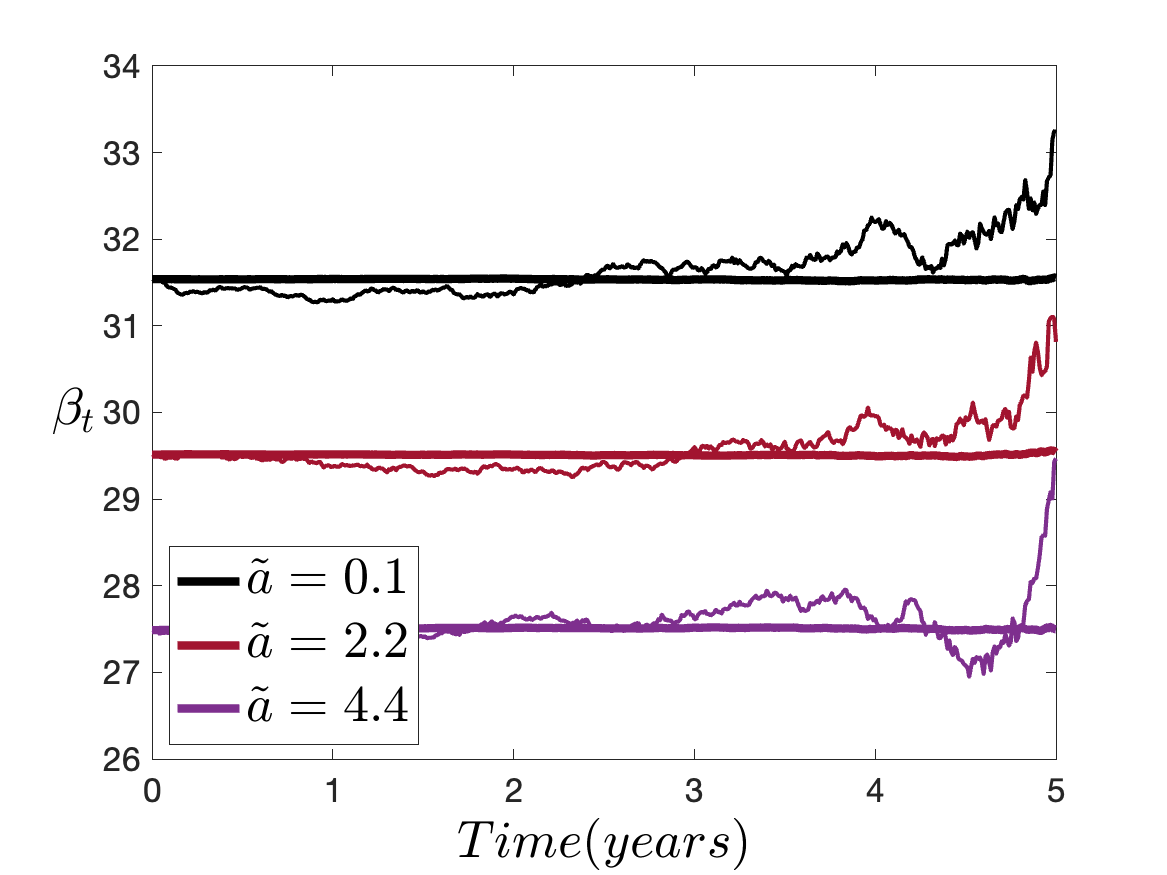}
    \includegraphics[width=0.49\linewidth]{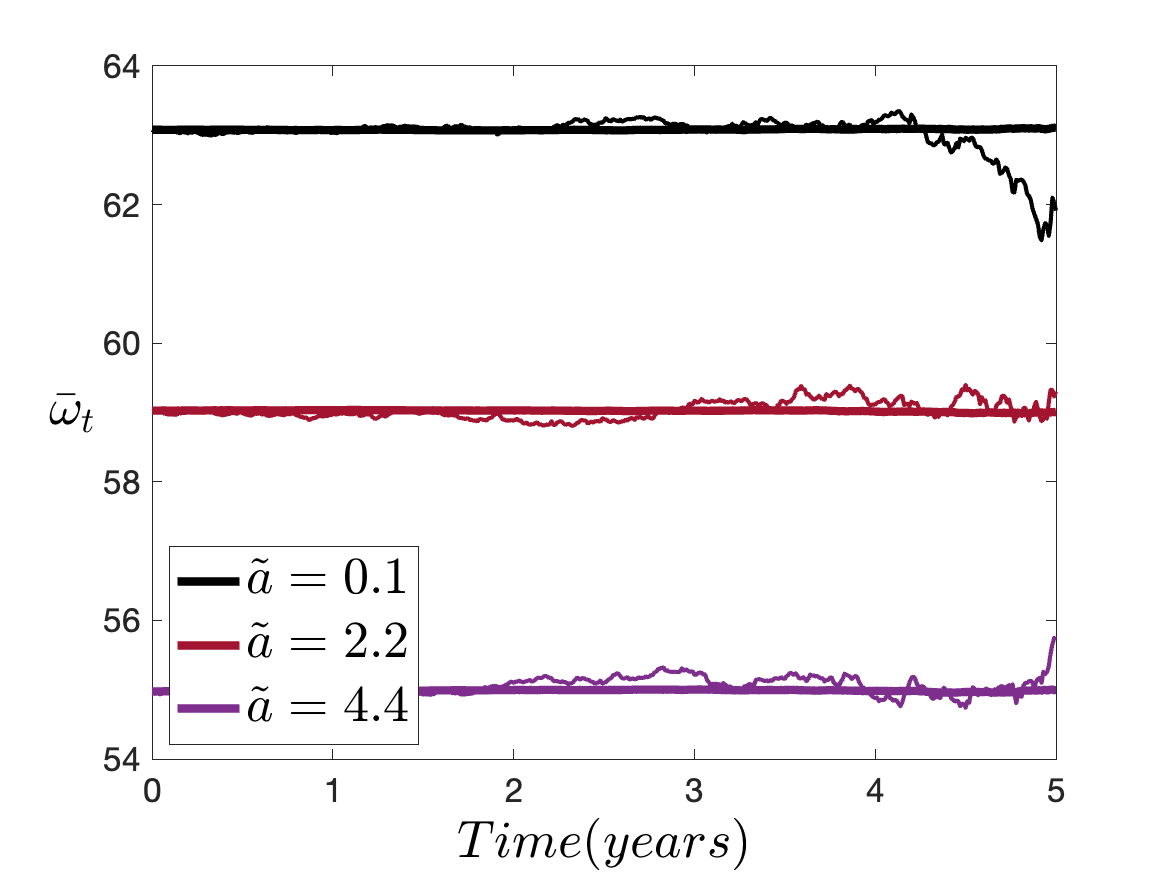}
    \includegraphics[width=0.49\linewidth]{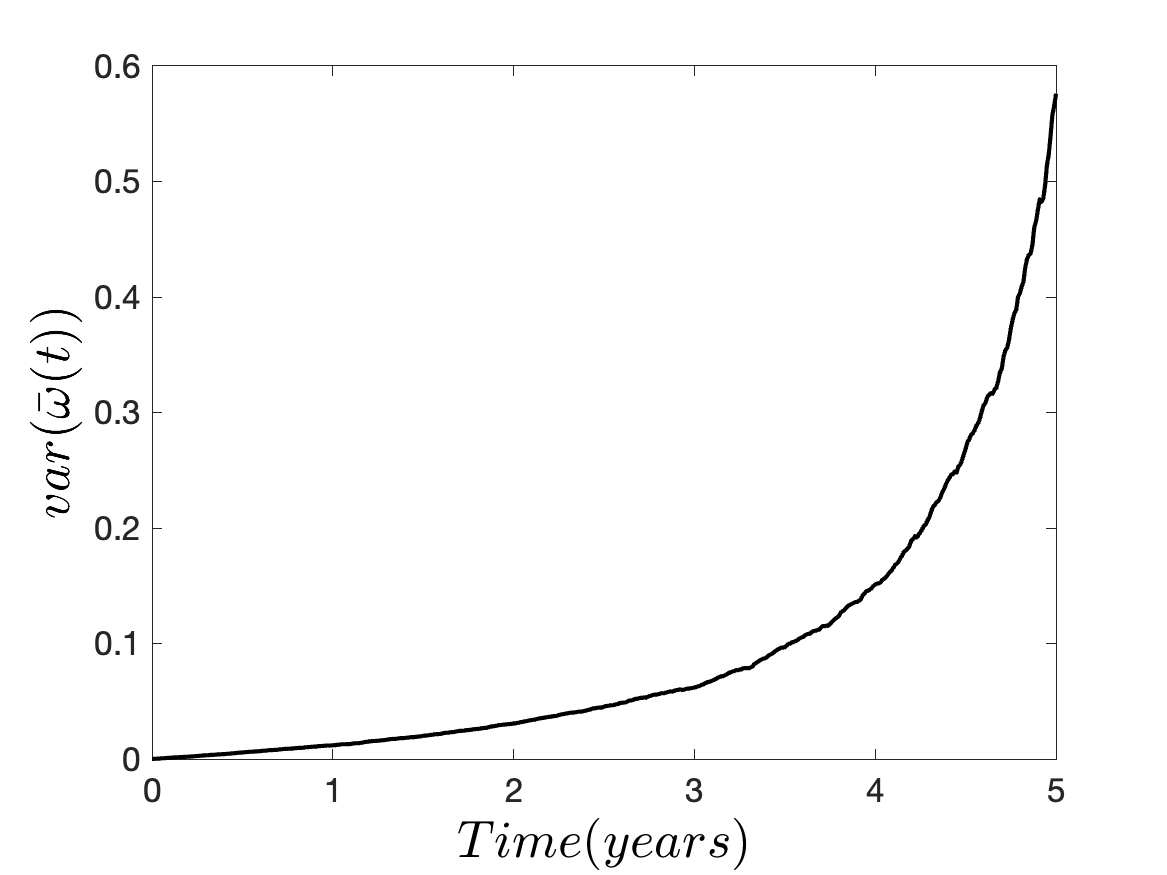}
    \caption{Effect of $\tilde{a}$ on the average pollution abatement rate $\alpha(t)$, average trading rate $\beta(t)$, and average price of permits $\bar{\omega}_t$. From left to right, from top to bottom: the figure plots 1. the optimal average level of pollution abatement $\alpha(t)$; 2. the optimal average level of optimal trading $\beta(t)$; 3. the optimal average price of permits, as a function of the growth rate emission permits $\tilde{a}$, and 4. the variance of the price of permits as a function of time. We also report a simulation of one trajectory of the corresponding quantities. Parameters value as in Table \ref{tab:par_model_simulation}.}
    \label{fig::averagesvstildaa}
\end{figure}

\begin{figure}
    \centering
    \includegraphics[width=0.49\linewidth]{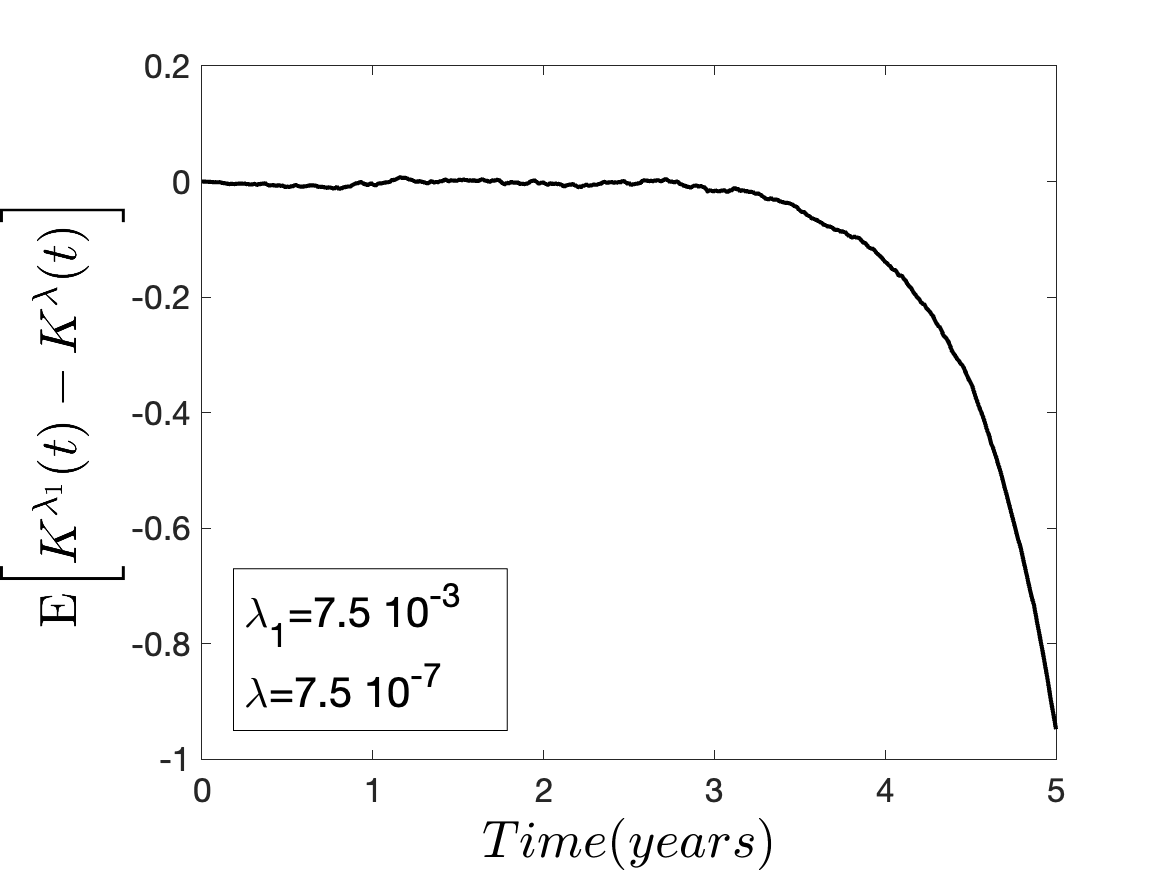}
    \includegraphics[width=0.49\linewidth]{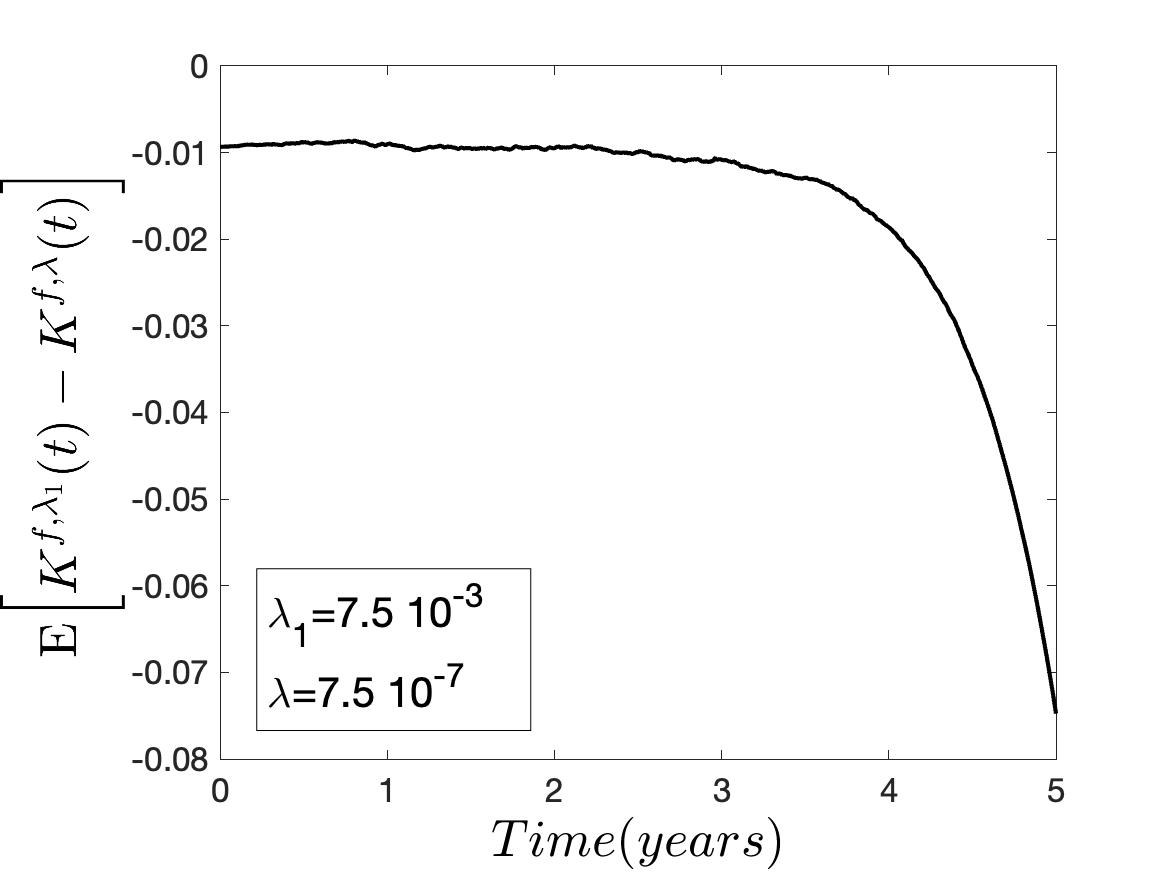}
    \includegraphics[width=0.49\linewidth]{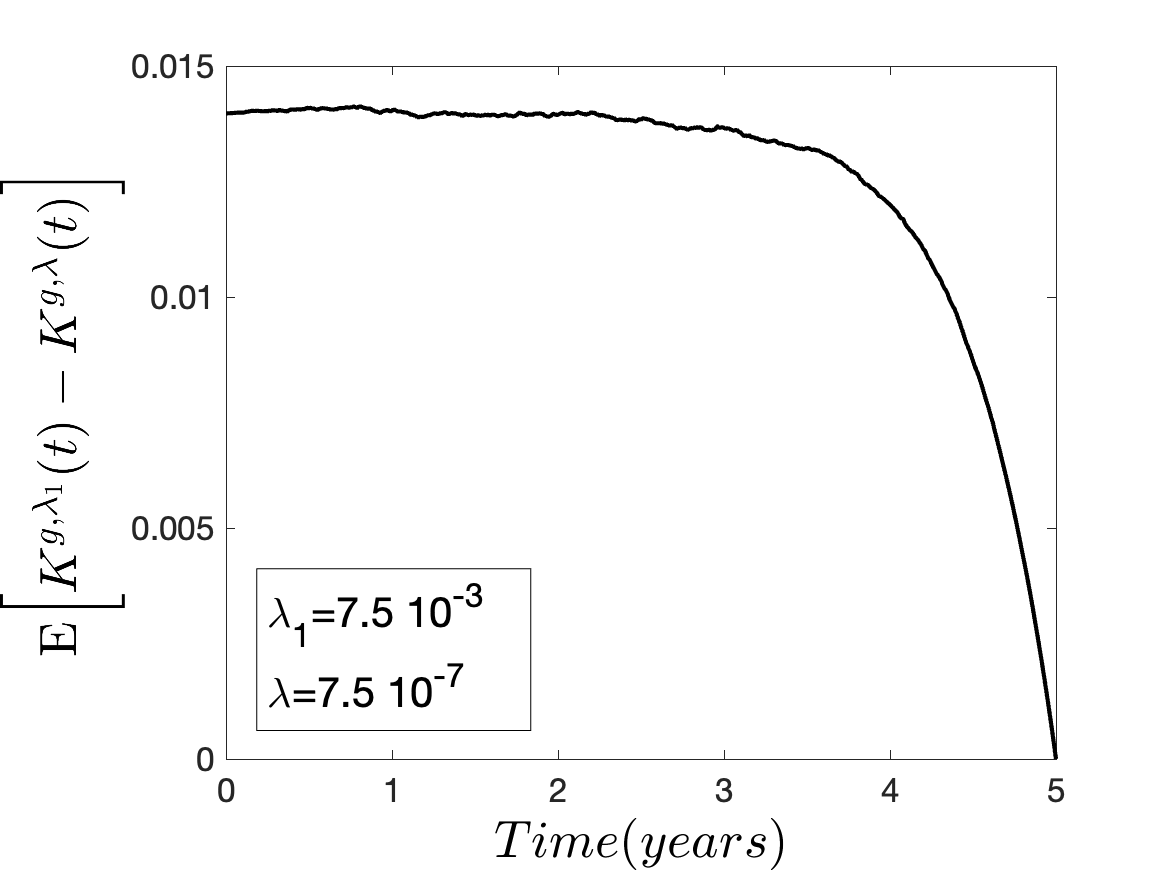}
    \caption{Effect of $\lambda$ on the average level of capital, on the average fossil-fuel level of capital, and on the average green level of capital.   From left to right, from top to bottom: The figure plots the average difference between the average level of capital, the average fossil-fuel level of capital, and the average green level of capital for two different values of the parameter $\lambda$ over the interval $[0,T]$. Parameters value as in Table \ref{tab:par_model_simulation}.}
    \label{fig::averagecapital}
\end{figure}

\begin{figure}
    \centering
    \includegraphics[width=0.85\linewidth]{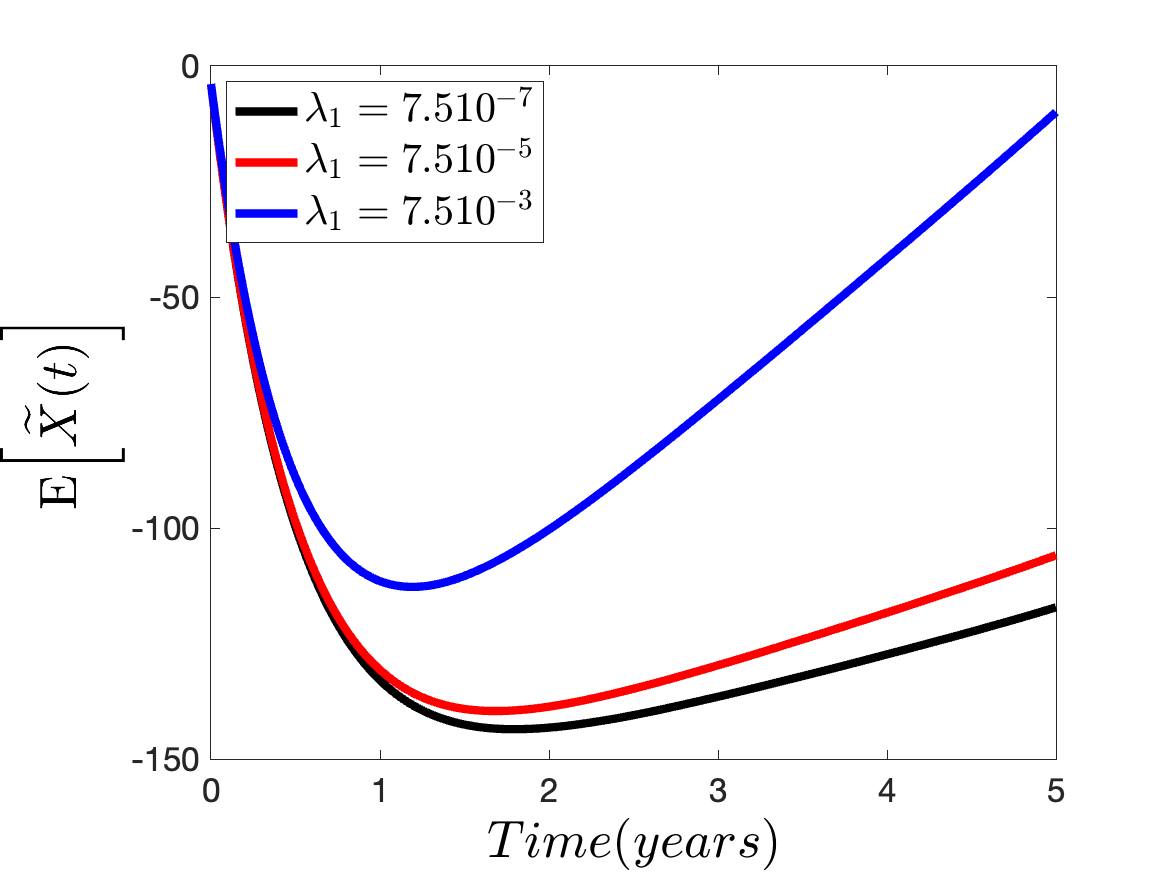}
    \caption{Effect of $\lambda$ on the average value of the bank account. The figure plots the average value of the bank account for three different values of the parameter $\lambda$ over the interval $[0,T]$. Parameters value as in Table \ref{tab:par_model_simulation}.}
    \label{fig::averagebankaccount}
\end{figure}

\newpage
\subsection{The economics of competition.}
\label{subsec::the_economics_of_competition}
Using a static model of imperfect market competition, \cite{anand2020pollution} emphasizes the importance of the degree of competition in determining the economic consequences of pollution regulation. It is therefore natural to ask whether their findings are also recovered in our (dynamic) setting with stochastic emissions\footnote{Stochastic emissions are considered as a direction for future research in \cite{anand2020pollution}, Section 7.2.} and production costs. Naturally,  the pollution regulator wants to encourage pollution abatement and discourage output reduction. Indeed, should the output be lower, consumer surplus and welfare would be hurt because of an increase in goods prices. Figure \ref{fig::gammavsprime} shows that the expected average output increases with $\gamma$,  whereas the average price of goods decreases.    This is because when $\gamma = 0$, the representative firm has significant market power and can charge higher prices, compensating it for lower sales. As $\gamma$ increases, competition with the rest of the population intensifies, thus reducing the firm's market power. This is made clear by the following relation: $p^{\overline{K}}(t)=a-b(1-\gamma) A_k K(t) - b \gamma A_k \overline{K}(t)$. Indeed, if the representative firm lowers the output, then this has a limited effect on the price because the rest of the population would increase its output in response. Consequently, the representative firm prefers pollution abatement over output reduction, with the caveat that even though the first-order partial effect is positive, it does not have the same compensatory dynamic as the one of the price. Hence, increasing the competition helps align the firm incentives with the goal of the regulation of pollution abatement; this is in line with \cite{anand2020pollution}. Consistently,  Figure \ref{fig::gammavspsecond} shows that ceteris paribus, the average level of capitals, and the average trading activity increase with an increase in the level of competition. In particular, the latter causes an increase in the average market price of permits because of increased liquidity.\\
\indent Figure \ref{fig::gammavspsecond} plots the value function along with its components. Generally,  the cap-and-trade mechanism has both a direct pollution abatement effect and an indirect output-reduction effect. Consistently with our theoretical argument in the previous paragraph, competition, i.e., a value of $\gamma \in (0,1]$, induces the representative firm to overproduce.   From a population perspective, it would be beneficial if every (representative) firm lowers the corresponding output to keep prices high. This is an unlikely scenario because no representative firm could credibly commit to such a lower output, as one would expect. The cap-and-trade mechanism should coordinate the previous mechanism in such a way that the population of firms agrees to reduce output by using the pollution constraint; naturally, this synchronization mechanism is expected to work under a suitable range of constraints imposed by the pollution regulator, the one for which the impact of output reduction on the representative firm's profits dominates the cost of pollution abatement, of trading, and production. In particular, in our numerical example costs dominate revenues, and therefore profits, as $\gamma \in (0,1]$ increases. When $\gamma=0$ (monopoly), the representative firm has significant market power and it can optimize its output.  However, should the regulator decide to tighten the cap, the output of the representative firm would further reduce and the representative firm can no longer leverage on the competition with the population of firms to implement the previously described synchronization mechanism. Therefore, a cap-and-trade mechanism hurts more monopoly than competitive firms, which is in agreement with the findings for competitive markets in \cite{anand2020pollution}.

\begin{figure}
    \centering
    \includegraphics[width=0.49\linewidth]{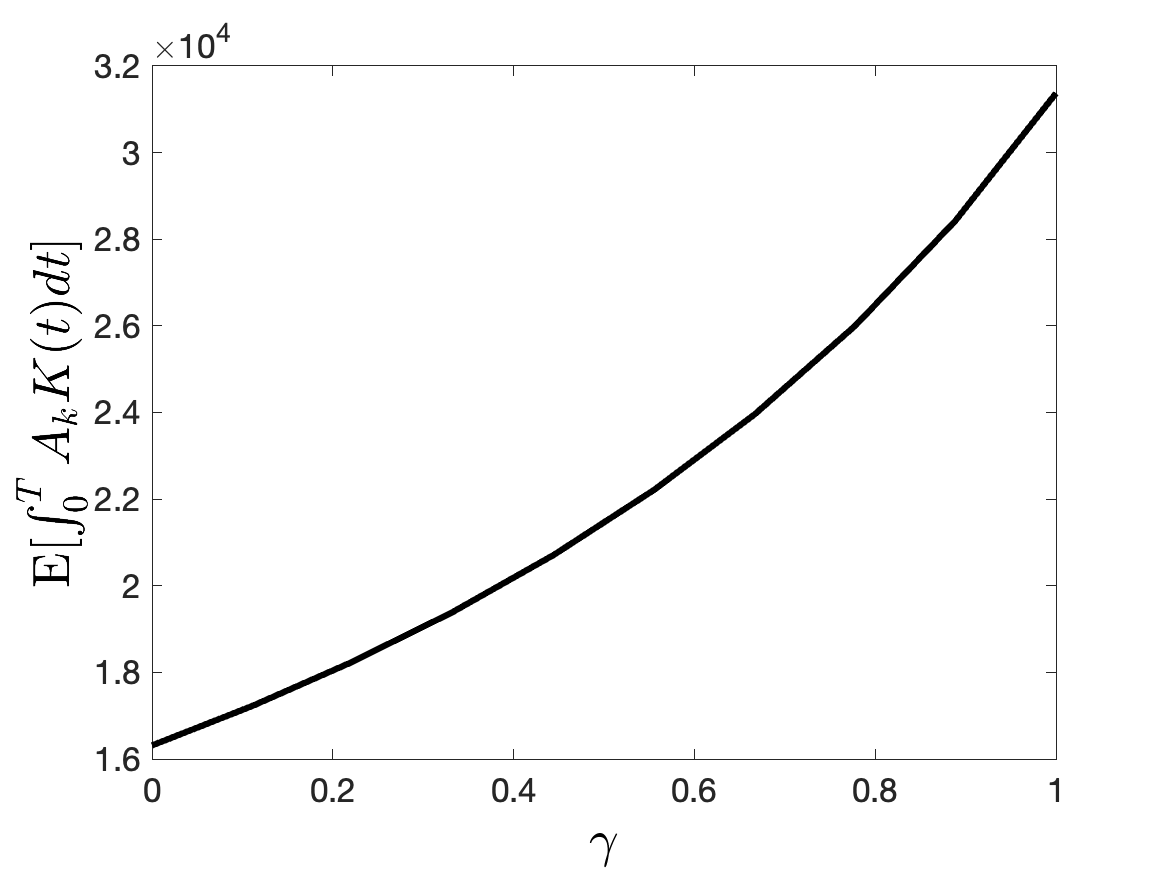}
    \includegraphics[width=0.49\linewidth]{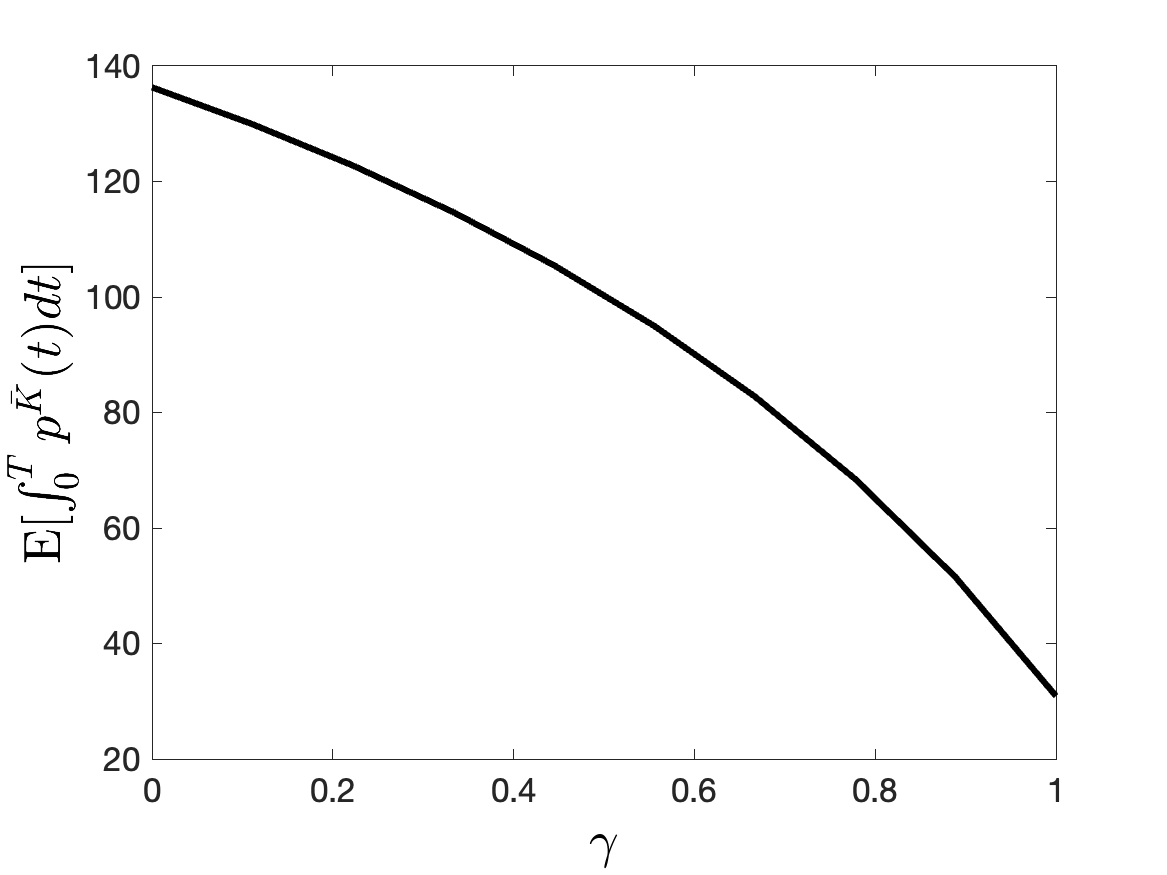}
    \includegraphics[width=0.49\linewidth]{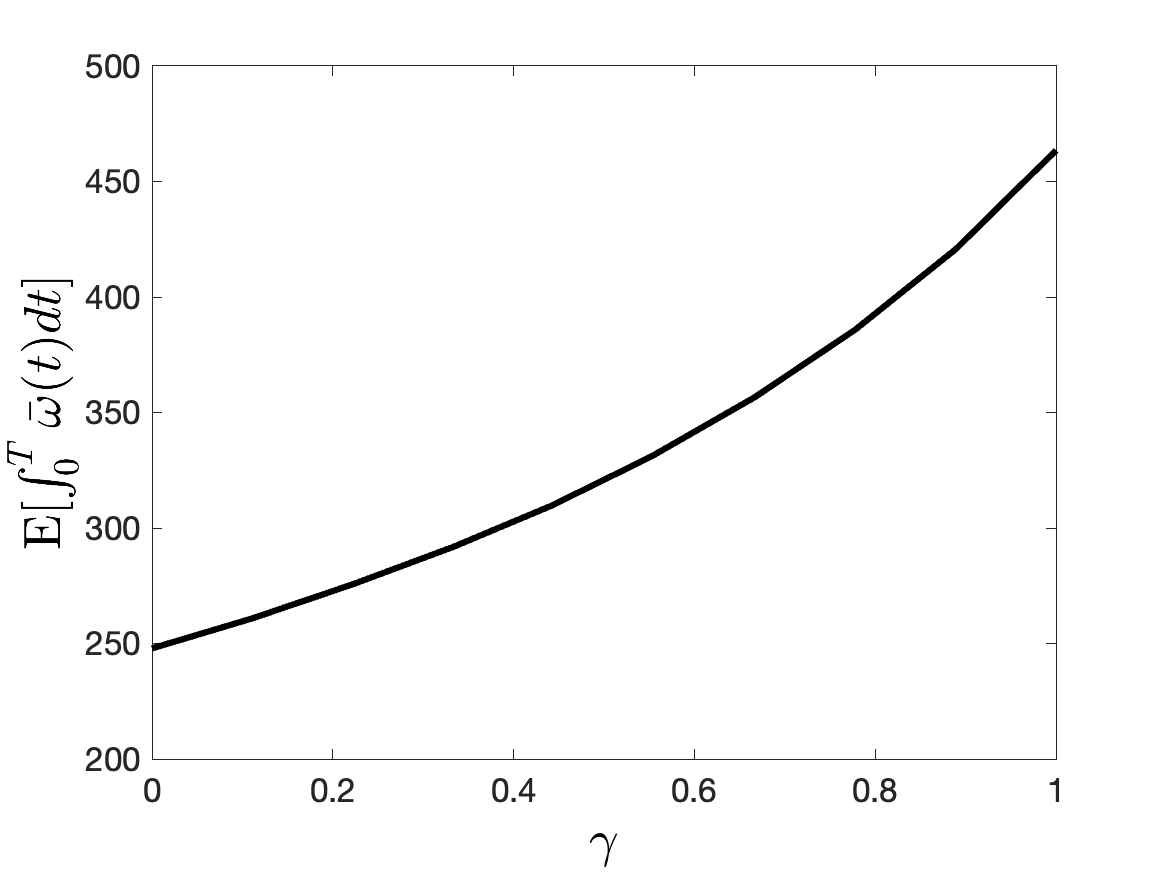}
    \caption{Effect of the level of competition $\gamma$ on the average level of production $\mathbb{E}\left[\int_{0}^{T}A_k K(t)\,dt\right]$, on the average price of good $\mathbb{E}\left[\int_{0}^{T}p^{\overline{K}}(t)\,dt\right]$, and on the average price of permits $\mathbb{E}\left[\int_{0}^{T} \overline{\omega}(t)\,dt\right]$. From left to right, from top to bottom: The figure plots the average value of the production, the average price of good, and the average price of permits as a function of the level of competition $\gamma$. Parameters value as in Table \ref{tab:par_model_simulation}.}
    \label{fig::gammavsprime}
\end{figure}

\begin{figure}
    \centering
    \includegraphics[width=0.49\linewidth]{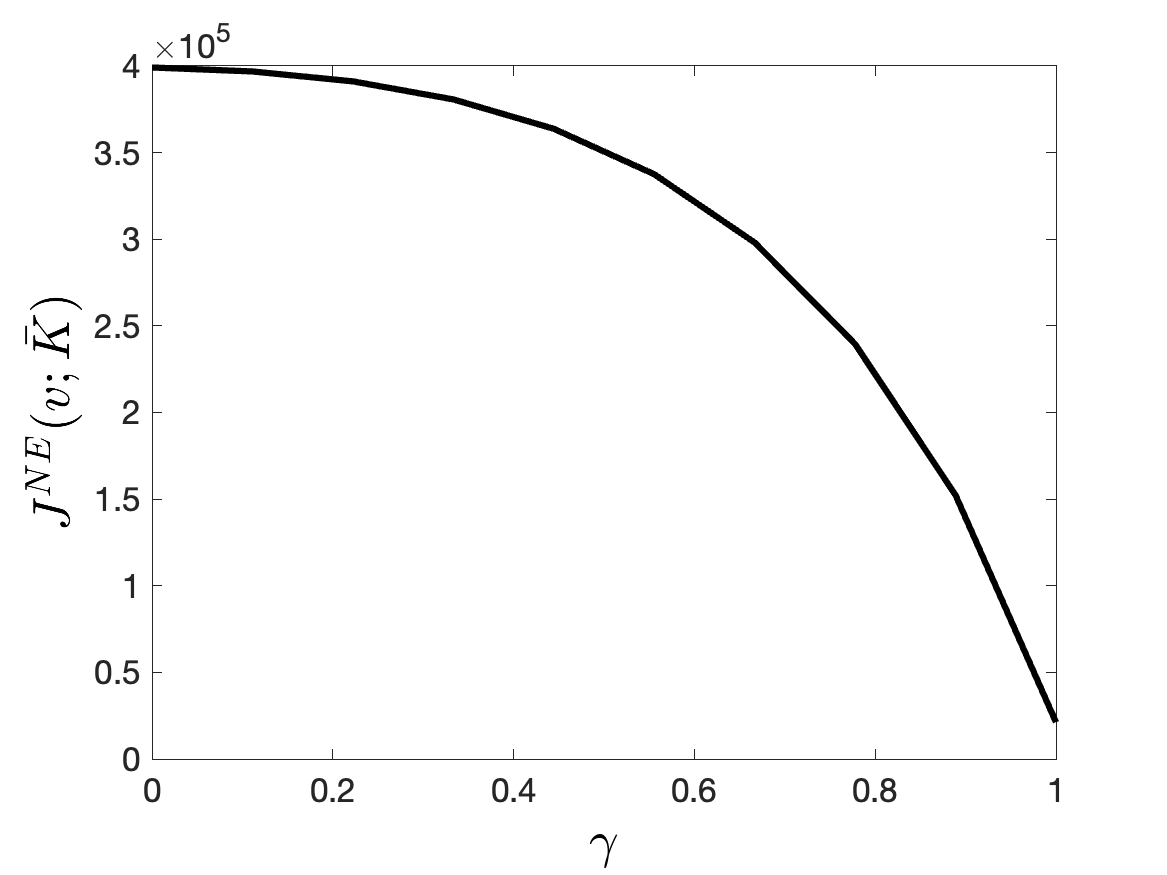}
    \includegraphics[width=0.49\linewidth]{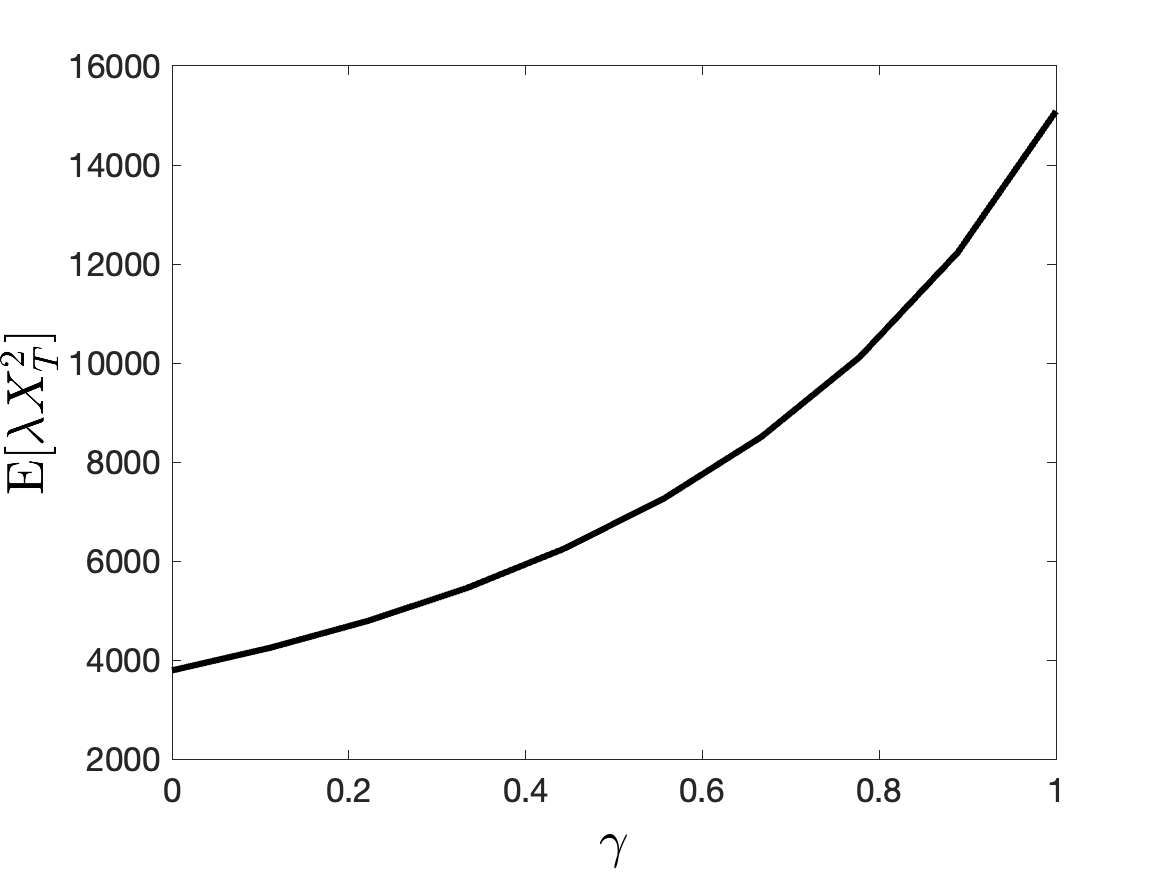}
    \includegraphics[width=0.49\linewidth]{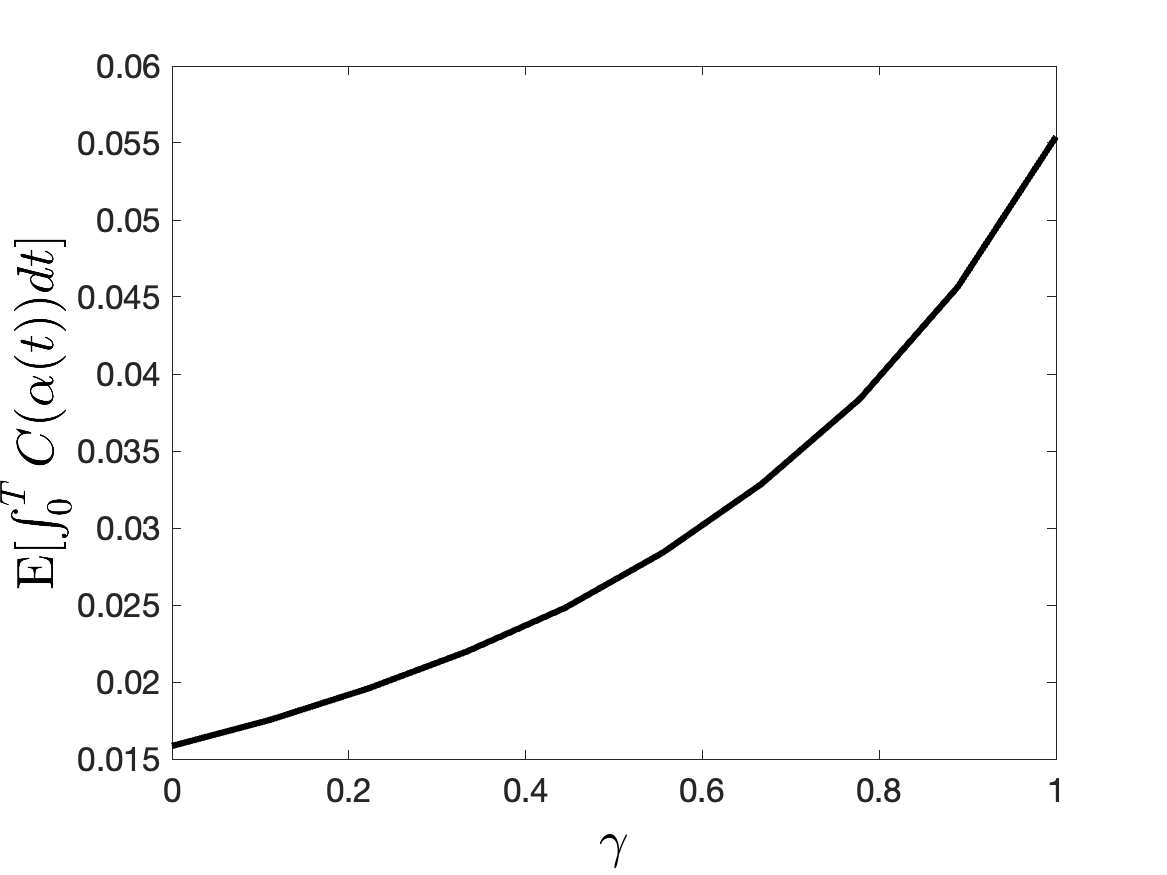}
    \includegraphics[width=0.49\linewidth]{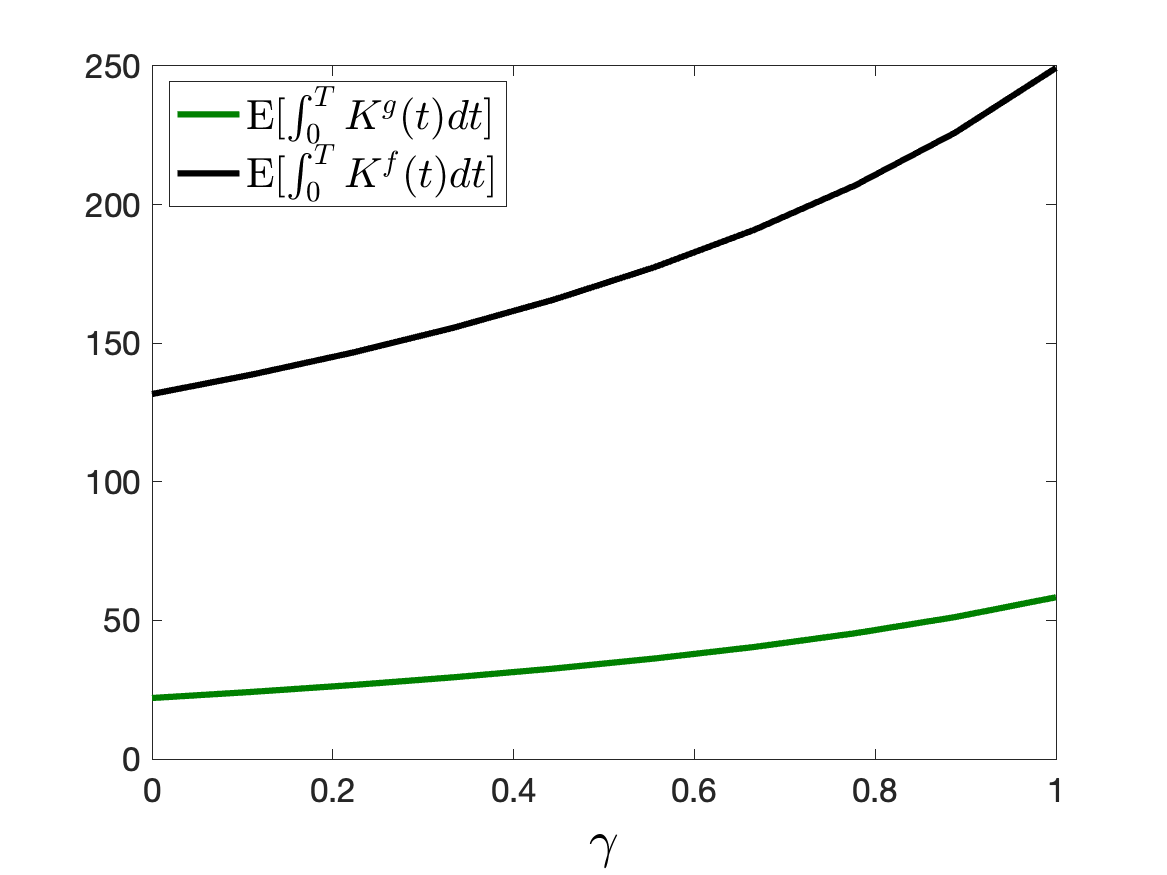}
     \includegraphics[width=0.49\linewidth]{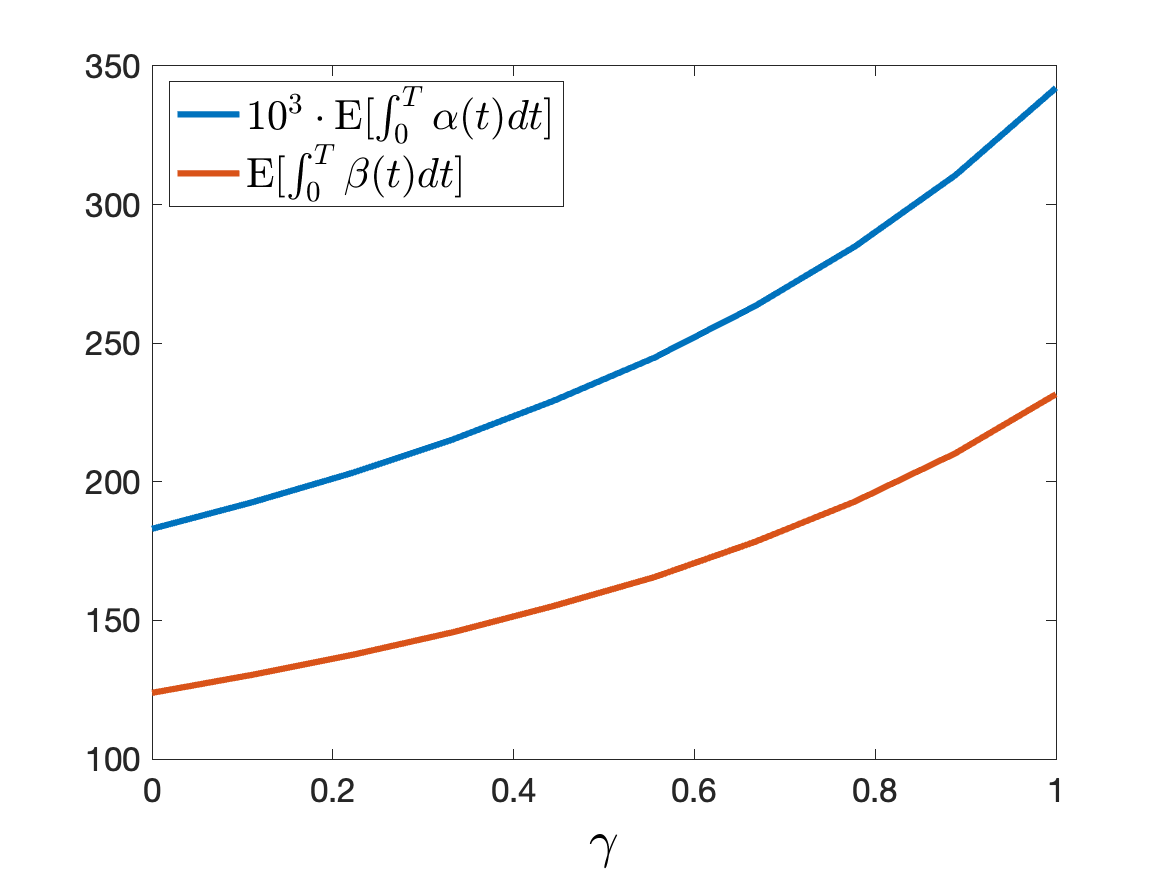}
    \caption{From left to right, from top to bottom: Effect of the level of competition $\gamma$ on the value function $\mathcal{J}^{NE}(v,\overline{K})$, on the average final penalty $\mathbb{E}[\lambda X^2_T]$, on the average abatement cost $\mathbb{E}[\int_{0}^{T}C(\alpha(t))\,dt]$, on the average level of fossil-fuel $\mathbb{E}[\int_{0}^{T}K^{f}(t)\,dt]$ and green capital $\mathbb{E}[\int_{0}^{T}K^{f}(t)\,dt]$, and on the average pollution abatement rate $\mathbb{E}[\int_{0}^{T}\alpha(t)\,dt]$ and average trading rate $\mathbb{E}[\int_{0}^{T}\beta(t)\,dt]$. Parameters value as in Table \ref{tab:par_model_simulation}.}
    \label{fig::gammavspsecond}
\end{figure}

\newpage
\section{Conclusion and future research}\label{sec::futureresearch}
The model proposed in this paper introduces several fundamental elements regarding pollution generation, abatement and costs, and regulation, which can serve as a basis for future research.\\
\indent The model assumes that firms produce products using a standard AK model with a positive depreciation rate of capital. Future research could consider relaxing this assumption and exploring a more realistic, namely non-linear, production function. Additionally, it assumes that the business cycle affecting the BAU carbon emissions does not correlate with the one affecting emission allowances.    Said differently,   we are assuming that the regulator has access to limited information and it is affected by a macroeconomic shocks driver that it is independent from the one influencing the emissions. Investigating correlated business cycles and the impact of asymmetric information on production and abatement costs could be potential areas for future research; the latter, in particular, can lead to interesting agency problems. In the present work, the dynamic allocation of the regulator is exogenous and we do not consider any compliance constraint, neither on the expected emissions nor on a point-wise value on the terminal net emissions of the representative firm, as done in the very recent research paper \cite{biagini2024}. Extending our
model to such a setting is an interesting direction for future research. The model also assumes that all firms share the same cost and coefficient functions. Extending the model to incorporate multiple populations, where firms within each population share the same cost and coefficient functions, but differ across populations, is an area for future exploration. This will provide an important tool to study the market equilibrium price in the presence of different types of firms. In addition, it's important for future research to consider relaxing the assumption of the carbon price being $(\overline{\mathcal{F}}^0)$-adapted. Also, it might be an interesting avenue for future research to account for the way in which the regulator allocates allowances to individual firms, in particular analysing the case when the initial allocation is through auctions, as initially intended by the European Union, which switched back to grandfathering in the third phase (after 2012). Extending the model to account for multiple compliance periods and the specific design of current cap-and-trade systems could lead to clear-cut predictions about permit prices and related derivatives. Finally, future research could consider integrated production-pollution-abatement models in continuous time and study other types of policies, such as the policy rules under the Market Stability Reserve (MSR), launched in 2019 by the EU.  

\newpage
\appendix
\section{Linear quadratic mean field games with common noise.}\label{app::LQMFG}
In this section, we present the general formulation of a linear quadratic mean field game class with common noise, which our framework fits into.\\  
Let $(\overline{\xi}(s))$ and $(\overline{\psi}(s))$ be given processes adapted to the filtration $(\overline{\mathcal{F}}^0_t)$. We consider the following dynamics:
\begin{equation}\label{eq::statedynamicsMfgGraber}
    \begin{split}
        dX(s) &= (A_0(s) + A(s)X(s)+\overline{A}(s)\overline{\xi}(s)+B(s)v(s)+\overline{B}(s)\overline{\psi}(s))\,ds\\
              &+ \sum_{j=1}^{d_1}\left(C_{0,j}(s) + C_{j}(s)X(s) + \overline{C}_{j}(s)\overline{\xi}(s)+D_{j}(s)v(s)+\overline{D}_{j}(s)\overline{\psi}(s)\right)\,dW^{j}(s)\\
              &+\sum_{\ell=1}^{d_0}\left(F_{0,j}(s)+F_{j}(s)X(s)+\overline{F}_{j}(s)\overline{\xi}(s)+G_{j}(s)v(s)+\overline{G}_{j}(s)\overline{\psi}(s)\right)\,dW^{0,j}(s),\\
    \end{split}
\end{equation}
with $X(t)=x$, and objective functional:
\begin{equation}\label{eq::costfunctionalMfgGraber}
    \begin{split}
        \mathcal{J}_{x,t}^{NE}(v)=\mathbb{E}\Bigg[\int_{t}^{T}&\Bigg( Q_0(s)+\langle Q(s)X(s), X(s)\rangle+ 2 \langle\overline{Q}(s)\overline{\xi}(s),\overline{X}(s)\rangle+\langle R(s)v(s),v(s)\rangle\\
        &+  2 \langle\overline{R}(s)\overline{\psi}(s), v(s)\rangle +2\langle S(s)X(s),v(s)\rangle
        +2\langle \overline{S}_1(s)\overline{\xi}(s), v(s)\rangle\\
        &+2\langle \overline{S}_2(s)X(s), \overline{\psi}(s)\rangle\\
        &+2\langle q(s),X(s)\rangle
        +2\langle \overline{q}(s),\overline{\xi}(s)\rangle +2\langle r(s), v(s)\rangle + 2\langle \overline{r}(s),\overline{\psi}(s)\rangle
        \Bigg)\,ds\\
        &+ \langle H X(T), X(T)\rangle + 2 \langle\overline{H}\overline{\xi}(T), X(T)\rangle\Bigg],
    \end{split}
\end{equation}
where $\langle\cdot, \cdot\rangle$ denotes the inner product on Euclidean space. The goal is to find a control $\hat{v}(s)$ with corresponding state process $\hat{X}$ such that $\mathcal{J}_{x,t}^{NE}(\hat{v};\overline{\xi}, \overline{\psi})=\inf_{v}\mathcal{J}_{x,t}^{NE}(v;\overline{\xi}, \overline{\psi})$ and $\mathbb{E}[\hat{X}(s)|\overline{\mathcal{F}}_s^0]=\overline{\xi}$, $\mathbb{E}[\hat{v}(s)|\overline{\mathcal{F}}_s^0]=\overline{\psi}$. The process $\hat{v}$ is called a mean field Nash equilibrium.
We state the following assumption on the coefficient matrices (cfr.~\cite{graber2016linear}, Assumption 3.1) 
\begin{enumerate}[label=(N\arabic*)]
\item\label{itm:N1} $A_0, C_{0,j}, F_{0,j} \in L^{\infty}([0,T];\mathbb{R}^{d})$, $1 \leq j \leq d_1$ and $\ell \leq 1 \leq d_0$, and $Q_0(s) \in L^{\infty}([0,T];\mathbb{R})$.
\item\label{itm:N2} $A, \overline{A}, C, \overline{C}, F, \overline{F} \in L^{\infty}([0,T];\mathbb{R}^{d \times d})$.
\item\label{itm:N3} $B, \overline{B}, D, \overline{D}, G, \overline{G} \in L^{\infty}([0,T];\mathbb{R}^{d \times d_2})$.
\item\label{itm:N4} $Q, \overline{Q}\in L^{\infty}([0,T];\mathcal{S}^{d})$,  $R, \overline{R}\in L^{\infty}([0,T];\mathcal{S}^{d_2})$,  $H, \overline{H} \in \mathcal{S}^{d}$.
\item\label{itm:N5} $H\geq 0$ and for some $\delta_1 \geq 0$, $\delta_2 >0$, $Q, Q\geq \delta_1 I_{d}$ and $R \geq \delta_2 I_{d}$.
\item\label{itm:N6} $S, \overline{S}_1, \overline{S}_2  \in L^{\infty}([0,T]; \mathbb{R}^{d_2 \times d})$; $q, \overline{q} \in L^{\infty}([0,T];\mathbb{R}^{d})$; $r, \overline{r} \in L^{\infty}([0,T];\mathbb{R}^{d_2})$.
\item\label{itm:N7} $\|S\|_{\infty}^2 < \delta_1\delta_2$ if $\delta_1>0$, $S=\overline{S}=0$ otherwise.
\end{enumerate}

\section{Linear quadratic mean field type control with common noise.}\label{app::LQMFC}
In this section, we provide explicit solutions of a class of linear quadratic mean field type control problems in terms of a system of Riccati equations. The class of problems we consider is a generalization of the one analyzed in \cite{graber2016linear}. In the latter, both the private states dynamics and the running cost appearing in the cost functional do not contain (possibly time-dependant) terms of order zero, and both common and idiosyncratic noise values are uni-dimensional, i.e., $d_0=d_1=1$. Instead, we consider the following dynamics:
\begin{equation}\label{eq::statedynamicsMfcGraber}
    \begin{split}
        dX(s) &= (A_0(s) + A(s)X(s)+\overline{A}(s)\overline{X}(s)+B(s)v(s)+\overline{B}(s)\overline{v}(s))\,ds\\
              &+ \sum_{j=1}^{d_1}\left(C_{0,j}(s) + C_{j}(s)X(s) + \overline{C}_{j}(s)\overline{X}(s)+D_{j}(s)v(s)+\overline{D}_{j}(s)\overline{v}(s)\right)\,dW^{j}(s)\\
              &+\sum_{\ell=1}^{d_0}\left(F_{0,j}(s)+F_{j}(s)X(s)+\overline{F}_{j}(s)\overline{X}(s)+G_{j}(s)v(s)+\overline{G}_{j}(s)\overline{v}(s)\right)\,dW^{0,j}(s),\\
    \end{split}
\end{equation}
with $X(t)=x$. In addition, the objective cost functional is given by:
\begin{equation}\label{eq::costfunctionalMfcGraber}
    \begin{split}
        \mathcal{J}_{x,t}^{LQ}(v)=\mathbb{E}\Bigg[\int_{t}^{T}&\Bigg( Q_0(s)+\langle Q(s)X(s), X(s)\rangle+ \langle\overline{Q}(s)\overline{X}(s),\overline{X}(s)\rangle+\langle R(s)v(s),v(s)\rangle\\
        &+  \langle\overline{R}(s)\overline{v}(s),\overline{v}(s)\rangle +2\langle S(s)X(s),v(s)\rangle+2\langle \overline{S}(s)\overline{X}(s),\overline{v}(s)\rangle+2\langle q(s),X(s)\rangle\\
        &+2\langle \overline{q}(s),\overline{X}(s)\rangle +2\langle r(s), v(s)\rangle + 2\langle \overline{r}(s),\overline{v}(s)\rangle
        \Bigg)\,ds\\
        &+ \langle H X(T), X(T)\rangle + 2 \langle\overline{H}\overline{X}(T), X(T)\rangle\Bigg],
    \end{split}
\end{equation}
where $\langle\cdot, \cdot\rangle$ denotes the inner product on Euclidean space.

We state the following assumption on the coefficient matrices (cfr.~\cite{graber2016linear}, Assumption 2.1) 
\begin{enumerate}[label=(M\arabic*)]
\item\label{itm:M1} $A_0, C_{0,j}, F_{0,j} \in L^{\infty}([0,T];\mathbb{R}^{d})$, $1 \leq j \leq d_1$ and $\ell \leq 1 \leq d_0$, and $Q_0(s) \in L^{\infty}([0,T];\mathbb{R})$.
\item\label{itm:M2} $A, \overline{A}, C, \overline{C}, F, \overline{F} \in L^{\infty}([0,T];\mathbb{R}^{d \times d})$.
\item\label{itm:M3} $B, \overline{B}, D, \overline{D}, G, \overline{G} \in L^{\infty}([0,T];\mathbb{R}^{d \times d_2})$.
\item\label{itm:M4} $Q, \overline{Q}\in L^{\infty}([0,T];\mathcal{S}^{d})$,  $R, \overline{R}\in L^{\infty}([0,T];\mathcal{S}^{d_2})$,  $H, \overline{H} \in \mathcal{S}^{d}$.
\item\label{itm:M5} $H, H+\overline{H}\geq 0$ and for some $\delta_1 \geq 0$, $\delta_2 >0$, $Q, Q+\overline{Q}\geq \delta_1 I_{d}$ and $R, R+\overline{R}\geq \delta_2 I_{d}$.
\item\label{itm:M6} $S, \overline{S} \in L^{\infty}([0,T]; \mathbb{R}^{d_2 \times d})$; $q, \overline{q} \in L^{\infty}([0,T];\mathbb{R}^{d})$; $r, \overline{r} \in L^{\infty}([0,T];\mathbb{R}^{d_2})$.
\item\label{itm:M7} $\|S\|_{\infty}^2, \|S+\overline{S}\|_{\infty}^2 < \delta_1\delta_2$ if $\delta_1>0$, $S=\overline{S}=0$ otherwise.
\end{enumerate} 
The procedure used in \cite{graber2016linear}, Section 2.2, uses a technique developed by \cite{yong2013linear}. In order to facilitate the reader, we will highlight in bold font the additional terms with respect \cite{graber2016linear}, Theorem 2.6, Equations (2.31)-- (2.32) and the subsequent non-numbered one, linked to the terms of order zero\footnote{Notice that the expression for $\Sigma_0$ and $\phi(s)$ derived in \cite{graber2016linear} presents some inaccuracies. First, the term $G^T P G$ is missed in the expression for $\Sigma_0$ (see Equation \eqref{eq::equationsixteenAppA}). Second, there is an extra term in the equation for $\phi$; nonetheless, the equation remains linear and, therefore, essentially trivial to solve (see Equation \eqref{eq::ODEforphi})}. We suppose that:
\begin{equation}\label{eq::equationOneAppA}
    Y(s) = P(s)(X(s)-\overline{X}(s)) + \Pi(s)\overline{X}(s) + \phi(s),
\end{equation}
where $P$ and $\Pi$ are $\mathcal{S}^{d}$-valued processes such that they satisfy the following terminal conditions: 
\begin{equation*}
    P(T)=H,\quad\Pi(T)=H+\overline{H},
\end{equation*}
and $\phi(s)$ is an $\mathbb{R}^d$-valued process; $P$, $\Pi$, and $\phi$ are deterministic. Hereafter, in order to ease the notation, we suppress the time indexes and we work under the assumption that $d_0=d_1=1$; we will provide the expressions for the case $d_0>1$ and $d_1>1$ at the end of the present section. By taking the conditional expectation in Equation \eqref{eq::equationOneAppA}, we obtain:
\begin{equation}\label{eq::equationTwoAppA}
    \overline{Y}=\Pi \overline{X} +\phi \,\,\text{and}\,\,Y -\overline{Y} = P (X -\overline{X}).
\end{equation}
Moreover, by taking the conditional expectation in Equation \eqref{eq::statedynamicsMfcGraber}, we obtain:
\begin{equation}\label{eq::equationThreeAppA}
    d\overline{X} = \left(\boldsymbol{A_0} + (A+\overline{A})\overline{X}+(B+\overline{B})\overline{v}\right)\,ds+\left(\boldsymbol{F_0} + (F+\overline{F})\overline{X}+(G+\overline{G})\overline{v}\right) dW^{0}.
\end{equation}
By subtracting the previous equation from Equation \eqref{eq::statedynamicsMfcGraber}, we have:
\begin{equation}\label{eq::equationFourAppA}
    \begin{split}
        d(X-\overline{X})&=\left(A (X - \overline{X}) + B (v - \overline{v})\right),ds\\
                    &+\left(\boldsymbol{C_0} + C (X-\overline{X}) + (C+\overline{C})\overline{X} + D(v-\overline{v})+(D+\overline{D})\overline{v}\right)\,dW\\
                    &+\left( F(X-\overline{X})+G(v-\overline{v})\right)\,dW^0.
    \end{split} 
\end{equation}
Proposition 2.4, Equation (2.4), in \cite{graber2016linear} gives us\footnote{Proposition 2.4, Equation (2.4), in \cite{graber2016linear} is not affected by the presence of zero-order terms.}:
\begin{equation}\label{eq::equationFiveAppA}
\begin{split}
    dY &=-\Bigl(A^T (Y-\overline{Y}) + (A^T+\overline{A}^T)\overline{Y} + C^T (Z-\overline{Z}) + (C^T+\overline{C}^T)\overline{Z}\\
       &+F^T(Z_0-\overline{Z}_0)+(F^T+\overline{F}^T)\overline{Z}_0 + Q (X-\overline{X}) + (Q+\overline{Q})\overline{X}\\
       &+ S^T v + \overline{S}^T \overline{v} + q + \overline{q}\Bigr)\,ds + Z\,dW + Z_0\,dW^0\\
\end{split}
\end{equation}
Now, on one hand we have:
\begin{equation}\label{eq::equationsixAppA}
    \begin{split}
        d(Y-\overline{Y}) &= \left(\overset{\cdot}{P}(X-\overline{X}) + P A (X - \overline{X}) + P B (v - \overline{v})\right)\,ds\\
                    &+P \left(\boldsymbol{C_0} + C (X-\overline{X}) + (C+\overline{C})\overline{X} + D(v-\overline{v})+(D+\overline{D})\overline{v}\right)\,dW\\
                    &+P \left(F(X-\overline{X})+G(v-\overline{v})\right)\,dW^0.
    \end{split}
\end{equation}
On the other hand, it holds that (see Equation \eqref{eq::equationTwoAppA}):
\begin{equation}\label{eq::equationsevenAppA}
\begin{split}
    d\overline{Y} &= (\overset{\cdot}{\phi} + \overset{\cdot}{\Pi})\,ds + \Pi d\overline{X}\\
             &=\left(\overset{\cdot}{\phi} + \overset{\cdot}{\Pi}\overline{X} + \boldsymbol{\Pi A_0} + \Pi(A+\overline{A})\overline{X}+\Pi(B+\overline{B})\overline{v}\right)\,ds\\
             &+\Pi\left(\boldsymbol{F_0} + (F+\overline{F})\overline{X}+(G+\overline{G})\overline{v}\right)\,dW^0
\end{split}
\end{equation}
Noting that $dY = d(Y-\overline{Y})+d\overline{Y}$, we compare the diffusion terms of the left and right hand side of this equation. We get:
\begin{equation}\label{eq::equationeightAppA}
    Z = P \left(\boldsymbol{C_0} + C (X-\overline{X}) + (C+\overline{C})\overline{X} + D(v-\overline{v})+(D+\overline{D})\overline{v}\right)
\end{equation}
\begin{equation}\label{eq::equationnineAppA}
    Z_0 =P \left(F(X-\overline{X})+G(v-\overline{v})\right) +\Pi\left(\boldsymbol{F_0} + (F+\overline{F})\overline{X}+(G+\overline{G})\overline{v}\right)
\end{equation}
which implies
\begin{equation}\label{eq::equationtenAppA}
    \overline{Z} = P \left(\boldsymbol{C_0}+(C+\overline{C})\overline{X} +(D+\overline{D})\overline{v}\right)
\end{equation}
\begin{equation}\label{eq::equationelevenAppA}
    Z-\overline{Z} = P \left(C (X-\overline{X}) + D(v-\overline{v})\right)
\end{equation}
\begin{equation}\label{eq::equationtwelveAppA}
    \overline{Z}_0 = \Pi\left(\boldsymbol{F_0}+(F+\overline{F})\overline{X}+(G+\overline{G})\overline{v}\right)
\end{equation}
and
\begin{equation}\label{eq::equationthirteenAppA}
    Z_0-\overline{Z}_0 = P \left(F(X-\overline{X})+G(v-\overline{v})\right)
\end{equation}
At this point, the coupling condition in \cite{graber2016linear}, Proposition 2.4, Equation (2.5) reads as\footnote{On the other hand, Equation (2.5) is affected by zero-order terms since it depends on $\overline{Z}$ and $\overline{Z}_0$.}:
\begin{equation}\label{eq::equationfourteenAppA}
    \begin{split}
        &B^T (Y-\overline{Y}) + (B^T+\overline{B}^T)\overline{Y} +D^T(Z-\overline{Z}) + (D+\overline{D}^T)\overline{Z}\\
        &+G^T(Z_0-\overline{Z}_0) + (G^T+\overline{G}^T)\overline{Z}_0 +R (v-\overline{v}) +(R+\overline{R})\overline{v} + S (X-\overline{X}) +(S+\overline{S})\overline{X}(s) +r +\overline{r}\\
        &=B^T P (X-\overline{X}) + (B^T+\overline{B}^T)\Pi\overline{X} + (B^T+\overline{B}^T)\phi + D^T P \left(C (X-\overline{X}) + D(v-\overline{v})\right)\\
        &+(D+\overline{D}^T)\left(\boldsymbol{C_0}+(C+\overline{C})\overline{X} +(D+\overline{D})\overline{v}\right) +G^TP \left(F(X-\overline{X})+G(v-\overline{v})\right)\\
        &+ (G^T+\overline{G}^T) \Pi\left(\boldsymbol{F_0} + (F+\overline{F})\overline{X}+(G+\overline{G})\overline{v}\right) \\
        &+R (v-\overline{v}) +(R+\overline{R})\overline{v} + S (X-\overline{X}) +(S+\overline{S})\overline{X} +r +\overline{r},\\
    \end{split}
\end{equation}
which can be rewritten in the following way
\begin{equation}\label{eq::equationfifteenAppA}
\Lambda_0 (X-\overline{X}) + \Lambda_1 \overline{X} + \Sigma_0 (v-\overline{v}) + \Sigma_1 \overline{v} + (B^T+\overline{B}^T)\phi + r + \overline{r} + \boldsymbol{(D+\overline{D}^T)C_0 + (G^T+\overline{G}^T)\Pi F_0}=0
\end{equation}
by setting
\begin{equation}\label{eq::equationsixteenAppA}
    \begin{split}
        \Lambda_0 &= B^T P + D^T P C + G^T P F + S;\\
        \Lambda_1 &= (B^T+\overline{B}^T)\Pi + (D+\overline{D}^T)P (C+\overline{C}) +  (G^T+\overline{G}^T) \Pi (F+\overline{F}) + (S+\overline{S}) \\
        \Sigma_0 &= D^T P D + G^T P G + R\\
        \Sigma_1 &= (D+\overline{D}^T)P (D+\overline{D}) + (G^T+\overline{G}^T) \Pi (G+\overline{G}) + R+\overline{R}\\
    \end{split}
\end{equation}
Taking the conditional expectation in Equation \eqref{eq::equationfifteenAppA}, assuming $\Sigma_1$ invertible, and making the term $\overline{v}$ explicit in Equation \eqref{eq::equationfifteenAppA}, we deduce 
\begin{equation}\label{eq::equationseventeenAppA}
   \overline{v}=- \Sigma_1^{-1}\left(\Lambda_1 \overline{X} + r + \overline{r} + (B^T+\overline{B}^T)\phi + \boldsymbol{(D+\overline{D}^T)C_0 + (G^T+\overline{G}^T)\Pi F_0}\right)
\end{equation}
Assuming $\Sigma_0$ is also invertible and observing that $v=v - \overline{v} + \overline{v}$ , we have:
\begin{equation}\label{eq::equationeighteenAppA}
\begin{split}
    v =&-\Sigma_0^{-1} \Lambda_0(X-\overline{X})\\
       &- \Sigma_1^{-1}\left(\Lambda_1 \overline{X} + r + \overline{r} + (B^T+\overline{B}^T)\phi + \boldsymbol{(D+\overline{D}^T)C_0 + (G^T+\overline{G}^T)\Pi F_0}\right)\\
\end{split}
\end{equation}
At this point, we compare the drift terms from \eqref{eq::equationFiveAppA} to those of \eqref{eq::equationsixAppA} and \eqref{eq::equationsevenAppA}. Using the relations \eqref{eq::equationtenAppA}, \eqref{eq::equationelevenAppA}, \eqref{eq::equationtwelveAppA}, \eqref{eq::equationthirteenAppA}, \eqref{eq::equationseventeenAppA}, and \eqref{eq::equationeighteenAppA} proved above. By noticing that 
$v-\overline{v}=-\Sigma_0^{-1}\Lambda_0(X-\overline{X})$, after some algebra, we deduce that $P$ and $\Pi$ should satisfy the following Riccati equations:
\begin{equation}\label{eq::RiccatiEquationforP}
\begin{cases}
\begin{split}
    &\overset{\cdot}{P} + P A + A^T P + C^T P C + F^T P F + Q - (P B +  C^T P D + F^T P G + S^T) \Sigma_0^{-1} \Lambda_0 = 0\\
    &\Lambda_0 = B^T P + D^T P C + G^T P F + S;\\
    &\Sigma_0 = D^T P D + G^T P G + R\\
    &P(T) = H.
\end{split}
\end{cases}
\end{equation}

\begin{equation}\label{eq::RiccatiEquationforPI}
   \begin{cases}
   \begin{split}
        &\overset{\cdot}{\Pi} + \Pi (A + \overline{A}) + (A^T + \overline{A}^T)\Pi + (C^T+\overline{C}^T)P(C+\overline{C})+ (F^T+\overline{F}^T)\Pi(F+\overline{F}) + (Q+\overline{Q})\\
        &-\left(\Pi (B+\overline{B}) +(C^T+\overline{C}^T)P(D+\overline{D}) + (F^T+\overline{F}^T)\Pi(G+\overline{G})+\overline{S}^T +  S^T +\right)\Sigma_1^{-1}\Lambda_1=0\\
        &\Lambda_1 = (B^T+\overline{B}^T)\Pi + (D+\overline{D}^T)P (C+\overline{C}) +  (G^T+\overline{G}^T) \Pi (F+\overline{F}) + (S+\overline{S}) \\
        &\Sigma_1 = (D+\overline{D}^T)P (D+\overline{D}) + (G^T+\overline{G}^T) \Pi (G+\overline{G}) + R+\overline{R}\\
        &\Pi(T) = H + \overline{H}
    \end{split} 
\end{cases}
\end{equation}
Once we have $P$ and $\Pi$ solution to Equation \eqref{eq::RiccatiEquationforP} and \eqref{eq::RiccatiEquationforPI}, we set:
\begin{equation}\label{eq::ODEforphi}
    \begin{split}
        &\overset{\cdot}{\phi} -\Bigl(\Pi (B+\overline{B}) + (C^T+\overline{C}^T)P(D+\overline{D}) + (F^T+\overline{F}^T)\Pi(G+\overline{G}) + S^T + \overline{S}^T\Bigr)\Sigma_1^{-1}\cdot\\
        &\Big(r+\overline{r}+\boldsymbol{(D+\overline{D}^T)C_0 + (G^T+\overline{G}^T)\Pi F_0}\\
        &+ q + \overline{q} + \boldsymbol{\Pi A_0-(C^T+\overline{C}^T)PC_0-(F^T+\overline{F}^T)\Pi F_0\Big)}\\
        &-\Bigl[\Bigl(\Pi (B+\overline{B}) + (C^T+\overline{C}^T)P(D+\overline{D}) + (F^T+\overline{F}^T)\Pi(G+\overline{G}) + S^T + \overline{S}^T\Bigr)\Sigma_1^{-1}(B^T+\overline{B}^T)\\
        &+ (A^T+\overline{A}^T)\Bigr]\phi=0
    \end{split}
\end{equation}
Finally, we obtain the optimal trajectory (using Equation \eqref{eq::equationeighteenAppA}), a formula for the process $Z$ (using Equations \eqref{eq::equationelevenAppA}, \eqref{eq::equationtwelveAppA} and \eqref{eq::equationeighteenAppA}), and a formula for the process $\overline{Z}_0$ (using Equations \eqref{eq::equationelevenAppA} and \eqref{eq::equationtwelveAppA}):
\begin{equation}\label{eq::optimaldynamics}
    \begin{split}
        dX &=\left(A (X-\overline{X}) + (A+\overline{A})\overline{X} + B(v-\overline{v}) + (B+\overline{B})\overline{v}\right)\,ds\\
           &+\left(C (X-\overline{X}) + (C+\overline{C})\overline{X} + D(v-\overline{v}) + (D+\overline{D})\overline{v}\right)\,dW\\ 
           &+\left(F (X-\overline{X}) + (F+\overline{F})\overline{X} + G(v-\overline{v}) + (G+\overline{G})\overline{v}\right)\,dW^0\\ 
    \end{split}
\end{equation}
  \begin{equation}\label{eq::processZ}
    \begin{split}
        Z   &=P\left(C (X-\overline{X}) - D\Sigma_0^{-1}\Lambda_0(X-\overline{X}) \right)\\
            &+P\Bigg(\boldsymbol{C_0} + (C+\overline{C})\overline{X}-(D+\overline{D})\Sigma_1^{-1}\Bigg(\Lambda_1 \overline{X} + r + \overline{r} + (B^T+\overline{B}^T)\phi\\
            &+ \boldsymbol{(D+\overline{D}^T)C_0 + (G^T+\overline{G}^T)\Pi F_0}\Bigg)\Bigg)
    \end{split}
\end{equation}
\begin{equation}\label{eq::processZ0}
\begin{split}
    Z_0 &=P \left(F - G\Sigma_0^{-1}\Lambda_0\right)(X-\overline{X}) \\
        &+\Pi\left(\boldsymbol{F_0} + (F+\overline{F}) - (G+\overline{G})\Sigma_1^{-1}\Lambda_1\right)\overline{X}\\
        &-\Pi(G+\overline{G})\Sigma_1^{-1}(r + \overline{r} +(B^T+\overline{B}^T)\phi + \boldsymbol{(D+\overline{D}^T)C_0 + (G^T+\overline{G}^T)\Pi F_0}).
\end{split}   
\end{equation}
Equations for the general case are easily obtained by using the summations where necessary (e.g., $C^T P C$ is replaced by $\sum_{j=1}^{d_1} C_j^T P C_j$).

\end{document}